\documentclass[12pt, reqno]{amsart}

\usepackage{lscape,amssymb, amsmath,verbatim,graphicx}
\usepackage{epsfig,subfigure,float}
\usepackage{graphicx}
\usepackage{threeparttable}
\usepackage{ragged2e}
\usepackage{epstopdf}
\usepackage{booktabs}
\usepackage{float}
\usepackage{multirow}
\usepackage{adjustbox}
\usepackage{lscape,lipsum}
\usepackage{url}

\newtheorem{theorem}{Theorem}[section]
\newtheorem{remark}{Remark}[section]
\newtheorem{lemma}{Lemma}[section]

\newtheorem{assumption}{Assumption}[section]

\numberwithin{theorem}{section}
\numberwithin{equation}{section}
\numberwithin{table}{section}
\numberwithin{figure}{section}

\newcommand\T{\top}
\newcommand\eps{\epsilon}

\newcommand{\cN}{{\mathcal N}}
\newcommand{\cP}{{\mathcal P}}
\newcommand{\cQ}{{\mathcal Q}}
\newcommand{\cM}{{\mathcal M}}
\newcommand\fa{{\frak a}}
\newcommand\fb{{\frak b}}
\newcommand\fc{{\frak c}}
\newcommand\fd{{\frak d}}

\def\lra{\longrightarrow}
\def\beq{\begin{equation}}
\def\eeq{\end{equation}}
\def\bals{\begin{align*}}
\def\eals{\end{align*}}
\def\bal{\begin{align}}
\def\eal{\end{align}}

\newcommand{\brho}{\boldsymbol \rho}

\newcommand{\bet}{\boldsymbol \eta}

\newcommand\bbe{\mbox{\boldmath${ \beta}$}}
\newcommand\bde{\mbox{\boldmath${ \delta}$}}

\newcommand\bA{{\bf A}}
\newcommand\bD{{\bf D}}
\newcommand\bE{{\bf E}}

\newcommand\bN{{\bf N}}

\newcommand\bX{{\bf X}}
\newcommand\bY{{\bf Y}}

\newcommand{\ba}{{\bf a}}
\newcommand{\bc}{{\bf c}}

\newcommand\bg{{\bf g}}

\newcommand{\bw}{{\bf w}}
\newcommand\bx{{\bf x}}
\newcommand{\bz}{{\bf z}}

\begin{document}

\title[Sequential monitoring]{Sequential monitoring of changes in housing prices}

\author {Lajos Horv\'ath}

\address{Lajos Horv\'ath, Department of Mathematics, University of Utah, Salt Lake City, UT 84112--0090 USA}

\author{Zhenya Liu}
\address{Zhenya Liu, China Financial Policy Research Center, School of Finance, Renmin University of China, Beijing, 100872 China and
CERGAM, Aix--Marseille University, 13090 Aix--en--Provence Cedex 02, France }
\author{Shanglin Lu}
\address{Shanglin Lu, School of Finance, Renmin University of China, Beijing, 100872 China }

\keywords{sequential change point detection, weak dependence, linear model, autoregressive model, real estate market\\
\hspace*{.5cm}{\it JEL classification.} C32, C58, R30}

\begin{abstract}
We propose a sequential monitoring scheme to find structural breaks in real estate markets. The changes in the real estate prices are modeled by
a combination of linear and autoregressive terms. The monitoring scheme is based on a detector and a suitably chosen boundary function. If the detector crosses the boundary function, a structural break is detected. We provide the asymptotics for the procedure under the stability null hypothesis and the stopping time under the change point alternative. Monte Carlo simulation is used  to show the size and the power of our method under several conditions.
We study the real estate markets in Boston, Los Angeles and at the national U.S.\ level.  We find structural breaks in the markets, and we segment the data into stationary segments. It is  observed that the autoregressive parameter is increasing but stays below 1.

\end{abstract}

\maketitle

\section{Introduction}\label{intro}
Housing has been the most substantial investment or cost for a large portion of the households so modeling changes in housing prices has received a considerable amount of attention in the literature. Following Case and Shiller (1989, 2003),  Piazzesi and Schneider (2009) and Zheng et al.\ (2016) we write the change in the $\log$ of the housing prices as a linear combination of macroeconomic fundamentals  and we also include a first--order autoregressive term of  the change in the $\log$ housing prices. One of the fundamental  questions is if the model stayed stable during the observation or it is segmented into several periods including stationary and nonstationary epochs. Himmelberg et al.\ (2005), Mayer (2011), Granziera and Kozicki (2015), Burnside et al.\ (2016) and Glaeser and Nathanson (2017) were interested in possible ``bubbles" in housing prices, i.e.\ a short explosive segment in the data.

The historical prices of the U.S.\ real estate markets have gone through several periods of booms, like the California housing boom of the 1880s, the Florida land boom of the 1920s and the peaks in the national real estate market in the 1980s and 2000s. Providing a suitable model for the dynamics of the U.S.\ housing market  has been an important theoretical question in the literature. The national wide booms of the 1980s and the 2000s show common as well as different features. Both booms started on the east coast and spread to the west. According to the S\&P CoreLogic Case--Shiller  Home Price Indices, Boston housing price increases peaked in early 2005, while Los Angeles real estate  price increases reached their maximum in 2006, as the national level price. On the other hand, while the 1980s boom can be explained by a general economic expansion, the source of the housing price increase in the 2000s is different. It has been explained by the ``amplification mechanism" of positive expectation of future housing price appreciation. Home buyers started to see real estate as an investment instrument. We refer to Case and Shiller (2003) and Shiller (2008) for more detailed reviews of the U.S.\ real estate market peaks. Our data example provides a sequential monitoring framework to see how this ``amplification mechanism" evolves in the 2000s.

In this paper we develop and study  a sequential monitoring scheme to detect changes in the parameters of a model which contains linear as well as autoregressive terms. The assumptions on the regressors and the errors are mild, and they are satisfied by nearly all linear as well as nonlinear time series
processes. Roughly speaking, they are well approximated with finitely dependent sequences. Under the null hypothesis the model describing the price changes is stable, i.e. it is a stationary process.
 Following Chu et al.\ (1996), the proposed monitoring is based on a detector and a boundary function. When the detector reaches the boundary function, a change is detected. The detector is based on the sum of residuals, but only the training sample is used to estimate some unknown parameters. The boundary function is chosen such that the probability of a false detection under the stability of the parameters null hypothesis is fixed. We also provide results for the consistency of the monitoring under various types of changes in the original model. In the sequential setup consistency means that we stop in finite time with probability one if a change occurred. We also provide several results on the distribution of the stopping time under the alternative. The limits can be normal or not normal depending on the type of the change and the size of the change. We focus on the autoregressive  parameter and after the change we can have a new stationary regime, random walk or explosive autoregressive process.

 The paper is organized as follows: in Section \ref{math} we formulate our model and the detection scheme. We also detail the conditions which are needed in the paper and obtain the limit distribution of the monitoring under the null hypothesis. Section \ref{altstop} contains the distributions of the stopping time introduced in Section \ref{math} under three types of  alternatives.   Detailed  proofs are given in Appendices  \ref{sec-pr-1} and \ref{sec-pr-2}. We study the empirical size and power of the sequential scheme in Section \ref{sec-mc}. Section \ref{sec-ch} provides in illustration for our method using data on three U.S.\ real estate markets. The conclusion of our research is in Section \ref{concl}.

\section{Mathematical model to sequentially detect  changes in real estate prices}\label{math}
In our model we assume that a training (historical) sample of size $M$ is available
\beq\label{e-1}
y_t=\bx_t^\T\bbe_0+\eps_t, \quad 1\leq t \leq M,
\eeq
where  $\bbe_0\in R^d$ and
\beq\label{xdef}
\bx_t=(x_{t,1}, x_{t,2}, \ldots , x_{t,d-1}, x_{t,d})^\T\in R^d\;\;\mbox{with}\;\;x_{t,1}=1\;\;\mbox{and}\;\;x_{t,d}=y_{t-1}.
\eeq
The model in \eqref{e-1} combines linear and autoregressive models. It is a linear model in the first $d-1$ coordinates and autoregressive in the last coordinate of $\bx_t$. After the training sample further observations are obtained, $y_{M+s}, s=1,2, \ldots$ and under the null hypothesis
\begin{align}\label{null}
H_0:\left\{
\begin{array}{ll}y_{M+s}=\bx_{M+s}^\T\bbe_0+\eps_{M+s}, \quad 1\leq s <\infty
\vspace{.1cm}\\
\mbox{and}
\vspace{.1cm}\\
\{\eps_t, 1\leq t <\infty\}\;\mbox{is a stationary sequence.}
\end{array}
\right.
\end{align}
This means that the structure of the observations $y_t$ is the same during the  training sample and the  observations collected after the training sample
obey the same model. Under the alternative the structure of the observations changes at an unknown time $M+s^*$:
\begin{align}\label{alt}
H_A:\left\{
\begin{array}{ll}y_{M+s}=\bx_{M+s}^\T\bbe_0+\eps_{M+s}, \quad 1\leq s \leq s^*,
\vspace{.1cm}\\
y_{M+s}=\bx_{M+s}^\T\bde_M+\eps_{M+s}, \quad s^*+1\leq s <\infty\quad \mbox{with} \quad
\bbe_0\neq\bde_M.
\end{array}
\right.
\end{align}
 The first  monitoring scheme to find changes in the regression parameter was introduced  by Chu et al.\ (1996) and it has become the starting point of substantial research.  Zeileis et al.\ (2005) and Aue et al.\ (2014) studied monitoring schemes in linear models with dependent errors. Kirch (2007, 2008) and Hu\v{s}kov\'a and Kirch (2012) provided resampling methods to find critical values for sequential monitoring.
Hl\'avka et al.\ (2012) investigated the sequential detection of changes   of the parameter  in autoregressive models, i.e.\ no regression terms are included in their theory. Homm and Breiting (2012) compared several methods to find bubbles in stock markets, detecting a change in an autoregressive process to an explosive one.  Horv\'ath et al.\ (2019+) showed that sequential methods will detect changes when the observations change from stationarity to mild  non--stationarity.  \\

The least square   estimator for $\bbe_0$ is given by
$$
\hat{\bbe}_M=({\bX}_M^\T{\bX}_M)^{-1}{\bX}_M^\T\bY_M,
$$
where $\bY_M=(y_1, y_2, \ldots, y_M)^\T$ and
\begin{displaymath}
{\bX}_{M}=\left(\begin{array}{ll}{\bx}_1^\T\\
{\bx}_2^\T\\
\hspace{.1cm}\vdots\\
{\bx}_M^\T
\end{array}
\right).
\end{displaymath}
Following Chu et al.\ (1996) we choose a detector $\Gamma(M,s)$, a boundary function $g(M,s)$, and define the stopping time
\begin{align*}
\tau_M=\left\{
\begin{array}{ll}
\inf\{s\geq 1: \Gamma(M,s)> g(M,s)\},
\vspace{.2cm}\\
\infty, \;\;\mbox{if}\; \Gamma(M,s)\leq g(M,s)\;\mbox{for all}\;s\geq 1.
\end{array}
\right.
\end{align*}

If $\tau_M<\infty$, we stop at time $\tau_M$ and we say that the null hypothesis is rejected. We choose the detector $\Gamma(M,s)$ and the boundary $g(M,s)$ such that
\beq\label{nual}
\lim_{M\to \infty}\left\{ \tau_M<\infty  \right\}=\alpha\;\;\;\mbox{under}\;\;H_0,
\eeq
where $0<\alpha<1$ is prescribed number and
\beq\label{nonual}
\lim_{M\to \infty}\left\{ \tau_M<\infty  \right\}=1\;\;\;\mbox{under}\;\;H_A.
\eeq
According to \eqref{nual}, the probability of   stopping  the procedure and rejecting $H_0$, when $H_0$, is $\alpha$.  We stop in finite time under the alternative. The definition of the detector follows Chu et al.\ (1996) and Horv\'ath et al. (2004).

The residuals of the model are defined as
\beq\label{resdef}
\hat{\eps}_t=y_t-{\bx}_t^\T\hat{\bbe}_M,\;\;1\leq t <\infty,
\eeq
i.e.\ in the definition of the residuals we also use $\hat{\bbe}_M$ even after the training period. The $\eps_t$'s  are stationary in the training sample under the null as well as under the alternative. Our detector is
$$
\displaystyle \Gamma(M,s)=\frac{1}{\hat{\sigma}_M}\left|\sum_{u=M+1}^{M+s} \hat{\eps}_u    \right|,\;\;1\leq s <\infty,
$$
where
$$
\hat{\sigma}_M^2=\frac{1}{M-d}\sum_{t=1}^M\hat{\eps}_t^2.
$$
We use the boundary function
\beq\label{bou-1}
g(M,s)=cM^{1/2}\left(1+\frac{s}{M}\right)\left( \frac{s}{M+s} \right)^\gamma,
\eeq
where $c=c(\gamma, \alpha)$ is chosen such that \eqref{nual} holds under the null hypothesis and
\beq\label{gadef}
0\leq \gamma<1/2.
\eeq
We discuss the choice of $\gamma$ in Section \ref{sec-mc}.  Following Brown et al.\ (1975), Horv\'ath et al.\ (2004) also used recursive residuals to define the detector in case of linear regression ($\beta_{0,d}=0$ under the null and the alternative).   Homm and Breiting (2012) applied fluctuation detectors when they wanted to test if a random walk changes to an explosive autoregression. They did not allow regression terms.

Next we discuss some conditions which will be needed to find $c=c(\gamma, \alpha)$ for our boundary function such that \eqref{nual} holds. Let
$$
\bz_t=(x_{t,2}, \ldots, x_{t,d-1}, \eps_t)^\T.
$$
The Euclidean norm of vectors and matrices is denoted by $\|\cdot\|$.

\begin{assumption}\label{bern}
$
\bz_t=\bg(\eta_t, \eta_{t-1}, \eta_{t-2}, \dots),
$
where
  $\bg(\cdots)$ is a nonrandom functional defined on ${\mathcal S}^{\infty}$ with values in  $R^{d-1}$ and ${\mathcal S}$ is a
 measurable space. Also, $\eta_t=\eta_t(s, \omega)$ is jointly measurable in $(s,\omega), -\infty<t<\infty$ and $\eta_t,\;-\infty<t<\infty$\; are independent and identically distributed random variables in ${\mathcal S}$.  The sequences  $\bz_t,
 -\infty<t<\infty$ can be approximated with $m$--dependent sequences $\bz_{t,m}$
 in the sense that with some $\kappa_1>4$, $\kappa_2>2$ and $c>0$, $E\|\bz_t\|^{\kappa_1}<\infty$,
 \beq\label{ber-2}
 \left(E\|\bz_t-\bz_{t,m}\|^{\kappa_1}\right)^{1/\kappa_1}\leq c m^{-\kappa_2}
 \eeq
 where
 $\bz_{t,m}=\bg(\eta_t, \eta_{t-1}, \eta_{t-2}, \ldots, \eta_{t-m+1}, \bet^*_{t,m}),$
  $\bet^*_{t,m}=(\eta^*_{t,m,t-m}, \eta^*_{t,m,t-m-1}, \eta^*_{t,m,t-m-2},\ldots)$ and the $\eta^*_{t,m,n}$'s are independent copies of $\eta_0$, independent
   of $\{\eta_t, -\infty<t<\infty\}.$
 \end{assumption}

 Assumption \ref{bern} appeared first in Ibragimov (1959, 1962)  in the proof of the central limit theorem for dependent variables. Billingsley (1968) also utilized $m$--decomposability. Nearly all time series, including linear and several nonlinear processes satisfy Assumption \ref{bern} (cf.\ H\"ormann and Kokoszka, 2010 and Aue et al.,\ 2014).

 \begin{assumption}\label{bernmean} $E\eps_t=0, 0<E\eps_t^2=\sigma^2<\infty$, $E\eps_t\eps_s=0, -\infty<t\neq s<\infty$ and $Ex_{s,\ell}\eps_t=0, \;2\leq \ell\leq d-1$ for all $-\infty<t,s   <\infty.  $
 \end{assumption}
 Assumption \ref{bernmean} means that $\{x_{t, \ell}, -\infty<t<\infty, 2\leq \ell\leq d-1\}$ and $\{\eps_t, -\infty<t<\infty\}$ are uncorrelated sequences.
 Clearly, if the $\eps_t$'s are  independent random variables, Assumption \ref{bernmean} holds but it is also satisfied by ARCH/GARCH type volatility sequences and orthogonal martingales.
We show in Lemma \ref{hatga} that
\beq\label{adef}
\frac{1}{M}{\bX}_M^\T{\bX}_M\stackrel{P}{\to}\bA.
\eeq
\begin{assumption}\label{a-adef} $\bA$ is nonsingular.
\end{assumption}

The asymptotic normality of $\hat{\bbe}_M$ has been established in case of independent and identically distributed $\eps_s$'s (cf.\ Zeckenhauser and Thompson,  1970).   These   results are extended by Wu (2007), Zhu (2013) and Caron (2019) to a large class of estimators for time series errors.

\begin{theorem}\label{th-null} If $H_0$ and Assumptions \ref{bern}--\ref{a-adef} hold, then we have that
\begin{align*}
\lim_{M\to\infty}P\left\{ \frac{\Gamma(M,s)}{g(M,s)}\leq 1\;\;\mbox{for all}\;\;s\geq 1     \right\}=P\left\{\sup_{0 < u \leq 1}\frac{|W(u)|}{u^\gamma}\leq c\right\},
\end{align*}
where $\{W(u), u\geq 0\}$ denotes a Wiener process (standard Brownian motion).
\end{theorem}
\medskip

We note that $\gamma=1/2$ is not allowed in Theorem \ref{th-null}  since in this case the limit distribution would be infinity. Horv\'ath et al.\ (2007) studied the ``square--root--boundary" case, i.e.\ when $\gamma=1/2$,  and they obtained a Darling--Erd\H{o}s type extreme value result for the limit distribution of the stopping time under the no change null hypothesis in linear regression. Chu et al.\ (1996) obtained an upper bound for the probability of false stopping under the null hypothesis (cf.\ Homm and Breitung, 2012). \\

The stopping time $\tau_M$ is an open ended since if there is no change we never stop collecting further  observations. In some applications we might want to stop at time $M+N$, i.e.\ only $N$ observations are collected after the training period.  Let
\begin{displaymath}
\bar{\tau}_M=\left\{
\begin{array}{ll}
\inf \{s: 1\leq s \leq N, \Gamma(M,s)>g(M,s)\},
\vspace{.3cm}\\
N+1,\; \mbox{if}\;\Gamma(M,s)\leq g(M,s)\;\;\mbox{for all}\;\;1\leq s \leq N
\end{array}
\right.
\end{displaymath}
denote the closed end version of $\tau_M$. Let $N=N(M)$ and define
$$
c_*=\lim_{M\to \infty}\frac{N}{M}
$$
\begin{remark}\label{close}{\rm Under the conditions of Theorem \ref{th-null} are satisfied, then we have that
$$
\lim_{M\to\infty}P\left\{ \frac{\Gamma(M,s)}{g(M,s)}\leq 1\;\;\mbox{for all}\;\;1\leq s \leq N     \right\}=P\left\{\sup_{0 < u \leq c_*/(1+c_*)}\frac{|W(u)|}{u^\gamma}\leq c\right\}
$$
for all $0<c_*<\infty$.
}
\end{remark}

Selected critical values for the limit distributions in Theorem \ref{th-null} and Remark \ref{close} can be found, for example, in Horv\'ath et al.\ (2004).
\medskip

\section{Asymptotic distribution of the stopping time under the alternative}\label{altstop}
In this section we investigate the properties of the sequential detection rule when the regression is not stable. Our procedure is tailored for early changes, i.e.\ $s^*$ is small, so we assume in this section that the changes occur early.  We concentrate on the autoregressive parameter $\beta_{0,d}$. We consider the cases (i) the observations stay stationary after the change, (ii) they change to a ``unit root" sequence and (iii) explosive autoregression after the change.  \\

First we assume that the regression parameter at time $M+s^*$ changes from $\bbe_0$ to $\bde=\bde_M=(\delta_{M,1}, \delta_{M,2}, \ldots, \delta_{M,d})^\T$
satisfying
\begin{assumption}\label{destat}
$\lim_{M\to\infty}\delta_{M,i}=\bar{\delta}_{i},\;\;1\leq i\leq d\;\;\mbox{and}\;\;|\bar{\delta}_d|<1.$
\end{assumption}
So for any fixed $M$, the sequence  changes from a stationary segment to an other stationary one. We allow that $\bar{\delta}_i=\beta_{0, i}$, i.e.\ the difference between the regression parameters can be small. We measure the size of change with
$$
\Delta=\Delta_M=\bc_A^\T(\bbe_0-\bde_M),
$$
where $\bc_A=(1, Ex_{0,2},\ldots ,Ex_{0,d-1}, Ey_A)^\T$ with
$$
y_A=\sum_{\ell=0}^\infty \bar{\delta}_{d}^\ell\left(\bw_{-\ell}^\T\bar{\delta}+\eps_{-\ell}\right),
$$
 \beq\label{bwdef}
 \bw_t=(1, x_{t,2}, \ldots, x_{t,d-1})^\T
\eeq
and $\bar{\bde}=(\bar{\delta}_{1}, \bar{\delta}_{2},\ldots , \bar{\delta}_{d-1})^\T$. Under the alternative $y_t$ converges in distribution to $y_A$. The assumption says that the size of the change cannot be too small:
\begin{assumption}\label{ta-2}
$M^{1/2}|\Delta_M|\to \infty.$
\end{assumption}

Analogue of Assumption \ref{ta-2} first appeared in retrospective change point detection in Picard (1985) and D\"umbgen (1991) when the time of change in the mean was estimated.

\begin{theorem}\label{th-alt-0} If Assumptions \ref{bern}--\ref{ta-2} hold, then we have that
$$
\lim_{C\to\infty}\liminf_{M\to\infty}P\left\{  \tau_M\leq C\left(\frac{M^{1/2-\gamma}}{|\Delta_M|}  \right)^{1/(1-\gamma)} \right\}=1.
$$
\end{theorem}

Next we show that the upper bound for $\tau_M$ in Theorem \ref{th-alt-0} is the best possible when we get the asymptotic normality of $\tau_M$.
 Let
$$
a_M=\left( \frac{c\sigma M^{1/2-\gamma}}{|\Delta_M|}  \right)^{1/(1-\gamma)}
$$
and
$$
b_M=\frac{\sigma a_M^{1/2}}{(1-\gamma)|\Delta_M|}.
$$
\begin{theorem}\label{alt-stat} If Assumptions \ref{bern}--\ref{ta-2} hold,
and
\beq\label{ta-3}
s^*=O(M^\theta)\;\;\;\mbox{with some}\;\;\;0\leq \theta<\left(\frac{1-2\gamma}{2(1-\gamma)}\right)^2,
\eeq
then we have that
$$
\frac{\tau_M-a_M}{b_M}\;\stackrel{{\mathcal D}}{\to}\;N,
$$
where $N$ is a standard normal random variable.
\end{theorem}

Aue and Horv\'ath (2004) proved Theorem \ref{th-alt-0} when the mean can change under the alternative. Their result was extended to linear regression by Horv\'ath et al.\ (2007).\\

Next we consider the case when $y_t$ changes to a random walk at time $M+s^*$:
\begin{assumption}\label{unir} \;$\delta_{M,d}=\bar{\delta}_d=1$
\end{assumption}
and the other parameters in the regression also might change
\begin{assumption}\label{destat-1}
$\lim_{M\to\infty}\delta_{M,i}=\bar{\delta}_{i},\;\;1\leq i\leq d-1.$
\end{assumption}
To describe the size of change we introduce
$$
\fa_1=E\bw_0^\T(\bar{\bde}-\bar{\bbe}_0)\;\;\;\mbox{and}\;\;\;\fb_1^2=\sigma^2+\sum_{s=-\infty}^\infty \mbox{\rm cov}(\bw_0^\T\bar{\bde}, \bw_s^\T\bar{\bde}),
$$
where $\bar{\bbe}_0=(\beta_{0,1}, \beta_{0,2}, \ldots, \beta_{0,d-1})^\T$.
\begin{theorem}\label{th-alt-2} If Assumptions \ref{bern}--\ref{a-adef}, \ref{unir}, \ref{destat-1} hold,
\beq\label{th-alt-20}
{s^*}{\displaystyle M^{-(1-2\gamma)/(3-2\gamma)}}\;\;\to \;\;0,
\eeq
and
\beq\label{th-alt-21}
\fa_1M^{(1-2\gamma)/(6-4\gamma)}\to \bar{\fa}_1,\;\;0\leq \bar{\fa}_1<\infty,
\eeq
then we have
\begin{align*}\label{th-alt-22}
\lim_{M\to \infty}P&\left\{\tau_M\leq x M^{(1-2\gamma)/(3-2\gamma)}   \right\}\\
&=1-P\left\{ \max_{0<s\leq x}\frac{1-\beta_{0,d}}{s^\gamma}\left|\fb_1\int_0^s W(u)du+\bar{\fa}_1s^2/2     \right|\leq c\sigma\right\}
\end{align*}
where $\{W(u), u\geq 0\}$ is a Wiener process.
\end{theorem}

\begin{remark}\label{rem-1} {\rm If $\bar{\bde}=\bar{\bbe}_0$, i.e.\ only the autoregressive  parameter changes, then $\bar{\fa}_1=0$. In this case
$$
\max_{0<s\leq x}\frac{(1-\beta_{0,d})}{s^\gamma}\left|\fb_1\int_0^s W(u)du    \right|\stackrel{{\mathcal D}}{=}
x^{3/2-\gamma}\max_{0<s\leq 1}\frac{(1-\beta_{0,d})|\fb_1|}{s^\gamma}\left|\int_0^s W(u)du    \right|
$$
for all $x>0$.
}
\end{remark}

Let
$$
c_M=\left(\frac{2c\sigma}{1-\beta_{0,d}}\right)^{1/(2-\gamma)}\fa_1^{1/(2-\gamma)}M^{(1-2\gamma)/(4-2\gamma)}
$$
and
$$
d_M=d_0\fa_1^{-(7-4\gamma)/(4-2\gamma)}M^{(1-2\gamma)(3-2\gamma)/(8-4\gamma)}c_M^{-(\gamma-1)}
$$
with
$$
d_0=\frac{1}{2-\gamma}\frac{\fb_1}{\sqrt{3}}\left(\frac{\sigma}{1-\beta_{0,d}}\right)^{(3-2\gamma)/(4-2\gamma)}c^{(3-\gamma)/(4-2\gamma)}2^{(7-4\gamma)/(4-2\gamma)}.
$$

In Theorem \ref{th-alt-2} and Remark \ref{rem-1} the change to a random walk in the autoregressive part dominates the limit distribution. Hence $y_t$ is a partial sum after $M+s^*$ and the limit is determined by the sums of partial sum processes. In the next result the change in the regression parameters are larger than in Theorem \ref{th-alt-2} and while $y_t$ is still a random walk after the change, we have the same limit as in Theorem \ref{alt-stat}.

\begin{theorem}\label{th-alt-3} If Assumptions \ref{bern}--\ref{a-adef}, \ref{unir}, \ref{destat-1} hold,
\beq\label{th-alt-31}
s^*M^{-(3-2\gamma)/(4-2\gamma)}\to 0,
\eeq
\beq\label{th-alt-23}
\limsup_{M\to\infty}|\fa_1|<\infty\;\;\; \quad \mbox{and}\quad |\fa_1|M^{(1-2\gamma)/(6-4\gamma)}\to \infty,
\eeq
then we have for all $x$ that
$$
\frac{\tau_M-c_M}{d_M}\;\;\stackrel{{\mathcal D}}{\to}\;\;N,
$$
where $N$ is a standard normal random variable.
\end{theorem}

Next we consider the case when the sequence $y_t$ turns explosive after the change at time $M+s^*$. Now we replace Assumption \ref{unir} with
\begin{assumption}\label{expo-1}$\;\;\delta_{M,d}=\bar{\delta}_d\;\;\;\mbox{and}\;\; |\bar{\delta}_d|>1.$
\end{assumption}
Let
\beq\label{zdef}
Z_{M+s^*}=y_{M+s^*}+\sum_{z=1}^\infty \bar{\delta}_d^{-z } (\bw_{M+s^*+z}^\T(\bar{\bde}-\bar{\bbe}_0)+\eps_{M+s^*+z})
\eeq
and define
$$
F(x)=P\{ Z_{M+s^*}\leq x\}.
$$
It follows from Assumption \ref{bern} that the infinite series defining $Z_{M+s^*}$ is finite with probability 1.

\begin{theorem}\label{expo} If Assumptions \ref{bern}--\ref{a-adef}, \ref{destat-1}, \ref{expo-1} and
\beq\label{expos}
s^*/\log M\to 0
\eeq
hold, then we have for all $x$ that
\begin{align*}
\lim_{M\to \infty}P&\left\{\tau_M\leq s^*+x+((1/2-\gamma)\log |\bar{\delta}_d|)\log M+(\gamma\log |\bar{\delta}_d|)\log \log M \right\}\\
&=1-F\left(\displaystyle|\bar{\delta}_d|^{\displaystyle -x}  \frac{c\sigma |\bar{\delta}_d-1|}{|\bar{\delta}_d-\beta_{0,d}|} \right).
\end{align*}
\end{theorem}
\medskip

Assumption \ref{expo} is often used  to find ``bubbles" in financial data. Phillips and Yu (2011) and Phillips et al.\ (2014, 2015a,b) estimated the autoregressive parameter in an AR(1) sequence and if the estimate is significantly larger than 1, a ``bubble" is detected. For a survey on ``bubble" detection we refer to Homm and Breiting (2012).

\section{Monte Carlo simulations}\label{sec-mc}
In this section we investigate the performance of our limit theorems in case of a finite training sample of size $M$. Preliminary results showed that the boundary $g(M,s)$ of \eqref{bou-1}  over rejects when $H_0$ holds. The false positive rates were improved when the boundary function
\beq\label{modbou}
\hat{g}(M,s)=c\left(1+\frac{(1+\gamma)\hat{\sigma}_M}{M^{1/2}}\right)M^{1/2}\left(1+\frac{s}{M}\right)\left(\frac{s}{s+M}\right)^\gamma,
\eeq
where $c=c(\gamma, \alpha)$. The values of $c(\gamma, \alpha)$ are defined from the equation
\beq\label{wewi}
P\left\{\sup_{0\leq u \leq 1}\frac{|W(u)|}{u^\gamma}>c(\gamma, \alpha)\right\}=\alpha.
\eeq
Since the correction term
$$
1+\frac{(1+\gamma)\hat{\sigma}_M}{M^{1/2}}\;\stackrel{P}{\lra}\;1
$$
under the conditions of Theorem \ref{th-null}
\beq\label{mod-nu}
\lim_{M\to \infty}P\left\{  \frac{\Gamma(M,s)}{\hat{g}(M,s)} >1\;\;\mbox{for some }\;s\geq 1  \right\}=\alpha
\eeq
and under the alternatives in Theorems \ref{th-alt-0}--\ref{expo}
\beq\label{mod-alt}
\lim_{M\to \infty}P\left\{  \frac{\Gamma(M,s)}{\hat{g}(M,s)} >1\;\;\mbox{for some }\;s\geq   1\right\}=1.
\eeq
\begin{table}[H]
\caption{Selected critical values $c(\gamma, \alpha)$ of \eqref{wewi}.}
		\centering	
		\resizebox{12cm}{!}{	
			\begin{threeparttable}
				\begin{tabular}{lllllllllllllllll}
					\toprule
					$\gamma$ & & & $\alpha$ & & &  & & &   & & &   & & &   \\
					\cmidrule{4-16}   & & & 0.010 & & & 0.025 & & & 0.050 & & & 0.100  & & & 0.250  \\
					\midrule
					0.00   & & & 2.7912 & & & 2.4948 & & & 2.2365 & & & 1.9497 & & & 1.5213 \\
					0.15   & & & 2.8516 & & &  2.5475 & & & 2.2996 & & & 2.0273 & & & 1.6126  \\
					0.25  & & & 2.9445 & & &  2.6396 & & & 2.3860 & & & 2.1060 & & & 1.7039  \\
					0.35  & & & 3.0475 & & &  2.7394 & & & 2.5050 & & & 2.2433 & & & 1.8467  \\
					0.45  & & & 3.3015 & & &  3.0144 & & & 2.7992 & & & 2.5437 & & & 2.1729  \\
					0.49  & & & 3.5705 & & &  3.2944 & & & 3.0722 & & & 2.8259 & & & 2.4487  \\
					\bottomrule
				\end{tabular}
				\label{tab-1}		
		\end{threeparttable}}		
	\end{table}
\begin{figure}[H]
		\caption{The rate of false detections, in percentages,  in case of DGP(i) at times $M, 2M, \ldots ,10M$ after the training sample.}	
	\hspace*{-1.4cm}	\includegraphics[width=7.45in]{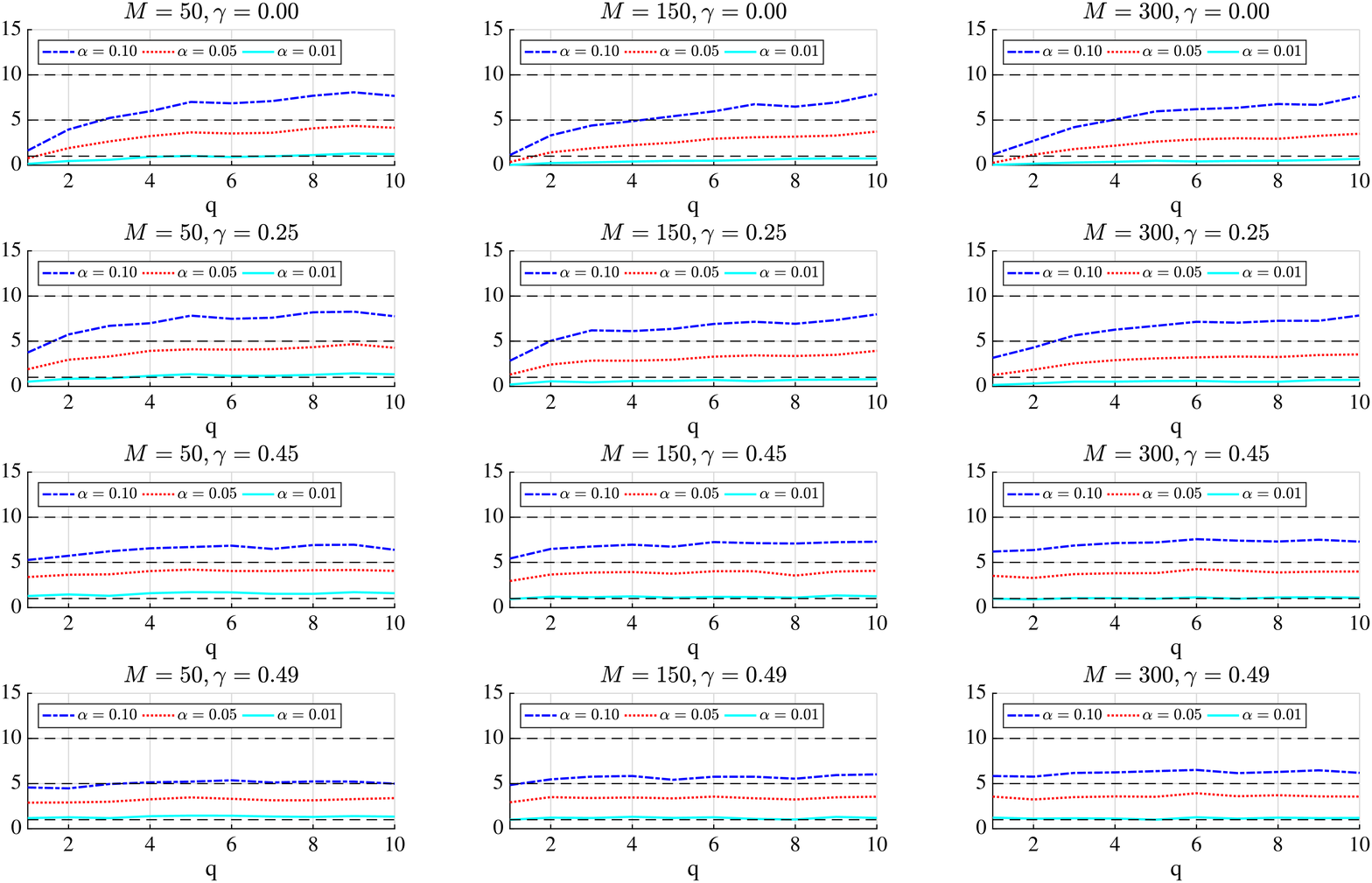}	
		\label{sim-null-f1}
	\end{figure}
The critical values of \eqref{wewi} were reported in Horv\'ath et al.\ (2004) and for convenience we provide them in Table \ref{tab-1}. The results in Table \ref{tab-1} are based on $50,000$  repetitions of   $\sup_{0 \leq u \leq 1} |W(u)|/u^{\gamma}$. The Wiener process was approximated on a grid of 10,000 equi--spaced points in [0,1]. We chose $d=6$ in our simulations and under the null hypothesis $\bar{\bbe}_0=(.02, .20, .25, .15, -.20)^\T $ and the autoregressive parameter was $\beta_{0,6}=.25.$ Our procedure is open ended but, of course, during the simulations we stopped the testing after additional  $M, 2M,\ldots, 10M$ observations were collected in the detection period. In Figures \ref{sim-null-f1}--\ref{sim-null-f4} we exhibit the number of false alarms before time $iM, 1\leq i \leq 10$. We used the boundary function $\hat{g}(M,s)$ of \eqref{modbou} with $\gamma=0, .25, .45, .49$ and the  size of the training sample was $M=50, 150$ and $300$. The results are based on 10,000 repetitions. Under the null hypothesis we considered the following data generating processes:\\
DGP(i)\;
\beq\label{ar-1/1}
x_{t,k}=\rho_k x_{t-1,k}+\eta_{t,k}, \quad \;\;2\leq k\leq 5, -\infty<t<\infty,
\eeq
 where the $\eta_{t,k}$'s are independent, identically distributed standard normal random variables. Also, the $\epsilon_t$ forms a GARCH(1,1) process defined by
 \beq\label{garch-1}
 \epsilon_t=\sigma_{t,\epsilon}h_{t,\epsilon}\quad \sigma^2_{t,\epsilon}=.2+.3\epsilon_{t-1}^2+.3\sigma_{t-1,\epsilon}^2,\;\;-\infty<t<\infty,
 \eeq
where the $h_{t,\epsilon}$'s  are independent, standard normal random variables, independent of $\{\eta_{t,k}, -\infty<t<\infty, 2\leq k \leq 5\}$. We used $\brho=(\rho_2, \rho_3, \rho_4, \rho_{5})^\T=(.15, .20, .10, .30)^\T$ to get the values in Figure \ref{sim-null-f1}.\\
DGP(ii)\; $x_{t,k}, 2\leq k \leq 5, -\infty<t<\infty$ satisfy \eqref{ar-1/1} but now $\eta_{t,2}=\ldots =\eta_{t,5}, -\infty<t<\infty$ are independent and identically distributed standard normal random  variables, independent of $\{\epsilon_t, -\infty<t<\infty\}$. The variables $\{\epsilon_t, -\infty<t<\infty\}$ are independent standard normal random variables.\\
\begin{figure}[H]
		\caption{The rate of false detections, in percentages,  in case of DGP(ii) at times $M, 2M, \ldots ,10M$ after the training sample.}
\hspace*{-1.4cm}		\includegraphics[width=7.45in]{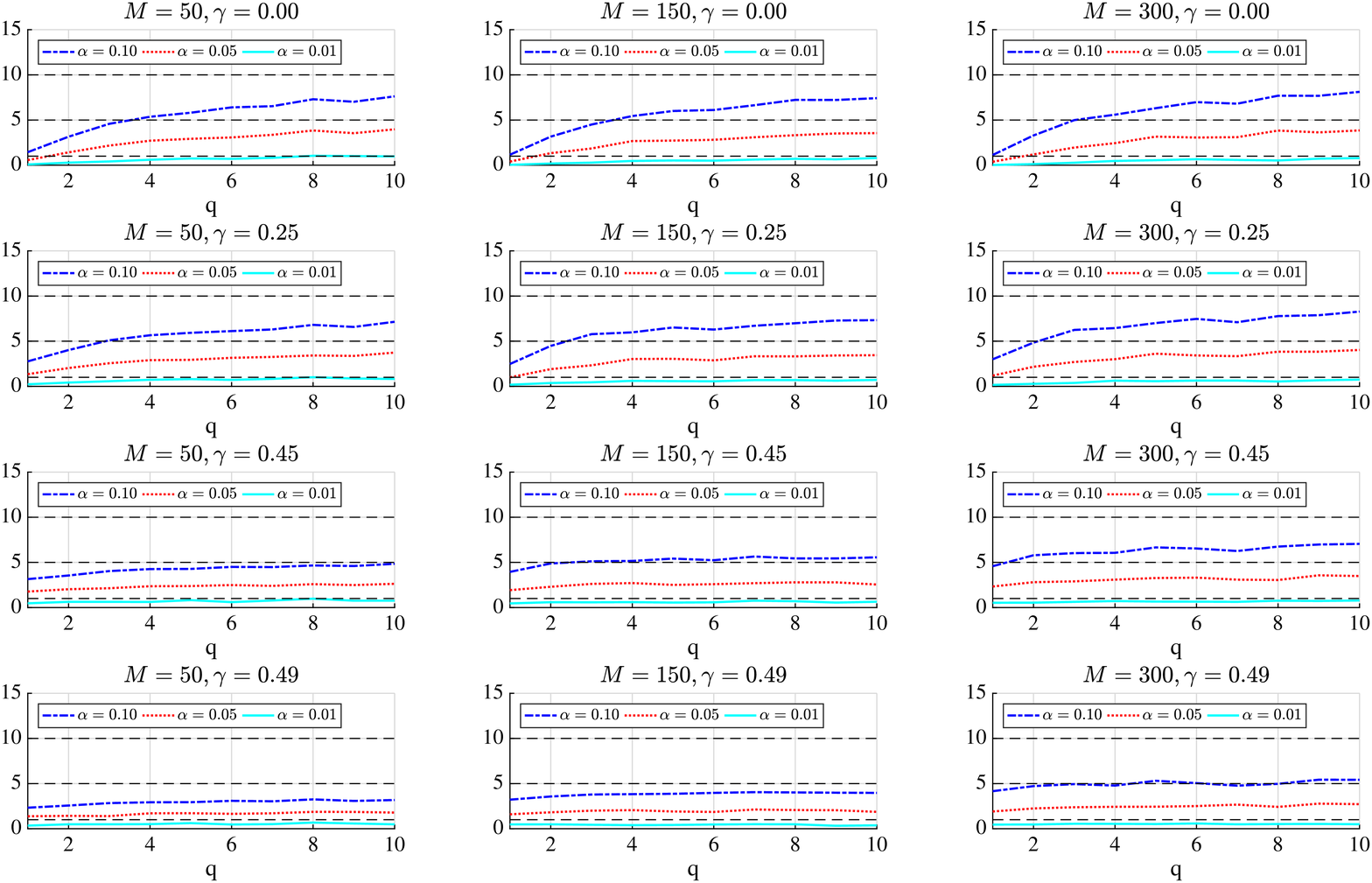}		
		\label{sim-null-f2}
	\end{figure}
DGP(iii) Now in addition to \eqref{garch-1}, the explanatory sequences are also given by GARCH(1,1) processes
\beq\label{garch-2}
 x_{t,k}  = \sigma_{t,k} h_{t,k}\mbox{, }\sigma_{t, k}^2  =  \omega_k +  \phi_k x_{t-1,k}^2 + \psi_k \sigma_{t-1, k}^2, \quad 2 \leq k \leq 5, -\infty<t<\infty,
\eeq
where the innovations $\{h_{t,k}, -\infty<t<\infty, 2\leq k\leq 5\}$ are standard normal random variables, independent of $\{h_{t,\epsilon}, -\infty<t<\infty\}$ of DGP(i). We used $(\omega_2, \ldots ,\omega_5)=(.3, .5, .4, .6)$,  $(\phi_2,\ldots ,\phi_5)=(.5, .3, .2, .6)$ and
$(\psi_2, \ldots ,\psi_5)=(.2, .3, .6, .2)$. \\
DGP(iv) The explanatory variables satisfy \eqref{garch-2} but now $h_{t,2}=h_{t,3}=h_{t,4} =h_{t,5}$ which are independent and identically distrubuted standard normal random variables. The variables $\{\epsilon_t, -\infty<t<\infty\}$  are independent, standard normal random, independent of  $\{h_{t,k}, -\infty<t<\infty, 2\leq k\leq 5\}$.\\

In our  Monte Carlo simulations the variables $\{(x_{t,2}, \ldots, x_{t,5}), -\infty<t<\infty\}$ and $\{\epsilon_t, -\infty<t<\infty\}$ are independent. In case of DGP(i) and (iii), the coordinates of $(x_{t,2}, \ldots, x_{t,5}) $ are independent while strongly dependent under DGP(ii) and (iv). The simulation results in Figures \ref{sim-null-f1}--\ref{sim-null-f4} show good performance, the empirical rate of false detections is at the described level. The structure of the ${\bw}_t$'s has little effect on false detection.\\

\begin{figure}[H]
		\caption{The rate of false detections, in percentages,  in case of DGP(iii) at times $M, 2M, \ldots ,10M$ after the training sample.}		
	\hspace*{-1.4cm}	\includegraphics[width=7.45in]{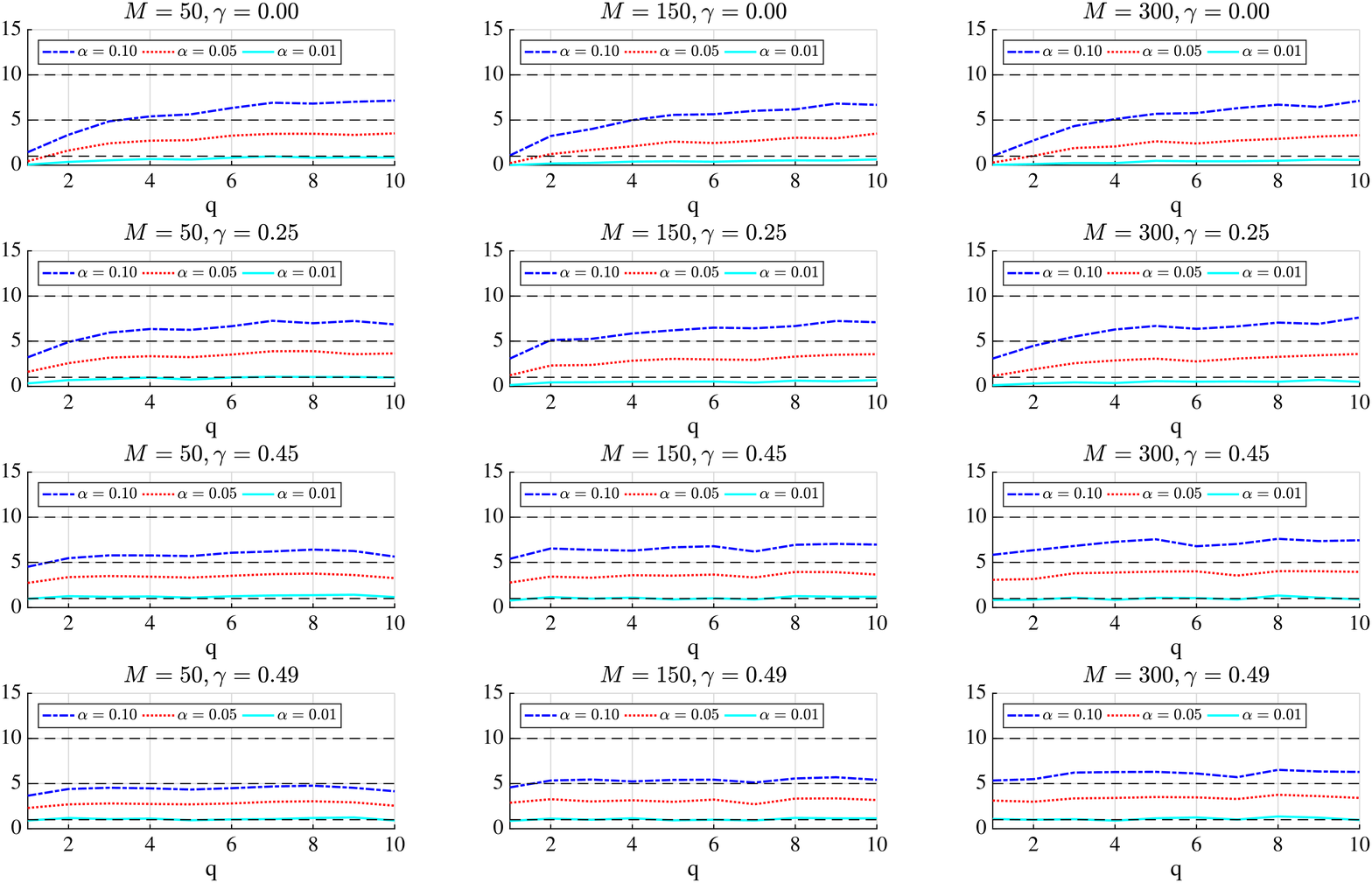}	
		\label{sim-null-f3}	
	\end{figure}

	\begin{figure}[H]
		\caption{The rate of false detections, in percentages,  in case of DGP(iv) at times $M, 2M, \ldots ,10M$ after the training sample.}	
	\hspace*{-1.4cm}	\includegraphics[width=7.45in]{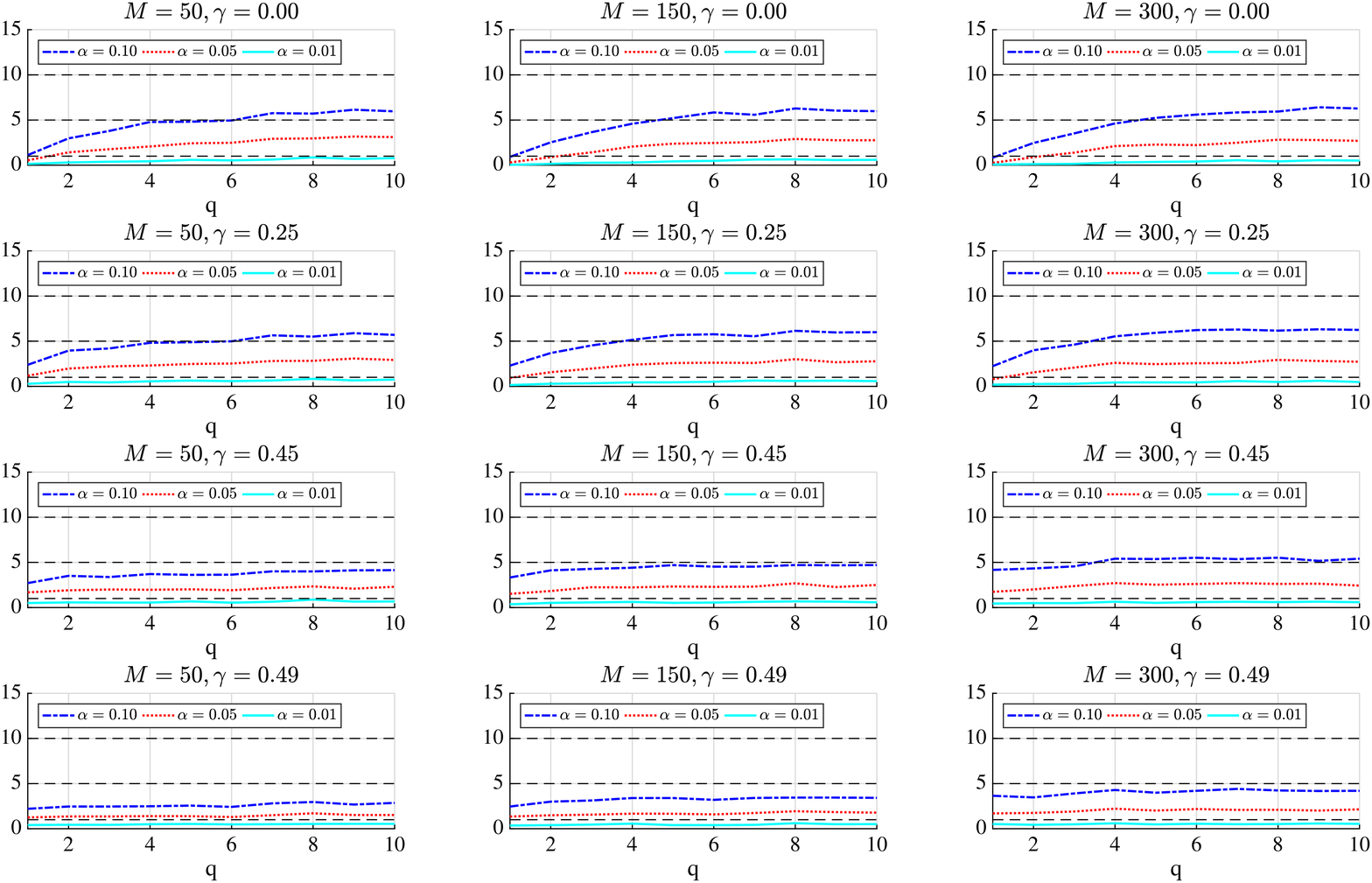}	
		\label{sim-null-f4}	
	\end{figure}

Next we consider the behaviour of the monitoring scheme under the alternatives discussed in Theorems \ref{th-alt-0}--\ref{expo}. We recall that under $H_A$
\begin{equation}\label{sim-alt-eq1}
 y_t = \left \{ \begin{aligned}
	    &{\bw}_t^\top \bar{\bbe}_0 + y_{t-1} \beta_{0,d} + \eps_t, \quad & 1 \leq t \leq M+s^* \\
	    &{\bw}_t^\top \bar{\bde}_M + y_{t-1} \delta_{M,d} + \eps_t, \quad & M+s^*+1 \leq t \leq T.
	\end{aligned} \right.
\end{equation}
The explanatory variables $(x_{t,2}, x_{t,3}, x_{t,4}, x_{t,5})$ are generated as in DGP(ii), i.e.\ dependent AR(1) sequences. The variables $\epsilon_t$ are independent standard  normals or GARCH (1,1) sequences. As before, we used the boundary function $\hat{g}(M,s)$ of \eqref{modbou}. The significance levels were $\alpha=.10, .05, .01$ and $s^*=1, 10$.   We considered the following data generating processes:\\
DGP(v) We used the initial values $\bar{\bbe}_0 = (.02, .20, .25, .15, -.20)^\T, \beta_{0,6} = .25 $ which changes to
$\bar{\bde}_M = (.04, 1.60, .75, .55, 1.20)^\T, \delta_{M,6} = .60$ at time $s^*$ after the training sample. The errors $\epsilon_t$ are independent and identically distributed random variables, independent of $\{(x_{t,2}, x_{t,3}, x_{t,4},$ $x_{t,5}), -\infty<t<\infty\}$.\\
DGP(vi) The data generating process is as in DGP(v) but now $\epsilon_t$ is given by the GARCH (1,1) sequence
\beq\label{garch-4}
\eps_t = \sigma_{t, \eps}h_{t, \eps},\quad \sigma_{t, \eps}^2 = .2+.3\eps_{t-1}^2+.3\sigma_{t-1, \eps}^2,-\infty<t<\infty,
\eeq
where $\{h_{t, \eps}, -\infty<t<\infty\}$ are independent standard normal random variables, independent of $\{(x_{t,2}, \ldots , x_{t,5}), -\infty<t<\infty\}$.\\
DGP(vii). In this case $\bar{\bbe}_0 = (.02, .20, .25, .15, -.20)^\T=\bar{\bde}_M$, but $\beta_{0,6}=.25$  changes to $\delta_{M,6}=.9, .95, .99$ and $1$. As in DGP(v), the $\epsilon_t$'s are independent standard normals, independent of $\{(x_{t,2}, x_{t,3}, x_{t,4}, x_{t,5}), -\infty<t<\infty\}$.\\
DGP(viii) We have the same parameters as in DGP(vii) but now $\epsilon_t$ is a GARCH(1,1) sequence satisfying \eqref{garch-4}. \\
DGP(ix)   The initial values $\bbe_0$ are the same as in DGP(v)--DGP(viii) but now $\bar{\bde}_M = (.04, 1.60, .75, .55, 1.20)^\T$ as in DGP(v) but $\delta_{M,6}=.9, .95, .99$ and $1$. The variables $\epsilon_t$ are independent standard normals.\\
DGP(x) The assumptions are the same as in DGP(ix) but now we use the GRACH(1,1) sequence of \eqref{garch-4} to generate the $\epsilon_t$'s.\\
DGP(xi)  The values of $\bbe_0$ and $\bar{\bde}_M$ are the same as in DGP(ix) and DGP(x), but now  $\delta_{M, 6}$ = 1.01, 1.05, 1.10 and 1.25. The variables $\eps_t$ are independent standard normals. \\
DGP(xii) The assumptions are the same as in DGP(xi) but now we use the GRACH(1,1) sequence of \eqref{garch-4} to generate the $\eps_t$'s. \\

The results of the simulations are given in Tables \ref{sim-alt-t1}--\ref{sim-alt-t7}. The empirical probability of stopping under the alternative is high in all cases we considered. The power increases with $\gamma$ except with slight drop at $\gamma=.49$ which is very close to the boundary case. The rate of convergence to the limit slows with the increase of $\gamma$ which is a possible explanation for the unexpected slight drop in power. Also the results show that our method is tailored to detect early changes, i.e.\ when $s^*$ is small. As expected, the power is increasing in Tables \ref{sim-alt-t2}--\ref{sim-alt-t5} as $\delta_{M,6}$ gets closer to 1. Allowing  $\bar{\bde}_M$ to differ, we increased the power substantially for $\delta_{M,6}=.9 $ and $.95$ but only mildly for  $\delta_{M,6}=.99 $ and 1. In this case the change to partial sum dominates the power. Based on our simulation study, we recommend  $\gamma=.45$ to achieve fast and reliable detection. This recommendation is also confirmed in Figures \ref{sim-alt-f1}--\ref{sim-alt-f4}, where the empirical density  of the stopping time $\tau_M$ is exhibited under different  assumptions. We note that according to Theorems \ref{alt-stat} and \ref{th-alt-3}, the limit distributions on Figure \ref{sim-alt-f1} and \ref{sim-alt-f3} can be approximated with normal densities as $M\to \infty$. The empirical densities have longer right tails than a normal density but they are clearly approaching a normal density. By Theorem \ref{th-alt-2}, the limits of the empirical densities on Figure \ref{sim-alt-f2} are not normal densities (cf.\ Remark \ref{rem-1}). The limit distribution in Theorem \ref{expo} is not necessarily normal. However, if the $\{\bw_t, \eps_t, -\infty<t<\infty\}$ are jointly normal, then the variable $Z_{M+s^*}$ of \eqref{zdef} is normally distributed. In Figure \ref{sim-alt-f4} the exhibited density is not derived from a normal distribution due to \eqref{garch-4}, the errors $\eps_t$ are only conditionally normal. Comparing Figures \ref{sim-alt-f1}--\ref{sim-alt-f4}, one sees that the limit distributions are getting less spread as $\delta_{M,6}$ increases, i.e.\ we need less and less observations to detect the change.
\begin{table}[H]
\caption{Empirical power of the sequential change--point monitoring  scheme under DGP(v) and DGP(vi).}	
		\centering	
		\resizebox{14cm}{!}{	
			\begin{threeparttable}
				\begin{tabular}{lllllllllllll}
					\toprule
					& \multicolumn{3}{c}{$M = 50, s^* = 1$} & \multicolumn{3}{c}{$M = 100, s^* = 1$} & \multicolumn{3}{c}{$M = 50, s^* = 10$} & \multicolumn{3}{c}{$M = 100, s^* = 10$} \\
					\cmidrule{2-13}
					$\gamma/\alpha$ & 0.10  & 0.05  & 0.01  & 0.10  & 0.05  & 0.01  & 0.10  & 0.05  & 0.01  & 0.10  & 0.05  & 0.01  \\
					\midrule
					& \multicolumn{12}{l}{\hspace{6cm} DGP(v)} \\
					\midrule
					0     & 87.64  & 80.75  & 62.63  & 99.64  & 99.07  & 95.98  & 72.68  & 61.53  & 40.26  & 98.78  & 97.53  & 92.44  \\
					0.25  & 91.00  & 86.08  & 72.70  & 99.75  & 99.52  & 97.91  & 77.57  & 68.88  & 50.09  & 99.29  & 98.52  & 95.70  \\
					0.45  & 89.67  & 85.29  & 74.78  & 99.71  & 99.47  & 98.25  & 73.80  & 66.07  & 51.13  & 99.03  & 98.39  & 96.06  \\
					0.49  & 86.31  & 81.29  & 70.45  & 99.54  & 99.22  & 97.71  & 67.56  & 60.53  & 45.82  & 98.53  & 97.66  & 94.61  \\
					\midrule
					& \multicolumn{12}{l}{\hspace{6cm} DGP(vi)} \\
					\midrule
					0     & 98.53  & 97.08  & 91.73  & 99.96  & 99.90  & 99.76  & 94.23  & 90.18  & 78.52  & 99.96  & 99.91  & 99.38  \\
					0.25  & 99.14  & 98.23  & 94.57  & 99.97  & 99.95  & 99.86  & 95.70  & 93.25  & 84.49  & 99.99  & 99.94  & 99.65  \\
					0.45  & 98.80  & 97.88  & 95.00  & 99.98  & 99.95  & 99.86  & 94.75  & 92.19  & 85.10  & 99.95  & 99.94  & 99.69  \\
					0.49  & 98.12  & 97.09  & 93.90  & 99.95  & 99.92  & 99.81  & 92.74  & 89.91  & 82.30  & 99.94  & 99.90  & 99.62  \\
					\bottomrule
				\end{tabular}
				\label{sim-alt-t1}			
		\end{threeparttable}}		
	\end{table}

	\begin{table}[H]
\caption{Empirical power of the sequential change--point monitoring  scheme under DGP(vii) .}	
		\centering	
		\resizebox{14cm}{!}{	
			\begin{threeparttable}
				\begin{tabular}{llllllllllllll}
					\toprule
					& & \multicolumn{3}{c}{$M = 50, s^* = 1$} & \multicolumn{3}{c}{$M = 100, s^* = 1$} & \multicolumn{3}{c}{$M = 50, s^* = 10$} & \multicolumn{3}{c}{$M = 100, s^* = 10$} \\
					\cmidrule{3-14}
					& $\gamma/\alpha$ & 0.10  & 0.05  & 0.01  & 0.10  & 0.05  & 0.01  & 0.10  & 0.05  & 0.01  & 0.10  & 0.05  & 0.01  \\
					\midrule
					\multirow{4}[2]{*}{$\delta_{M,6} = 0.90$} & 0     & 73.61  & 65.89  & 52.28  & 88.13  & 82.27  & 70.43  & 62.16  & 54.38  & 40.15  & 84.16  & 77.53  & 64.44  \\
					& 0.25  & 81.25  & 74.96  & 62.56  & 93.96  & 90.45  & 81.93  & 68.30  & 61.39  & 48.28  & 90.07  & 85.64  & 74.94  \\
					& 0.45  & 82.13  & 77.26  & 67.44  & 95.33  & 93.00  & 87.26  & 66.81  & 61.00  & 50.56  & 90.69  & 87.31  & 79.50  \\
					& 0.49  & 78.92  & 74.49  & 65.13  & 94.13  & 91.80  & 86.25  & 62.63  & 57.62  & 47.11  & 88.50  & 85.22  & 77.39  \\
					\midrule
					\multirow{4}[2]{*}{$\delta_{M,6} = 0.95$} & 0     & 86.12  & 81.37  & 71.56  & 96.32  & 94.45  & 89.09  & 76.66  & 70.76  & 59.95  & 94.63  & 91.81  & 85.13  \\
					& 0.25  & 90.51  & 86.91  & 78.61  & 98.24  & 97.05  & 94.09  & 81.09  & 75.91  & 66.23  & 96.91  & 95.33  & 90.52  \\
					& 0.45  & 90.68  & 88.00  & 81.87  & 98.61  & 97.97  & 95.99  & 79.84  & 75.47  & 67.84  & 97.11  & 95.88  & 92.70  \\
					& 0.49  & 88.95  & 86.32  & 80.19  & 98.27  & 97.56  & 95.60  & 76.43  & 72.82  & 65.31  & 96.32  & 95.05  & 91.66  \\
					\midrule
					\multirow{4}[2]{*}{$\delta_{M,6} = 0.99$} & 0     & 92.84  & 90.41  & 84.54  & 98.99  & 98.36  & 96.96  & 86.69  & 82.86  & 75.06  & 98.74  & 97.80  & 95.51  \\
					& 0.25  & 95.11  & 93.22  & 88.93  & 99.60  & 99.27  & 98.34  & 89.33  & 86.14  & 79.35  & 99.34  & 98.83  & 97.34  \\
					& 0.45  & 95.22  & 93.76  & 90.50  & 99.70  & 99.51  & 98.91  & 88.48  & 85.96  & 80.46  & 99.34  & 98.97  & 97.95  \\
					& 0.49  & 94.23  & 92.87  & 89.39  & 99.60  & 99.39  & 98.78  & 86.66  & 83.87  & 78.58  & 99.11  & 98.71  & 97.68  \\
					\midrule
					\multirow{4}[2]{*}{$\delta_{M,6} = 1.00$} & 0     & 93.67  & 91.50  & 86.86  & 99.36  & 98.88  & 97.78  & 88.50  & 85.55  & 78.53  & 99.15  & 98.49  & 96.91  \\
					& 0.25  & 95.49  & 93.86  & 90.25  & 99.74  & 99.52  & 98.85  & 90.62  & 88.16  & 82.35  & 99.56  & 99.18  & 98.20  \\
					& 0.45  & 95.70  & 94.32  & 91.73  & 99.83  & 99.66  & 99.34  & 89.97  & 87.93  & 83.29  & 99.50  & 99.27  & 98.60  \\
					& 0.49  & 94.72  & 93.58  & 90.83  & 99.72  & 99.61  & 99.16  & 88.54  & 86.44  & 81.68  & 99.35  & 99.16  & 98.42  \\
					\bottomrule
				\end{tabular}
				\label{sim-alt-t2}			
		\end{threeparttable}}		
	\end{table}

\begin{table}[H]	
\caption{Empirical power of the sequential change--point monitoring  scheme under DGP(viii).}
		\centering	
		\resizebox{14cm}{!}{	
			\begin{threeparttable}
				\begin{tabular}{llllllllllllll}
					\toprule
					& & \multicolumn{3}{c}{$M = 50, s^* = 1$} & \multicolumn{3}{c}{$M = 100, s^* = 1$} & \multicolumn{3}{c}{$M = 50, s^* = 10$} & \multicolumn{3}{c}{$M = 100, s^* = 10$} \\
					\cmidrule{3-14}
					& $\gamma/\alpha$ & 0.10  & 0.05  & 0.01  & 0.10  & 0.05  & 0.01  & 0.10  & 0.05  & 0.01  & 0.10  & 0.05  & 0.01  \\
					\midrule
					\multirow{4}[2]{*}{$\delta_{M,6} = 0.90$} & 0     & 75.13  & 68.33  & 55.77  & 88.99  & 84.01  & 74.10  & 64.15  & 57.01  & 44.35  & 85.80  & 80.24  & 68.90  \\
					& 0.25  & 82.01  & 76.59  & 65.66  & 94.00  & 91.02  & 83.63  & 70.26  & 63.71  & 51.88  & 90.93  & 87.07  & 77.96  \\
					& 0.45  & 82.89  & 78.76  & 70.23  & 95.14  & 93.08  & 88.17  & 69.26  & 64.15  & 54.44  & 91.39  & 88.58  & 81.78  \\
					& 0.49  & 80.38  & 76.63  & 68.24  & 94.06  & 91.94  & 87.24  & 65.49  & 60.84  & 51.73  & 89.50  & 86.69  & 80.09  \\
					\midrule
					\multirow{4}[2]{*}{$\delta_{M,6} = 0.95$} & 0     & 86.77  & 82.47  & 73.56  & 96.73  & 94.85  & 90.20  & 77.44  & 72.22  & 61.89  & 95.15  & 92.93  & 87.36  \\
					& 0.25  & 90.70  & 87.57  & 80.45  & 98.38  & 97.38  & 94.43  & 81.78  & 77.03  & 68.42  & 97.07  & 95.65  & 91.87  \\
					& 0.45  & 91.18  & 88.73  & 83.39  & 98.60  & 97.93  & 96.18  & 80.99  & 77.18  & 70.27  & 97.06  & 96.19  & 93.42  \\
					& 0.49  & 89.68  & 87.36  & 81.97  & 98.33  & 97.56  & 95.75  & 78.23  & 74.78  & 67.93  & 96.51  & 95.49  & 92.63  \\
					\midrule
					\multirow{4}[2]{*}{$\delta_{M,6} = 0.99$} & 0     & 93.38  & 90.80  & 86.08  & 99.23  & 98.73  & 97.24  & 86.59  & 82.91  & 76.30  & 98.60  & 97.75  & 95.80  \\
					& 0.25  & 95.32  & 93.72  & 89.67  & 99.55  & 99.34  & 98.58  & 89.39  & 86.40  & 80.28  & 99.06  & 98.70  & 97.48  \\
					& 0.45  & 95.51  & 94.18  & 91.20  & 99.66  & 99.46  & 98.96  & 88.82  & 86.26  & 81.38  & 99.07  & 98.76  & 97.89  \\
					& 0.49  & 94.68  & 93.34  & 90.58  & 99.53  & 99.38  & 98.85  & 86.99  & 84.61  & 80.07  & 98.90  & 98.43  & 97.67  \\
					\midrule
					\multirow{4}[2]{*}{$\delta_{M,6} = 1.00$} & 0     & 94.48  & 92.64  & 88.34  & 99.48  & 99.07  & 98.14  & 88.52  & 85.27  & 79.45  & 99.01  & 98.51  & 97.22  \\
					& 0.25  & 96.28  & 94.69  & 91.41  & 99.71  & 99.55  & 99.05  & 90.78  & 88.20  & 82.92  & 99.34  & 99.02  & 98.22  \\
					& 0.45  & 96.32  & 95.20  & 92.71  & 99.77  & 99.63  & 99.32  & 90.32  & 88.17  & 83.80  & 99.38  & 99.14  & 98.53  \\
					& 0.49  & 95.55  & 94.44  & 91.97  & 99.74  & 99.58  & 99.24  & 88.80  & 86.81  & 82.64  & 99.19  & 98.97  & 98.31  \\
					\bottomrule
				\end{tabular}
				\label{sim-alt-t3}			
		\end{threeparttable}}		
	\end{table}

\begin{table}[H]
\caption{Empirical power of the sequential change--point monitoring  scheme under DGP(ix).}
\centering	
		\resizebox{14cm}{!}{	
			\begin{threeparttable}
				\begin{tabular}{llllllllllllll}
				\toprule	
					& & \multicolumn{3}{c}{$M = 50, s^* = 1$} & \multicolumn{3}{c}{$M = 100, s^* = 1$} & \multicolumn{3}{c}{$M = 50, s^* = 10$} & \multicolumn{3}{c}{$M = 100, s^* = 10$} \\	
					\cmidrule{3-14}
					& $\gamma/\alpha$ & 0.10  & 0.05  & 0.01  & 0.10  & 0.05  & 0.01  & 0.10  & 0.05  & 0.01  & 0.10  & 0.05  & 0.01  \\
					\midrule
					\multirow{4}[2]{*}{$\delta_{M,6} = 0.90$} & 0     & 90.28  & 86.93  & 80.03  & 98.74  & 98.01  & 95.81  & 81.91  & 77.41  & 68.04  & 97.90  & 96.75  & 94.03  \\
					& 0.25  & 93.28  & 90.73  & 84.82  & 99.36  & 99.01  & 97.68  & 85.13  & 81.13  & 73.23  & 98.84  & 98.06  & 96.02  \\
					& 0.45  & 93.53  & 91.35  & 86.92  & 99.51  & 99.23  & 98.46  & 83.84  & 80.78  & 74.35  & 98.84  & 98.30  & 96.64  \\
					& 0.49  & 92.14  & 90.05  & 85.37  & 99.38  & 99.10  & 98.22  & 81.44  & 78.20  & 71.89  & 98.51  & 97.84  & 96.27  \\
					\midrule
					\multirow{4}[2]{*}{$\delta_{M,6} = 0.95$} & 0     & 95.29  & 93.55  & 89.47  & 99.75  & 99.51  & 98.89  & 90.00  & 86.92  & 80.91  & 99.42  & 99.09  & 98.07  \\
					& 0.25  & 96.90  & 95.56  & 92.63  & 99.92  & 99.78  & 99.41  & 91.93  & 89.54  & 84.02  & 99.70  & 99.48  & 98.87  \\
					& 0.45  & 96.94  & 95.91  & 93.84  & 99.93  & 99.86  & 99.61  & 91.25  & 89.16  & 84.76  & 99.63  & 99.48  & 99.12  \\
					& 0.49  & 96.23  & 95.30  & 93.07  & 99.88  & 99.82  & 99.54  & 89.69  & 87.65  & 83.25  & 99.54  & 99.37  & 98.96  \\
					\midrule
					\multirow{4}[2]{*}{$\delta_{M,6} = 0.99$} & 0     & 97.61  & 96.79  & 94.85  & 99.94  & 99.90  & 99.67  & 94.07  & 92.33  & 88.86  & 99.82  & 99.78  & 99.45  \\
					& 0.25  & 98.40  & 97.68  & 96.40  & 99.97  & 99.95  & 99.88  & 95.27  & 93.92  & 90.87  & 99.89  & 99.82  & 99.70  \\
					& 0.45  & 98.47  & 97.96  & 96.87  & 99.97  & 99.95  & 99.92  & 94.90  & 93.68  & 91.31  & 99.89  & 99.85  & 99.76  \\
					& 0.49  & 98.08  & 97.63  & 96.48  & 99.97  & 99.95  & 99.91  & 93.98  & 92.92  & 90.52  & 99.87  & 99.82  & 99.72  \\
					\midrule
					\multirow{4}[2]{*}{$\delta_{M,6} = 1.00$} & 0     & 98.02  & 97.18  & 95.47  & 99.97  & 99.89  & 99.79  & 94.92  & 93.37  & 90.47  & 99.85  & 99.80  & 99.69  \\
					& 0.25  & 98.62  & 97.99  & 96.68  & 99.98  & 99.97  & 99.87  & 96.13  & 94.69  & 92.13  & 99.94  & 99.86  & 99.78  \\
					& 0.45  & 98.66  & 98.05  & 97.18  & 99.98  & 99.98  & 99.92  & 95.73  & 94.45  & 92.50  & 99.95  & 99.89  & 99.80  \\
					& 0.49  & 98.21  & 97.80  & 96.83  & 99.98  & 99.97  & 99.91  & 94.75  & 93.73  & 91.82  & 99.89  & 99.86  & 99.78  \\
					\bottomrule
				\end{tabular}
				\label{sim-alt-t4}			
		\end{threeparttable}}		
	\end{table}

	\begin{table}[H]	
\caption{Empirical power of the sequential change--point monitoring  scheme under DGP(x).}
		\centering	
		\resizebox{14cm}{!}{	
			\begin{threeparttable}
				\begin{tabular}{llllllllllllll}
					\toprule
					& & \multicolumn{3}{c}{$M = 50, s^* = 1$} & \multicolumn{3}{c}{$M = 100, s^* = 1$} & \multicolumn{3}{c}{$M = 100, s^* = 5$} & \multicolumn{3}{c}{$M = 150, s^* = 5$} \\
					\cmidrule{3-14}
					& $\gamma/\alpha$ & 0.10  & 0.05  & 0.01  & 0.10  & 0.05  & 0.01  & 0.10  & 0.05  & 0.01  & 0.10  & 0.05  & 0.01  \\
					\midrule
					\multirow{4}[2]{*}{$\delta_{M,6} = 0.90$} & 0     & 97.23  & 95.82  & 92.95  & 99.89  & 99.80  & 99.59  & 93.21  & 90.62  & 85.31  & 99.80  & 99.63  & 99.23  \\
					& 0.25  & 98.24  & 97.34  & 94.90  & 99.97  & 99.92  & 99.75  & 94.66  & 92.74  & 88.30  & 99.89  & 99.83  & 99.54  \\
					& 0.45  & 98.17  & 97.38  & 95.64  & 99.97  & 99.94  & 99.85  & 94.14  & 92.69  & 88.91  & 99.87  & 99.82  & 99.62  \\
					& 0.49  & 97.62  & 96.93  & 95.13  & 99.96  & 99.92  & 99.77  & 93.01  & 91.32  & 87.77  & 99.86  & 99.73  & 99.54  \\
					\midrule
					\multirow{4}[2]{*}{$\delta_{M,6} = 0.95$} & 0     & 98.88  & 98.36  & 96.95  & 99.98  & 99.97  & 99.92  & 96.67  & 95.42  & 92.57  & 99.93  & 99.88  & 99.81  \\
					& 0.25  & 99.26  & 98.94  & 98.00  & 100.00  & 100.00  & 99.98  & 97.41  & 96.49  & 94.28  & 99.94  & 99.93  & 99.88  \\
					& 0.45  & 99.25  & 98.96  & 98.34  & 100.00  & 100.00  & 99.99  & 97.18  & 96.40  & 94.59  & 99.93  & 99.93  & 99.89  \\
					& 0.49  & 99.09  & 98.80  & 98.12  & 100.00  & 100.00  & 99.98  & 96.57  & 95.76  & 93.94  & 99.93  & 99.92  & 99.87  \\
					\midrule
					\multirow{4}[2]{*}{$\delta_{M,6} = 0.99$} & 0     & 99.47  & 99.28  & 98.69  & 99.99  & 99.98  & 99.97  & 98.09  & 97.56  & 95.93  & 99.98  & 99.96  & 99.94  \\
					& 0.25  & 99.62  & 99.48  & 99.13  & 100.00  & 100.00  & 99.99  & 98.57  & 98.04  & 96.90  & 99.99  & 99.99  & 99.96  \\
					& 0.45  & 99.61  & 99.54  & 99.22  & 100.00  & 100.00  & 100.00  & 98.47  & 98.04  & 97.10  & 99.99  & 99.99  & 99.96  \\
					& 0.49  & 99.55  & 99.44  & 99.15  & 100.00  & 100.00  & 100.00  & 98.14  & 97.76  & 96.67  & 99.99  & 99.96  & 99.96  \\
					\midrule
					\multirow{4}[2]{*}{$\delta_{M,6} = 1.00$} & 0     & 99.58  & 99.36  & 98.96  & 100.00  & 99.99  & 99.98  & 98.42  & 97.84  & 96.58  & 99.99  & 99.97  & 99.94  \\
					& 0.25  & 99.71  & 99.56  & 99.28  & 100.00  & 100.00  & 99.99  & 98.82  & 98.27  & 97.38  & 99.99  & 99.99  & 99.97  \\
					& 0.45  & 99.69  & 99.56  & 99.38  & 100.00  & 100.00  & 100.00  & 98.72  & 98.25  & 97.58  & 99.99  & 99.99  & 99.97  \\
					& 0.49  & 99.59  & 99.51  & 99.29  & 100.00  & 100.00  & 100.00  & 98.33  & 98.00  & 97.16  & 99.99  & 99.98  & 99.96  \\
					\bottomrule
				\end{tabular}
				\label{sim-alt-t5}			
		\end{threeparttable}}		
	\end{table}

\begin{table}[H]
\caption{Empirical power of the sequential change-point monitoring scheme under DGP(xi).}	
		\centering	
		\resizebox{14cm}{!}{	
			\begin{threeparttable}
				\begin{tabular}{llllllllllllll}
					\toprule
					& & \multicolumn{3}{c}{$M = 50, s^* = 1$} & \multicolumn{3}{c}{$M = 100, s^* = 1$} & \multicolumn{3}{c}{$M = 100, s^* = 5$} & \multicolumn{3}{c}{$M = 150, s^* = 5$} \\
					\cmidrule{3-14}
					& $\gamma/\alpha$ & 0.10  & 0.05  & 0.01  & 0.10  & 0.05  & 0.01  & 0.10  & 0.05  & 0.01  & 0.10  & 0.05  & 0.01  \\
					\midrule
					\multirow{4}[2]{*}{$\delta_{M,6} = 1.01$} & 0     & 99.98  & 99.98  & 99.97  & 100.00  & 100.00  & 100.00  & 99.94  & 99.87  & 99.74  & 100.00  & 100.00  & 100.00  \\
					& 0.25  & 99.98  & 99.98  & 99.98  & 100.00  & 100.00  & 100.00  & 99.96  & 99.93  & 99.80  & 100.00  & 100.00  & 100.00  \\
					& 0.45  & 99.99  & 99.99  & 99.98  & 100.00  & 100.00  & 100.00  & 99.95  & 99.93  & 99.84  & 100.00  & 100.00  & 100.00  \\
					& 0.49  & 99.99  & 99.99  & 99.98  & 100.00  & 100.00  & 100.00  & 99.93  & 99.91  & 99.80  & 100.00  & 100.00  & 100.00  \\
					\midrule
					\multirow{4}[2]{*}{$\delta_{M,6} = 1.05$} & 0     & 99.98  & 99.98  & 99.94  & 100.00  & 100.00  & 100.00  & 99.96  & 99.96  & 99.89  & 100.00  & 100.00  & 100.00  \\
					& 0.25  & 99.98  & 99.98  & 99.97  & 100.00  & 100.00  & 100.00  & 99.99  & 99.96  & 99.93  & 100.00  & 100.00  & 100.00  \\
					& 0.45  & 99.98  & 99.98  & 99.97  & 100.00  & 100.00  & 100.00  & 99.97  & 99.97  & 99.94  & 100.00  & 100.00  & 100.00  \\
					& 0.49  & 99.98  & 99.97  & 99.97  & 100.00  & 100.00  & 100.00  & 99.97  & 99.96  & 99.92  & 100.00  & 100.00  & 100.00  \\
					\midrule
					\multirow{4}[2]{*}{$\delta_{M,6} = 1.10$} & 0     & 100.00  & 100.00  & 99.99  & 100.00  & 100.00  & 100.00  & 99.97  & 99.95  & 99.91  & 100.00  & 100.00  & 100.00  \\
					& 0.25  & 100.00  & 100.00  & 100.00  & 100.00  & 100.00  & 100.00  & 99.98  & 99.97  & 99.94  & 100.00  & 100.00  & 100.00  \\
					& 0.45  & 100.00  & 100.00  & 100.00  & 100.00  & 100.00  & 100.00  & 99.98  & 99.97  & 99.95  & 100.00  & 100.00  & 100.00  \\
					& 0.49  & 100.00  & 100.00  & 100.00  & 100.00  & 100.00  & 100.00  & 99.97  & 99.95  & 99.93  & 100.00  & 100.00  & 100.00  \\
					\midrule
					\multirow{4}[2]{*}{$\delta_{M,6} = 1.25$} & 0     & 100.00  & 100.00  & 100.00  & 100.00  & 100.00  & 100.00  & 100.00  & 99.99  & 99.99  & 100.00  & 100.00  & 100.00  \\
					& 0.25  & 100.00  & 100.00  & 100.00  & 100.00  & 100.00  & 100.00  & 100.00  & 100.00  & 99.99  & 100.00  & 100.00  & 100.00  \\
					& 0.45  & 100.00  & 100.00  & 100.00  & 100.00  & 100.00  & 100.00  & 100.00  & 100.00  & 99.99  & 100.00  & 100.00  & 100.00  \\
					& 0.49  & 100.00  & 100.00  & 100.00  & 100.00  & 100.00  & 100.00  & 100.00  & 99.99  & 99.99  & 100.00  & 100.00  & 100.00  \\
					\bottomrule
				\end{tabular}
				\label{sim-alt-t6}			
		\end{threeparttable}}		
	\end{table}

	\begin{table}[H]	
\caption{Empirical power of the sequential change-point monitoring scheme under DGP(xii).}
		\centering	
		\resizebox{14cm}{!}{	
			\begin{threeparttable}
				\begin{tabular}{llllllllllllll}
					\toprule
					& & \multicolumn{3}{c}{$M = 50, s^* = 1$} & \multicolumn{3}{c}{$M = 100, s^* = 1$} & \multicolumn{3}{c}{$M = 100, s^* = 5$} & \multicolumn{3}{c}{$M = 150, s^* = 5$} \\
					\cmidrule{3-14}
					& $\gamma/\alpha$ & 0.10  & 0.05  & 0.01  & 0.10  & 0.05  & 0.01  & 0.10  & 0.05  & 0.01  & 0.10  & 0.05  & 0.01  \\
					\midrule
					\multirow{4}[2]{*}{$\delta_{M,6} = 1.01$} & 0     & 100.00  & 100.00  & 100.00  & 100.00  & 100.00  & 100.00  & 99.99  & 99.98  & 99.98  & 100.00  & 100.00  & 100.00  \\
					& 0.25  & 100.00  & 100.00  & 100.00  & 100.00  & 100.00  & 100.00  & 99.99  & 99.99  & 99.98  & 100.00  & 100.00  & 100.00  \\
					& 0.45  & 100.00  & 100.00  & 100.00  & 100.00  & 100.00  & 100.00  & 99.99  & 99.99  & 99.98  & 100.00  & 100.00  & 100.00  \\
					& 0.49  & 100.00  & 100.00  & 100.00  & 100.00  & 100.00  & 100.00  & 99.99  & 99.98  & 99.98  & 100.00  & 100.00  & 100.00  \\
					\midrule
					\multirow{4}[2]{*}{$\delta_{M,6} = 1.05$} & 0     & 100.00  & 100.00  & 100.00  & 100.00  & 100.00  & 100.00  & 99.99  & 99.99  & 99.98  & 100.00  & 100.00  & 100.00  \\
					& 0.25  & 100.00  & 100.00  & 100.00  & 100.00  & 100.00  & 100.00  & 100.00  & 99.99  & 99.99  & 100.00  & 100.00  & 100.00  \\
					& 0.45  & 100.00  & 100.00  & 100.00  & 100.00  & 100.00  & 100.00  & 100.00  & 99.99  & 99.99  & 100.00  & 100.00  & 100.00  \\
					& 0.49  & 100.00  & 100.00  & 100.00  & 100.00  & 100.00  & 100.00  & 99.99  & 99.99  & 99.99  & 100.00  & 100.00  & 100.00  \\
					\midrule
					\multirow{4}[2]{*}{$\delta_{M,6} = 1.10$} & 0     & 100.00  & 100.00  & 100.00  & 100.00  & 100.00  & 100.00  & 100.00  & 100.00  & 100.00  & 100.00  & 100.00  & 100.00  \\
					& 0.25  & 100.00  & 100.00  & 100.00  & 100.00  & 100.00  & 100.00  & 100.00  & 100.00  & 100.00  & 100.00  & 100.00  & 100.00  \\
					& 0.45  & 100.00  & 100.00  & 100.00  & 100.00  & 100.00  & 100.00  & 100.00  & 100.00  & 100.00  & 100.00  & 100.00  & 100.00  \\
					& 0.49  & 100.00  & 100.00  & 100.00  & 100.00  & 100.00  & 100.00  & 100.00  & 100.00  & 100.00  & 100.00  & 100.00  & 100.00  \\
					\midrule
					\multirow{4}[2]{*}{$\delta_{M,6} = 1.25$} & 0     & 100.00  & 100.00  & 100.00  & 100.00  & 100.00  & 100.00  & 100.00  & 100.00  & 100.00  & 100.00  & 100.00  & 100.00  \\
					& 0.25  & 100.00  & 100.00  & 100.00  & 100.00  & 100.00  & 100.00  & 100.00  & 100.00  & 100.00  & 100.00  & 100.00  & 100.00  \\
					& 0.45  & 100.00  & 100.00  & 100.00  & 100.00  & 100.00  & 100.00  & 100.00  & 100.00  & 100.00  & 100.00  & 100.00  & 100.00  \\
					& 0.49  & 100.00  & 100.00  & 100.00  & 100.00  & 100.00  & 100.00  & 100.00  & 100.00  & 100.00  & 100.00  & 100.00  & 100.00  \\
					\bottomrule
				\end{tabular}
				\label{sim-alt-t7}			
		\end{threeparttable}}		
	\end{table}

\begin{figure}[H]
		\caption{The empirical density functions of $\tau_M$ under DGP(vi)  with significance level $\alpha=.05$.}
		\centering		
		\includegraphics[width=6.50in]{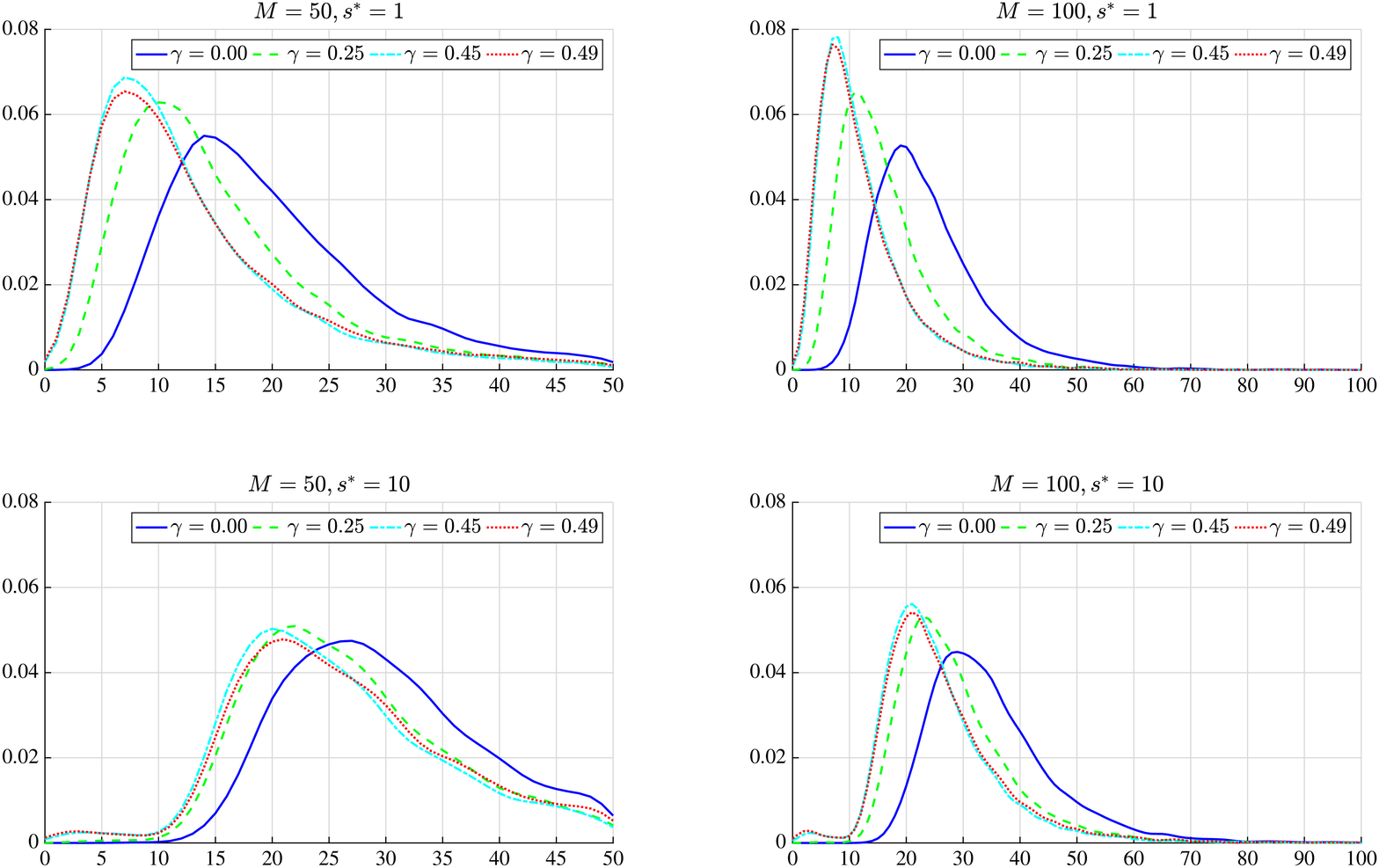}	
		\label{sim-alt-f1}	
	\end{figure}

\begin{figure}[H]
		\caption{The empirical density functions of $\tau_M$ under DGP(viii)  with $\delta_{M,6}=1$ and significance level $\alpha=.05$.}
		\centering		
		\includegraphics[width=6.0in]{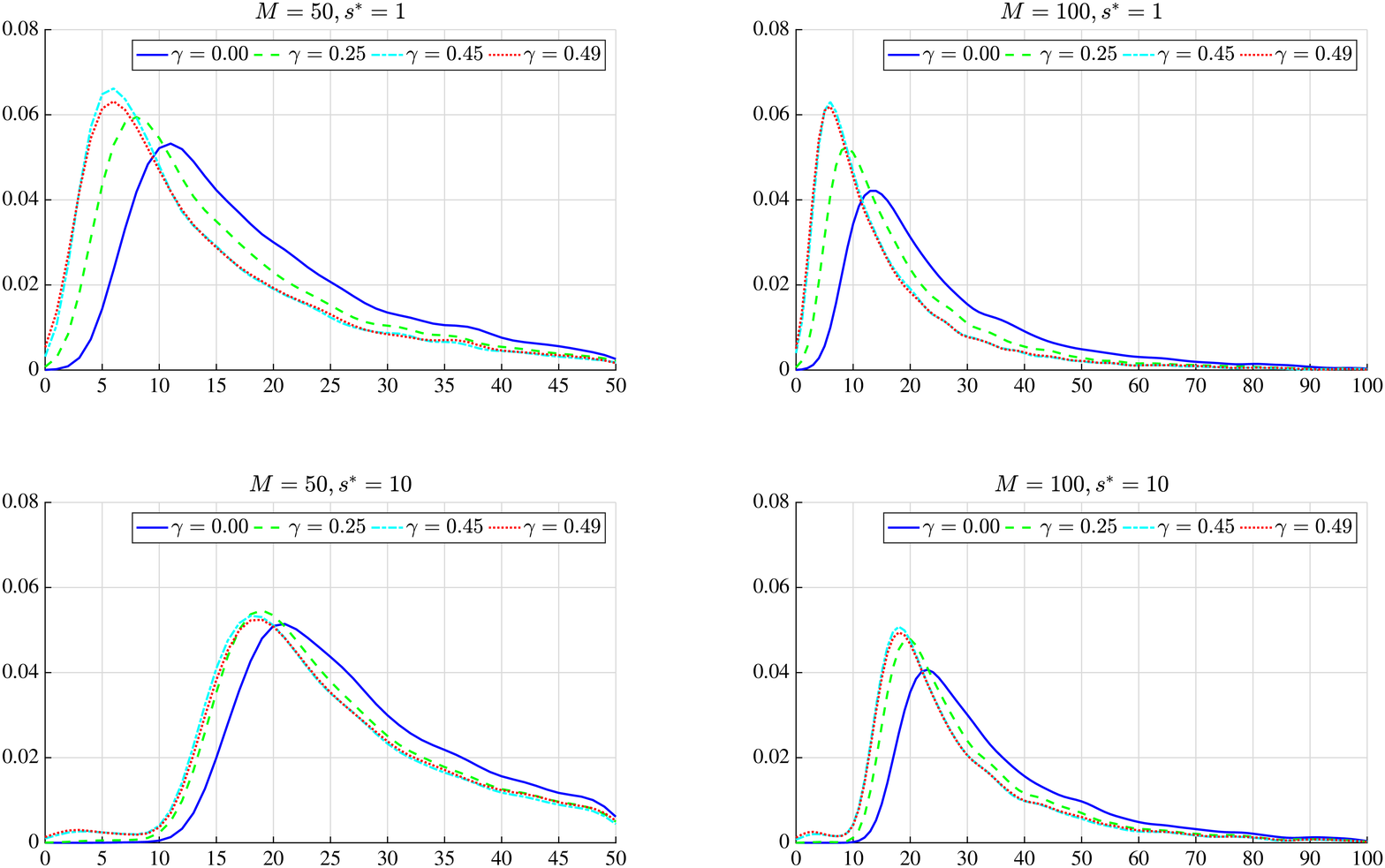}	
		\label{sim-alt-f2}	
	\end{figure}

	\begin{figure}[H]
		\caption{The empirical density functions of $\tau_M$ under DGP(x) with with $\delta_{M,6}=1$ and significance level $\alpha=.05$.}
		\centering		
		\includegraphics[width=6.0in]{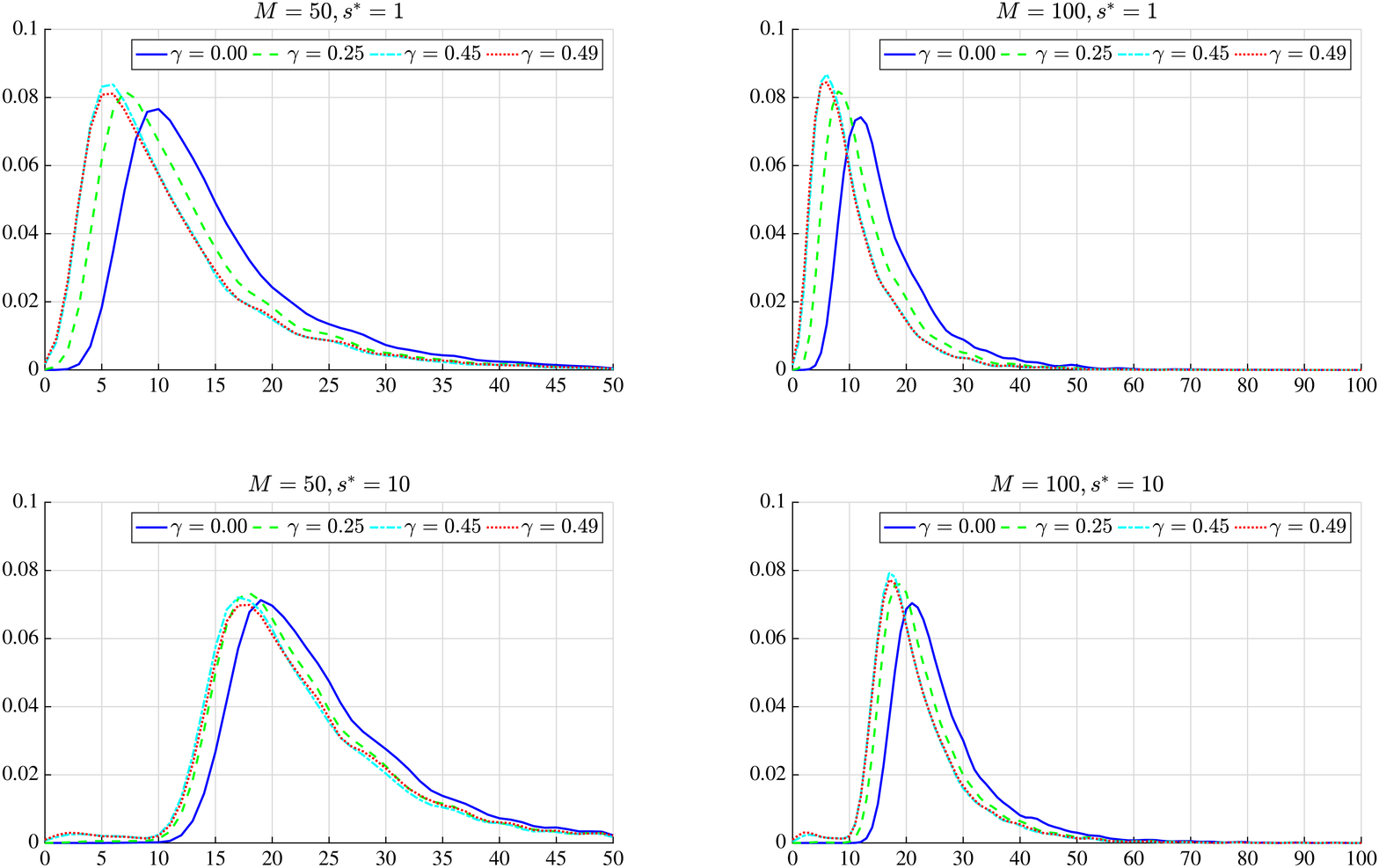}	
		\label{sim-alt-f3}	
	\end{figure}

\begin{figure}[H]
		\caption{The empirical density functions of $\tau_M$ under DGP(xii) with  $\delta_{M, 6} = 1.25$ and significance level $\alpha = .05$.}
		\centering		
		\includegraphics[width=6.0in]{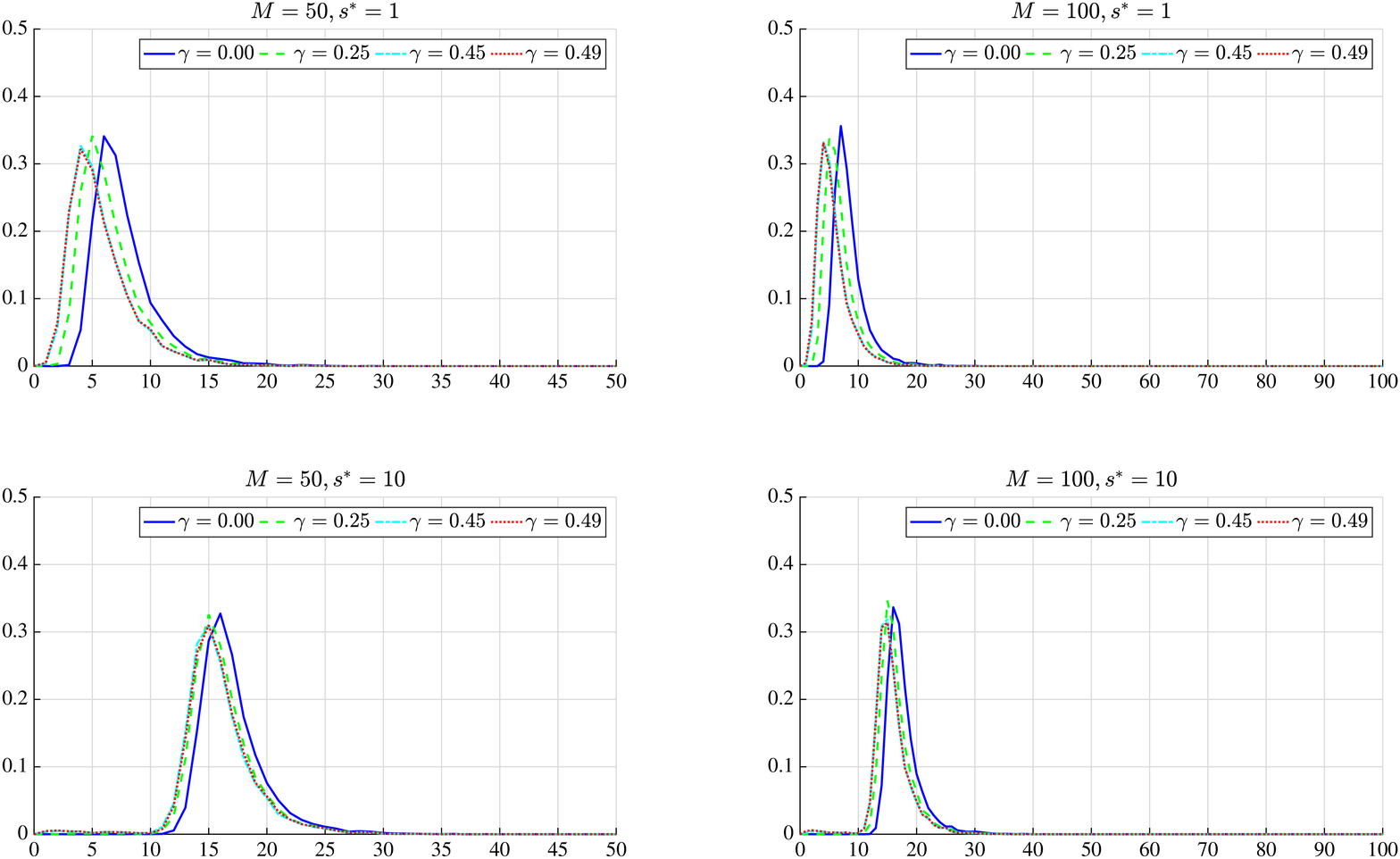}	
		\label{sim-alt-f4}	
	\end{figure}

\medskip

\section{An application to housing prices in the U.S.A. at the national level, in  the Los Angeles  and Boston markets }\label{sec-ch}
In this section, as an example for our theory, we  focus on the U.S.\ housing prices  to  illustrate our online monitoring procedure. The  literature  has discussed the link of housing prices to macroeconomic fundamental variables   using linear regression model.
 The fundamental variables frequently applied  in the literature  include personal income per capita, mortgage interest rate, employment on the demand side and housing starts on the supply side. These variables are used   to explain the dynamics of U.S.\ real estate prices in the long run horizon  (Case and Shiller, 2003; Gallin, 2006; Shiller, 2015). Beside these macroeconomic fundamental variables,  first--order autoregressive term of the change in the log housing prices was included in the regression model to account for the momentum effect  because  real estate acts as an investing instrument. For further information we refer to  Case and Shiller (1989, 2003),  Piazzesi and  Schneider (2009) and  Zheng et al.\ (2016).  In addition,  Himmelberg et al.\ (2005),  Davis and
Heathcote  (2007),  Saiz (2010) and Gyourko et al.\ (2013) suggested land supply elasticity, cost of ownership, demographic and geographic statistics  to explain the difference of housing prices across cities. \\
 We used  the S\&P CoreLogic Case--Shiller  Home Price Index series, which is the leading measure of U.S. residential real estate prices and  tracks changes in the value of residential real estate, as the proxy of housing prices.   We studied the housing prices in  U.S.\  at the national level and at two metropolitan areas: Los Angeles  and Boston. The S\&P CoreLogic Case--Shiller  Home Price Index series for these three markets are exhibited
   in Figure \ref{App-US-f1}. Figure \ref{App-US-f1} depicts an upward housing price trend in  the U.S.\ at the national level, as well as in Los Angeles and Boston  between January 1994 to December 2000.
   The set of macroeconomic fundamental variables included in our model:\\
$x_{t,2}$: the lagged \textit{disposable personal income per capita} change, we used the national level data as a proxy for Los Angeles and Boston since only yearly data of personal income per capita for states and Metropolitan Statistical Areas are available by the U.S.\ Bureau of Economic Analysis. \\
$x_{t,3}$: the lagged change of \textit{30-year fixed rate of mortgage average in U.S.}, transformed from weekly frequency to monthly.\\
$x_{t,4}$:  the lagged \textit{all non--farm employment} change in terms of the national level and corresponding Metropolitan Statistical Areas level originally released by the U.S.\ Bureau of Labor Statistics.\\
$x_{t,5}$:  the lagged change of \textit{housing starts} at the national level and the U.S.\ Census Bureau Regions (West Region series was used for Los Angeles and Northeast Region series was used for Boston). We use the lagged term of these variables here to mitigate the endogenous problem because of the interactive effect among the housing prices and these macroeconomic fundamentals (Case and  Shiller, 2003).\\
All data that we used are seasonally adjusted monthly data from the economic database of the Federal Reserve Bank of St.\ Louis\footnote{{\texttt https://fred.stlouisfed.org}}.

According to Shiller (2008), the beginning  of 1991 was the turning point in the 1980s boom. The housing prices started to drop and later they flattened out. Thus
we used January 1994--December 1996 as the training (historical) sample , so $M=36$ in our calculations. The upper part of Table \ref{app-des} reports the  summary statistics of all variables of the training sample we used in the regression. Statistics, including the number of observations, mean, standard deviation, minimum, maximum and the testing results of  the KPSS  test (Kwiatkowski et al.\ 1992) to check stationarity without linear term are tabulated. According to our results, stationarity cannot be rejected for the training sample. The detector $\hat{g}(36, s)$ is defined by \eqref{modbou} with $\gamma=.45$ and $\alpha=.01$.   The boundary function as well as the detectors are given in Figure  \ref{App-US-f2}. According to our calculations, $\tau_{36}^{(1)}=10$ (October 1997) for the national level, $\tau_{36}^{(2)}=8$ (August 1997) for Los Angeles and $\tau_{36}^{(3)}=19$ (July 1998) for Boston. The detections of changes in the parameters of the model in \eqref{e-1} are denoted by vertical lines in Figure \ref{App-US-f1}. Table \ref{app-reg} shows the estimated values of the parameters for the periods $[1, 36]$ (training sample), $[1, 36+ \tau_{36}^{(i)}]$ (before detection) and $[37, 36+\tau_{36}^{(i)}+11]$
(after detection).

We checked for more possible changes in the data in each market after the detection of changes at $\tau_{36}^{(i)}, i=1,2,3$. We used the training periods $[\tau_{36}^{(i)}, \tau_{36}^{(i)}+35], i=1,2,3$, i.e.\ October 1997--September 2000  for the national market, August 1997--July 2000 for Los Angeles and July 1998-June 2001 for Boston. The lower part of Table \ref{app-des} shows the summary statistics and the values of the stationarity test for these training samples. We started a new monitoring procedure for all three markets, the starting dates were October 2000  at the national level, August 2000  for Los Angeles and July 2001 for Boston. Our procedure detected changes at the national level and March 2004 was the estimated time of change. A change was also found for the Los Angeles market dated December 2003. No further changes were found on the Boston market. Figure \ref{App-II-f1} exhibits the housing price indices  and the time of the changes are indicated by vertical lines. Figure \ref{App-II-f2} shows the boundary function and the detectors. We note that on Figure \ref{App-II-f2} the monitoring starts at the same point but it is a different physical time for the three markets. It is clear from Table \ref{app-reg} that the autoregressive parameter changes if there is a change and it is increasing with time.
 However, with the exception of the national market, the autoregressive parameter stays far away from 1. The estimates are .93 and .84 for the national market and for Los Angeles, respectively. During the second monitoring phase, structural breaks were detected almost two years before the prices peaked in 2006 during the 2000s real estate boom.

 Our monitoring process finds increasing autoregressive parameters in the three markets and hence it  confirms the ``amplification mechanism" advocated by Case and Shiller (2003). The ``amplification mechanism" is the strongest in Los Angeles, which was undergoing faster price changes than Boston. Since the autoregressive parameters are below 1 in the first and also in the second phase of our monitoring, it is unlikely that ``bubbles" formed in the sense of Linton (2019). It is also useful to note that the estimated R--square is increasing with the autoregressive parameter, so the autoregressive part explains more and more of the changes in the housing prices. The momentum effect, caused by the herding behavior of transactions, tends to disengage  the $log$
of housing price index changes from the macro fundamentals.

\begin{figure}[H]
		\caption{Housing Price Index at  U.S.\ national level and two metropolitan areas  between January 1994  and December 2000.}
		\centering		
		\includegraphics[width=6.0in]{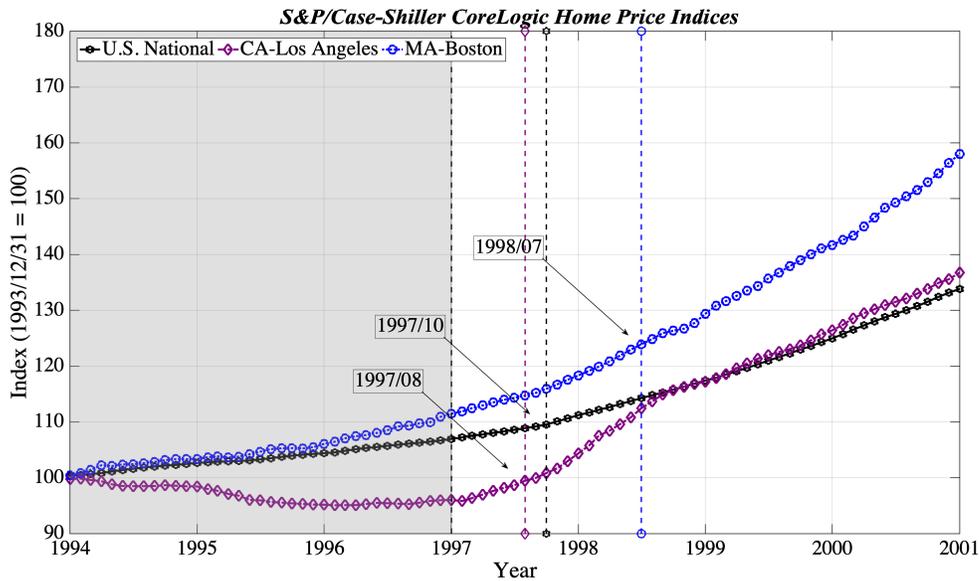}	
		\label{App-US-f1}	
	\end{figure}

\begin{table}[H]	
	\centering	
\caption{Variable definitions and summary statistics.}
	\resizebox{16cm}{!}{	
		\begin{threeparttable}
			\begin{tabular}{llllllll}
				\toprule
				\multicolumn{2}{l}{Variables} & Sample Size & Mean  & SD    & Min   & Max   & KPSS-test \\
				\midrule
				\multicolumn{8}{c}{Training Sample for the first monitoring: January 1994 - December 1996} \\
				\midrule
				Housing Market Indicators: & \multicolumn{7}{l}{S\&P Case--Shiller Home Price Indices} \\
				$y_t:\Delta \mathrm{log(CSHPI)}$ & U.S. National Level & 36    & 0.0019 & 0.0006 & 0.0005 & 0.0034 & 0.1300 \\
				& Los Angeles & 36    & -0.0011 & 0.0023 & -0.0075 & 0.0032 & 0.3070 \\
				& Boston & 36    & 0.0030 & 0.0026 & -0.0021 & 0.0098 & 0.2647 \\
				\multicolumn{8}{l}{Fundamentals:} \\
				$x_{t,2}:\Delta \mathrm{log(DPIPC)\_lag1}$ & Disposable Personal Income per Capita & 36    & 0.0011 & 0.0054 & -0.0198 & 0.0116 & 0.2328 \\
				$x_{t,3}:\Delta \mathrm{log(Mortgage Rate)\_lag1}$ & 30-Year Fixed Mortagage Rate & 36    & 0.0016 & 0.0342 & -0.0501 & 0.0802 & 0.1795 \\
				$x_{t,4}:\Delta \mathrm{log(Employment)\_lag1}$ & \multicolumn{7}{l}{All Employees, Total Nonfarm} \\
				& U.S. National Level & 36    & 0.0021 & 0.0009 & -0.0002 & 0.0041 & 0.3036 \\
				& Los Angeles  & 36    & 0.0012 & 0.0015 & -0.0034 & 0.0040 & 0.2872 \\
				& Boston & 36    & 0.0017 & 0.0013 & -0.0027 & 0.0038 & 0.1188 \\
				$x_{t,5}:\Delta \mathrm{log(HStarts)\_lag1}$ & \multicolumn{7}{l}{New Privately Owned Housing Units Started} \\
				& U.S. National Level & 36    & -0.0031 & 0.0633 & -0.1866 & 0.1568 & 0.1173 \\
				& Los Angeles & 36    & -0.0082 & 0.1275 & -0.3027 & 0.2446 & 0.1848 \\
				& Boston & 36    & 0.0028 & 0.1439 & -0.2776 & 0.3502 & 0.1517 \\
				\midrule
				\multicolumn{8}{c}{Training Sample for the second monitoring : October 1997--September 2000 (national level),  August 1997--July 2000  (Los Angeles) and July 1998--June 2001 (Boston)} \\
				\midrule
				Housing Market Indicators: & \multicolumn{7}{l}{S\&P Case--Shiller Home Price Indices}  \\
				$y_t:\Delta \mathrm{log(CSHPI)}$ & U.S. National Level & 36    & 0.0060 & 0.0010 & 0.0035 & 0.0078 & 0.3505 \\
				& Los Angeles & 36    & 0.0091 & 0.0033 & 0.0039 & 0.0176 & 0.2429 \\
				& Boston & 36    & 0.0110 & 0.0038 & 0.0036 & 0.0188 & 0.2626 \\
				\multicolumn{8}{l}{Fundamentals:} \\
				$x_{t,2}:\Delta \mathrm{log(DPIPC)\_lag1}$ & Disposable Personal Income per Capita & 36    & 0.0030 & 0.0025 & -0.0029 & 0.0079 & 0.1355 \\
				$x_{t,3}:\Delta \mathrm{log(Mortgage Rate)\_lag1}$ & 30-Year Fixed Mortagage Rate & 36    & 0.0017 & 0.0218 & -0.0293 & 0.0551 & 0.1987 \\
				$x_{t,4}:\Delta \mathrm{log(Employment)\_lag1}$ & \multicolumn{7}{l}{All Employees, Total Nonfarm} \\
				& U.S. National Level & 36    & 0.0019 & 0.0009 & -0.0003 & 0.0036 & 0.3527 \\
				& Los Angeles  & 36    & 0.0019 & 0.0018 & -0.0018 & 0.0062 & 0.1853 \\
				& Boston & 36    & 0.0011 & 0.0024 & -0.0043 & 0.0062 & 0.2207 \\
				$x_{t,5}:\Delta \mathrm{log(HStarts)\_lag1}$ & \multicolumn{7}{l}{New Privately Owned Housing Units Started} \\
				& U.S. National Level & 36    & -0.0007 & 0.0442 & -0.0963 & 0.0807 & 0.2275 \\
				& Los Angeles & 36    & 0.0023 & 0.1094 & -0.2840 & 0.1909 & 0.2834 \\
				& Boston & 36    & -0.0010 & 0.1341 & -0.3432 & 0.1569 & 0.1961 \\
				\bottomrule
			\end{tabular}
Note: The  critical values for the KPSS test are 0.347 (10\% level), 0.463 (5\% level), 0.739 (1\% level).
			\label{app-des}		
		\end{threeparttable}}		
		\end{table}

	\begin{figure}[H]
		\caption{The first sequential monitoring for structural breaks  in the $\log$ of housing prices.}
		\centering		
		\includegraphics[width=6.0in]{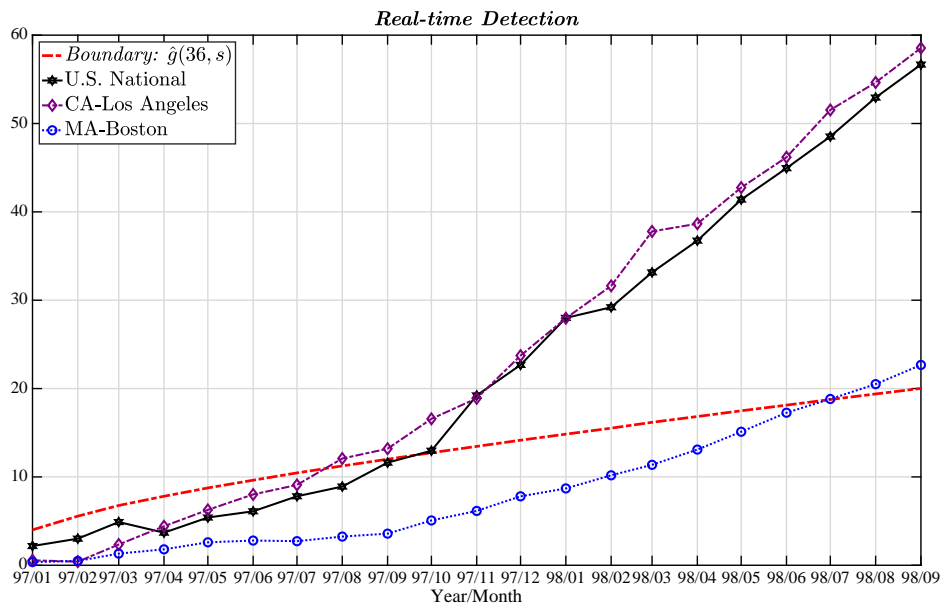}	
		\label{App-US-f2}	
	\end{figure}

	\begin{table}[H]	
		\centering	
\caption{Least square estimates for the model coefficients during the first and second monitoring. 
}
		\resizebox{16cm}{!}{	
			\begin{threeparttable}
				
				\begin{tabular}{lllllllllllll}
					\toprule
					&       & \multicolumn{3}{c}{U.S. National} &       & \multicolumn{3}{c}{Los Angeles} &       & \multicolumn{3}{c}{Boston} \\
					    Variables &       & Training  & Before & After &       & Training  & Before & After &       & Training  & Before & After \\
					&       & Sample & Detection & Detection &       & Sample & Detection & Detection &       & Sample & Detection & Detection \\
					\midrule
					\multicolumn{13}{c}{Estimated Coefficients During the first Monitoring  (the size of the training sample is $ 36$)} \\
					\midrule
					$\hat{\beta}_1$ &       & 0.0006 & 0.0006 & 0.0011 &       & -0.0009 & -0.0006 & 0.0010 &       & 0.0020 & 0.0012 & 0.0028 \\
					$\hat{\beta}_2$ &       & 0.0246 & 0.0271 & -0.1202 &       & -0.0452 & -0.0355 & 0.2139 &       & 0.0570 & 0.0766 & 0.0186 \\
					$\hat{\beta}_3$ &       & 0.1511 & 0.0625 & -0.0582 &       & 0.3545 & 0.5183 & 0.5264 &       & 0.0842 & 0.1432 & 0.1166 \\
					$\hat{\beta}_4$ &       & 0.0026 & 0.0014 & -0.0084 &       & 0.0002 & -0.0052 & -0.0337 &       & -0.0035 & -0.0114 & 0.0159 \\
					$\hat{\beta}_5$ &       & 0.0008 & 0.0010 & 0.0011 &       & 0.0017 & 0.0007 & -0.0107 &       & -0.0017 & -0.0018 & -0.0027 \\
					$\hat{\beta}_6$ &       & \textbf{0.4612} & \textbf{0.6295} & \textbf{0.8820} &       & \textbf{0.4060} & \textbf{0.6007} & \textbf{0.7135} &       & \textbf{0.2639} & \textbf{0.6316} & \textbf{0.6152} \\
					\midrule
					$R^2$    &       & 0.5172 & 0.4872 & 0.7347 &       & 0.3019 & 0.5153 & 0.7035 &       & 0.0673 & 0.3819 & 0.4155 \\
					\midrule
					\multicolumn{13}{c}{Estimated Coefficients During the Second Monitoring  (the size of the training sample is $ 36$)} \\
					\midrule
					$\hat{\beta}_1$ &       & 0.0015 & 0.0006 & 0.0008 &       & 0.0031 & 0.0020 & 0.0019 &       & 0.0054 &       &  \\
					$\hat{\beta}_2$ &       & 0.0008 & 0.0051 & 0.0019 &       & 0.3234 & 0.0434 & 0.0179 &       & -0.0689 &       &  \\
					$\hat{\beta}_3$ &       & 0.0363 & 0.0133 & 0.1101 &       & 0.1053 & -0.1486 & -0.0556 &       & 0.3305 &       &  \\
					$\hat{\beta}_4$ &       & 0.0059 & 0.0054 & 0.0072 &       & -0.0192 & -0.0031 & 0.0105 &       & 0.0189 &       &  \\
					$\hat{\beta}_5$ &       & -0.0005 & 0.0009 & 0.0008 &       & -0.0039 & -0.0020 & 0.0014 &       & -0.0067 &       &  \\
					$\hat{\beta}_6$ &       & \textbf{0.7474} & \textbf{0.9294} & \textbf{0.9100} &       & \textbf{0.5363} & \textbf{0.8353} & \textbf{0.8686} &       & \textbf{0.4854} &       &  \\
					\midrule
					$R^2$     &       & 0.6998 & 0.8234 & 0.8623 &       & 0.4766 & 0.6450 & 0.7850 &       & 0.2932 &       &  \\
					\bottomrule
				\end{tabular}
				\label{app-reg}		
		\end{threeparttable}}		
	\end{table}
\begin{figure}[H]
		\centering		
		\caption{Housing Price Index at  U.S.\ national level and two metropolitan areas  between August 1997  and June 2005.}
		\resizebox{16cm}{!}{	
			\includegraphics{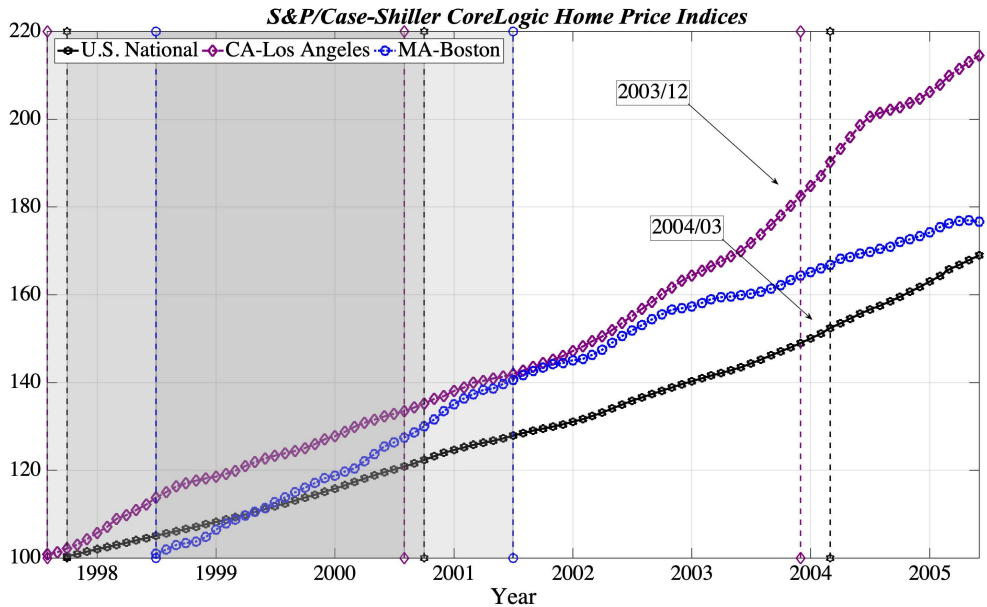}}	
		\label{App-II-f1}
	\end{figure}
	\begin{figure}[H]
		\centering		
		\caption{The second sequential monitoring for structural breaks  in the $\log$ of housing prices..}
		\resizebox{16cm}{!}{	
			\includegraphics{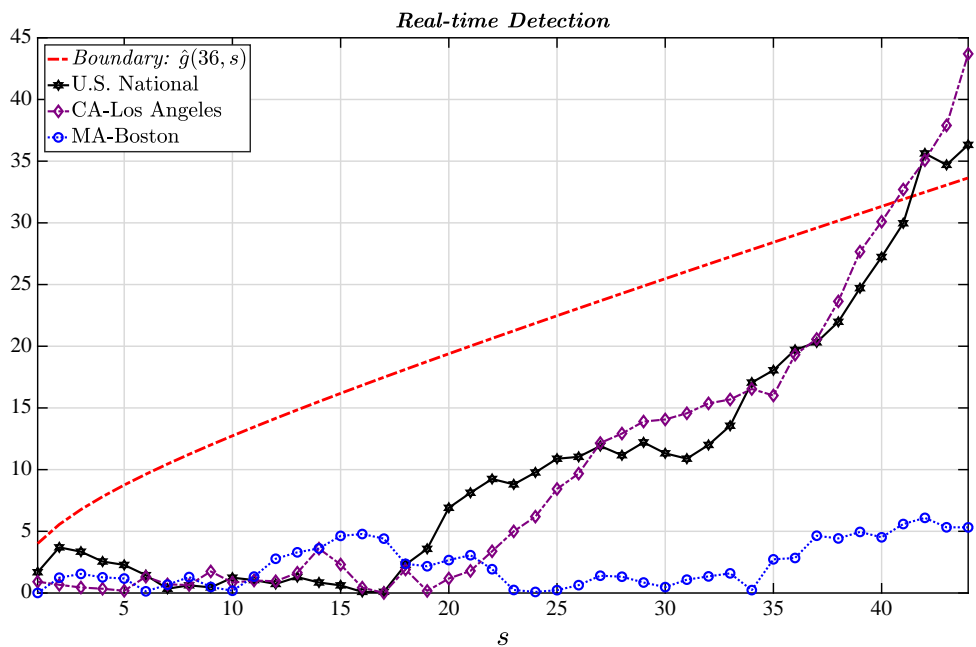}}	
		\label{App-II-f2}
	\end{figure}

\section{Conclusion}\label{concl}
In this paper we consider a model which includes linear and autoregressive terms to model changes in real estate prices. The observations and errors are weakly dependent, including the most often used linear and nonlinear time series sequences. We propose a sequential method to detect possible changes in the parameters of the model. The monitoring scheme is based on a detector and a suitably chosen boundary function. The limit distribution of the sequential monitoring scheme is established under the null hypothesis of stability of the model. We determine the asymptotic distribution of the stopping time when structural break is present. We focus on the possible changes in the autoregressive parameter. Using Monte Carlo simulations we illustrate that our results
can be applied in case of finite sample sizes. We suggest a boundary function which provides the right size of the monitoring even in case of small and moderate historical (training) samples. We also study the power of the procedure and the time to detect the structural break. A data example is also given.
We sequentially looking for possible structural breaks in the real estate markets of Boston, Los Angeles and   at the U.S.\ national level. We find structural breaks in the data, and find stationary segments. The autoregressive parameter of the segments is increasing but it stays below 1. Hence
the ``amplification mechanism" of Case and Shiller (2003) is confirmed by the data analysis but no bubbles in the sense of Linton (2019) were found.

\medskip
\noindent
{\bf Acknowledgements} Part of the research was done while Shanglin Lu was visiting the University of Utah. We  appreciate  the support of the Department of Mathematics.

\medskip
\appendix
\medskip
\section{Proof of Theorem \ref{th-null}}\label{sec-pr-1}
We assume in this section that $H_0$ holds. According to Assumption \ref{bern}, $x_{t,2}, x_{t,3}, \ldots ,x_{t,d-1}$ is an $m$--decomposable Bernoulli shift. Next we show that the autoregressive term $y_{t-1}$, the last coordinate of $\bx_t$, is also  is an $m$--decomposable Bernoulli  shift. We also show that some functions of $\bx_t$ and $\eps_t$ are also decomposable and we obtain the corresponding rate. Since $y_t$ is stationary under $H_0$, we obtain immediately that
\beq\label{ar-1}
|\beta_{0,d}|<1.
\eeq
\begin{lemma}\label{ybern} We assume that $H_0$ and Assumption \ref{bern} hold.\\
(i) There is  functional $a$ defined on ${\mathcal{ S}}$ with values in $R$  such that $y_t=a(\eta_t, \eta_{t-1}, \eta_{t-2},\ldots)$, $E|y_t|^{\kappa_1}<\infty$ and with some constant $C$
\beq\label{ydeco}
\left(E\left|y_t-y_{t,m} \right|^{\kappa_1}\right)^{1/\kappa_1}\leq C m^{-\kappa_2}
\eeq
where $y_{t,m}=a(\eta_t, \eta_{t-1}, \eta_{t-2}, \ldots, \eta_{t-m+1}, \bet^*_{t,m})$, $\{\eta_t, \bet^*_{t,m}, -\infty<t,m<\infty\}$ are defined in Assumption \ref{bern}.\\
(ii) Also,
\beq\label{dobx}
\left(E\|\bx_t\bx_t^\T-\bx_{t,m}\bx_{t,m}^\T\|^{\kappa_1/2}\right)^{2/\kappa_1}\leq C m^{-\kappa_2}
\eeq
and
\beq\label{xeps}
\left(E\|\bx_t\eps_t-\bx_{t,m}\eps_{t,m}\|^{\kappa_1/2}\right)^{2/\kappa_1}\leq C m^{-\kappa_2}
\eeq
\end{lemma}
\begin{proof}
Elementary arguments give that
\beq\label{yrep}
y_t=\sum_{\ell=0}^\infty \beta_{0,d}^\ell\left(\bw_{t-\ell}^\T\bar{\bbe}_0+\eps_{t-\ell}\right),
\eeq
where $\bw_t$ is defined in \eqref{bwdef} and $\bar{\bbe}_0=(\beta_{0,1},\beta_{0,2},\ldots ,\beta_{0,d-1})^\T$. By the stationarity of $\bw_t$ and $\eps_t$ we have
\begin{align*}
(E|y_t|^{\kappa_1})^{1/\kappa_1}\leq \sum_{\ell=0}^\infty |\beta_{0,d}|^\ell\left\{E\left(\|\bw_{t-\ell}\|\|\bar{\bbe}_0\|+\|\eps_{t-\ell}\|\right)^{\kappa_1}\right\}^{1/\kappa_1}<\infty.
\end{align*}
Using now Assumption \ref{bern}, the Bernoulli representation for $y_t$ is established. According to the definition of $y_{t,m}$ we have that
$$
y_{t,m}=\sum_{\ell=0}^{m-1} \beta_{0,d}^\ell\left(\bw_{t-\ell, m-\ell}^\T\bar{\bbe}_0+\eps_{t-\ell,m-\ell}\right)+\sum_{\ell=m}^\infty \beta_{0,d}^\ell\left(\bw_{t-\ell}^\T\bar{\bbe}_0+\eps_{t-\ell}\right)
$$
and therefore
$$
y_t-y_{t,m}=\sum_{\ell=0}^{m-1} \beta_{0,d}^\ell\left((\bw_{t-\ell}-\bw_{t-\ell, m-\ell})^\T\bar{\bbe}_0+(\eps_{t-\ell}-\eps_{t-\ell,m-\ell})\right).
$$
Using Assumption \ref{bern}
\begin{align*}
\left(E|y_t-y_{t,m}|^{\kappa_1}\right)^{1/\kappa_1}&\leq \sum_{\ell=0}^{m-1} |\beta_{0,d}|^\ell\left\{E\left(\|\bw_{t-\ell}-\bw_{t-\ell,m-\ell}\|\|\bar{\bbe}_0\|+|\eps_{t-\ell}
-\eps_{t-\ell,m-\ell}|\right)^{\kappa_1}\right\}^{1/\kappa_1}\\
&\leq C_1 \sum_{\ell=0}^{m-1}(m-\ell)^{-\kappa_2}\\
&\leq C_2 m^{-\kappa_2}.
\end{align*}
Assumption \ref{bern} and \eqref{ydeco} imply that
\begin{align*}
E|x_{t,\ell}x_{t,k}-x_{t,\ell,m}x_{t,k,m}|^{\kappa_1/2}&\leq 2^{\kappa_1/2}(E[|x_{t,\ell}||x_{t,k}-x_{t,k,m}|]^{\kappa_1/2}+
E[|x_{t,k,m}||x_{t,\ell}-x_{t,\ell,m}|]^{\kappa_1/2})\\
&\leq 2^{\kappa_1/2}(E|x_{t,\ell}|^{\kappa_1}E|x_{t,k}-x_{t,k,m}|^{\kappa_1})^{1/2}\\
&\hspace{3cm}+(E|x_{t,k,m}|^{\kappa_1}E|x_{t,\ell}-x_{t,\ell,m}|^{\kappa_1})^{1/2}\\
&\leq cm^{-\kappa_1\kappa_2/2},
\end{align*}
where $x_{t,j,m}$ is the $j^{\mbox{th}}$ coordinate of $\bx_{t,m}$ and $c$ is a constant. Hence \eqref{dobx} is proven.
Similar argument gives \eqref{xeps}.
\end{proof}
\begin{lemma}\label{hatga}We assume that $H_0$ and Assumptions \ref{bern}--\ref{a-adef} hold.
(i) As $M\to\infty$ we have that
$$
M^{1/2}(\hat{\bbe}_M-\bbe_0)\;\stackrel{{\mathcal D}}{\to}\;\bN,
$$
where $\bN$ is a $d$--dimensional normal random vector with $E\bN={\bf 0}$ and $E\bN\bN^\T=\bA^{-1}\bD\bA^{-1}$,
\end{lemma}
where $\bD=\{d_{k,\ell}, 1\leq k,\ell \leq d\}$ with
$$
d_{k,\ell}=\sum_{t=-\infty}^\infty Ex_{0,k}x_{t,\ell}\eps_0\eps_t.
$$
(ii) If $\{ x_{t,\ell}, -\infty< t<\infty, 2\leq \ell\leq d-1\}$ and $\{\eps_t, -\infty<t<\infty\}$ are independent,  then $\bD=\sigma^2\bA$.
\begin{proof}
Using \eqref{e-1} we get that
$$
\bY_M={\bX}_M\bbe_0+\bE_M
$$
with $\bE_M=(\eps_1, \eps_2, \ldots ,\eps_M)^\T$. It is well known that
\begin{align}\label{be-1}
\hat{\bbe}_M
=\bbe_0
+({\bX}_M^\T{\bX}_M)^{-1}{\bX}_M^\T\bE_M.
\end{align}

Using Lemma \ref{ybern} and Lemma B.1  in Aue et al.\ (2014)
we obtain that
$$
E\left(\sum_{t=1}^M(x_{t,k}-Ex_{t,k})\right)^2=O(M),\;\;2\leq k\leq d
$$
and
$$
E\left(\sum_{t=1}^M(x_{t,k}-Ex_{0,k})(x_{t,\ell}-Ex_{0,\ell})-M\mbox{cov}(x_{0,k},x_{0,\ell})\right)^2=O(M), \;\;2\leq k,\ell\leq d
$$
so by Markov's inequality we have
\beq\label{mark}
\frac{1}{M}\sum_{t=1}^Mx_{t,k}=Ex_{0,k}+O_P(M^{-1/2}),\;\;2\leq k \leq d
\eeq
and
$$
\frac{1}{M}\sum_{t=1}^Mx_{t,k}x_{t,\ell}=Ex_{0,k}x_{0,\ell}+O_P(M^{-1/2}), \;\;2\leq k,\ell\leq  d.
$$
Hence
\beq\label{be-2}
\frac{1}{M}{\bX}_M^\T{\bX}_M=\bA+O_P(M^{-1/2}),
\eeq
where $\bA=\{ a_{k,\ell},1\leq k,\ell \leq d\}$,
\beq\label{aexp}
a_{1,1}=1,\; a_{1,k}=a_{k,1}=Ex_{0,k},\; a_{k,\ell}=Ex_{0,k}x_{0,\ell}\;\;2\leq k, \ell\leq d.
\eeq
and therefore \eqref{adef} is proven. Using Assumption \ref{a-adef} we get that
\beq\label{aconv}
({\bX}_M^\T{\bX}_M)^{-1}=\frac{1}{M}\bA^{-1}+O_P(M^{-3/2}).
\eeq
It follows from Assumption \ref{bernmean} that
\beq\label{longvar}
E\left(M^{-1/2}\sum_{t=1}^M{x}_{t,\ell}\eps_t\right)\left( M^{-1/2}\sum_{t=1}^M{x}_{t,\ell}\eps_t  \right)=
\frac{1}{M}\sum_{t,s=1}^Mx_{t,k}\eps_tx_{s,\ell}\eps_s\to d_{k,\ell}.
\eeq
Applying now Theorem B.1  in Aue et al.\ (2014) we conclude that
\beq\label{xenor}
M^{-1/2}\sum_{t=1}^M\bx_{t}\eps_t\stackrel{{\mathcal D}}{\to}\bN_1,
\eeq
where $\bN_1$ is a $d$--dimensional normal random vector with $E\bN_1={\bf 0}$ and $\bN_1\bN_1^\T=\bD$.
Hence  the proof the first part of Lemma \ref{hatga} is now complete. \\
It follows immediately from the independence of $\{ x_{t,\ell}, -\infty< t<\infty, 2\leq \ell\leq  d-1\}$ and $\{\eps_t, -\infty<t<\infty\}$ and Assumption \ref{bernmean} that
$$
Ex_{0,1}x_{t,1}\eps_0\eps_t=E\eps_0\eps_t=\sigma^2,\;\;\mbox{if}\;\;t=0\;\;\mbox{and}\;\; 0\;\;\mbox{if}\;\;t\neq 0,
$$
for $2\leq k\leq d-1$
$$
Ex_{0,1}x_{t,k}\eps_0\eps_t=Ex_{t,k}E\eps_0\eps_t=\sigma^2Ex_{0,k},\;\;\mbox{if}\;\;t=0\;\;\mbox{and}\;\; 0\;\;\mbox{if}\;\;t\neq 0.
$$
Similarly,

$$
Ex_{0,k}x_{t,\ell}\eps_0\eps_t=Ex_{t,k}x_{t,\ell}E\eps_0\eps_t=\sigma^2Ex_{0,k}x_{0,\ell},\;\;\mbox{if}\;\;t=0\;\;\mbox{and}\;\; 0\;\;\mbox{if}\;\;t\neq 0.
$$
By the definition, $x_{t,d}=y_{t-1}$. Using the representation in \eqref{yrep} we get that
$$
Ex_{0,1}x_{t,d}\eps_0\eps_t=Ey_{t-1}\eps_0\eps_t=\sigma^2Ey_{0},\;\;\mbox{if}\;\;t=0\;\;\mbox{and}\;\; 0\;\;\mbox{if}\;\;t\neq 0,
$$
$$
Ex_{s,d}x_{t,d}\eps_t\eps_s=Ey_{t-1}y_{s-1}\eps_t\eps_s=Ey_{t-1}y_{s-1}\eps_tE\eps_s=0,\;\;\mbox{if}\;\;s>t
$$
and
$$
Ex_{t,d}^2\eps_t^2=Ey_{t-1}^2\eps_t^2=\sigma^2Ey_0^2.
$$
Hence the proof of Lemma \ref{hatga} is complete.
\end{proof}
\begin{lemma}\label{xlaw} If $H_0$ and  Assumptions \ref{bern}--\ref{bernmean} hold, then we have that
\beq\label{xbound}
\sup_{1\leq t <\infty}\frac{1}{t^\zeta}\left\|\sum_{s=M+1}^{t+M}{\bx}_s-\ba t\right\|=O_P(1)
\eeq
for any $\zeta>1/2,$ where $\ba=(1,Ex_{0,2}, \ldots, Ex_{0, d-1}, Ey_0 )^\T$.
\end{lemma}
\begin{proof}
By the stationarity of $\bx_t$, Assumption \ref{bern} and Lemma \ref{ybern} yield for all $ 2\leq \ell\leq d$
(cf.\  Lemma B.1 of Aue et al.\ 2014
and  the maximal  inequality in  Billingsley  1968, p.\ 94) that
$$
E\max_{1\leq s\leq t}\left(\sum_{u=M+1}^{M+s}(x_{u,\ell}-Ex_{0,\ell})\right)^4\leq C_1t^2,
$$
where $C_1$ is a constant. We write that
\begin{align*}
\sup_{1\leq t<\infty}\frac{1}{t^\zeta}\left|\sum_{u=M+1}^{M+t}(x_{u,\ell}-Ex_{0,\ell})\right|
&\leq \sup_{1\leq i <\infty}\max_{e^{i-1}<t\leq e^i}\frac{1}{t^\zeta}\left|\sum_{u=M+1}^{M+t}(x_{u,\ell}-Ex_{0,\ell})\right|\\
&\leq \sup_{1\leq i <\infty}e^{-(i-1)\zeta}\max_{e^{i-1}<t\leq e^i}\left|\sum_{u=M+1}^{M+t}(x_{u,\ell}-Ex_{0,\ell})\right|\\
&\leq \sup_{1\leq i <\infty}e^{-(i-1)\zeta}\max_{1\leq t\leq e^i}\left|\sum_{u=M+1}^{M+t}(x_{u,\ell}-Ex_{0,\ell})\right|.
\end{align*}
Hence for any $v>0$ we have that
\begin{align*}
&P\left\{ \sup_{1\leq t <\infty}\frac{1}{t^\zeta}\left\|\sum_{s=M+1}^{t+M}{\bx}_s-\ba t\right\|>v \right\}\\
&\leq P\left\{ \sup_{1\leq i <\infty}e^{-(i-1)\zeta}\max_{e^{i-1}<t\leq e^i}\left|\sum_{u=M+1}^{M+t}(x_{u,\ell}-Ex_{0,\ell})\right|     >v \right\}\\
&\leq \sum_{i=1}^\infty P\left\{ \max_{1\leq t\leq e^i}\left|\sum_{u=M+1}^{M+t}(x_{u,\ell}-Ex_{0,\ell})\right|     >ve^{(i-1)\zeta} \right\}\\
&\leq \sum_{i=1}^\infty\frac{1}{v^4}e^{-4(i-1)\zeta}E\max_{1\leq t\leq e^i}\left|\sum_{u=M+1}^{M+t}(x_{u,\ell}-Ex_{0,\ell})\right|^4\\
&\leq C_2\sum_{i=1}^\infty\frac{1}{v^4}e^{-4(i-1)\zeta}e^{2i}\\
&\leq \frac{C_3}{v^4}
\end{align*}
with some constant $C_2$ and $C_3$ on account of $\zeta>1/2.$ Thus we conclude that for all $M$
$$
\lim_{v\to\infty}P\left\{\sup_{1\leq t <\infty}\frac{1}{t^\zeta}\left\|\sum_{s=M+1}^{t+M}{\bx}_s-\ba t\right\|>v
\right\}=0,
$$
completing the proof of Lemma \ref{xlaw}.
\end{proof}
\begin{lemma}\label{empeps} If $H_0$ and Assumptions \ref{bern}--\ref{a-adef} hold, then we have that
\begin{align*}
\sup_{1\leq s<\infty}\left|\sum_{u=M+1}^{M+s}\hat{\eps}_u-\left(\sum_{u=M+1}^{M+s}{\eps}_u  -\frac{s}{M}\sum_{u=1}^{M}{\eps}_u \right)   \right|\biggl/g(M,s)=o_P(1).
\end{align*}
\end{lemma}
\begin{proof}It follows from the definition of $\hat{\eps}_u$ that
\begin{align*}
\sum_{u=M+1}^{M+s}\hat{\eps}_u=
\sum_{u=M+1}^{M+s}{\eps}_u-\sum_{u=M+1}^{M+s}{\bx}_u^\T(\hat{\bbe}_M-\bbe_0).
\end{align*}
Using Lemmas \ref{hatga} and  \ref{xlaw} we get
\begin{align*}
\sup_{1\leq s<\infty}\left\|\sum_{u=M+1}^{M+s}\bx_u-s\ba\right\|\|\hat{\bbe}_M-\bbe_0\|\biggl/g(M,s)&=O_P(M^{-1/2})\sup_{1\leq s <\infty}\frac{s^\zeta}{g(M,s)}\\
&=O_P\left( M^{\zeta-1} \right)\sup_{0<u<\infty}\frac{u^\zeta}{(1+u)(u/(1+u))^\gamma}\\
&=o_P(1),
\end{align*}
since $u^\zeta/[(1+u)(u/(1+u))^\gamma]$ is bounded on $(0, \infty)$. Using now \eqref{be-1} and \eqref{be-2}
we conclude
$$
\ba^\T(\hat{\bbe}_M-\bbe_0)=\ba^\T\frac{1}{M}(\bA^{-1}+O_P(M^{-1/2}))\bX^\T_M\bE_M=\frac{1}{M}\ba^\T\bA^{-1}\bX^\T_M\bE_M+O_P(1/M),
$$
since according to \eqref{xenor}
$$
\left\|\bX^\T_M\bE_M\right\|=O_P(M^{1/2}).
$$
Observing that
$$
\frac{1}{M}\frac{s}{g(M,s)}=o(1)
$$
we get that
$$
\left|\frac{s}{M}\ba^\T\left[\left(\frac{1}{M}\bX^\T_M\bX_M\right)^{-1}-\bA^{-1}\right]\bX^\T_M\bE_M  \right|\biggl/g(M,s)=o_P(1).
$$
We note that $\ba$ is the first row (column) of $\bA$ so $\ba^\T\bA^{-1}=(1,0,\ldots, 0)$ and $x_{t,1}=1$ for all $t$ by definition. Hence
$$
\ba^\T\bA^{-1}\bX^\T_M\bE_M=\sum_{u=1}^M\eps_u,
$$
completing the proof of Lemma \ref{empeps}.
\end{proof}
\begin{lemma}\label{epsapp} If Assumption \ref{bern} holds, then for every $M$ we can define two independent Wiener processes $\{W_{M,1}(s), 0\leq s \leq M\}$ and $\{W_{M,2}(s), 0\leq s <\infty\}$ such that
\beq\label{oneapp}
\left|\sum_{u=1}^M\eps_u-\sigma W_{M,1}(M)\right\|=O_P(M^{1/2-\delta})
\eeq
and
\beq\label{twoapp}
\sup_{1\leq s <\infty}\frac{1}{s^{1/2-\delta}}\left|\sum_{u=M+1}^{M+s}\eps_u-\sigma W_{M,2}(s)\right|=O_P(1)
\eeq
with some $\delta>0$.
\end{lemma}
\begin{proof} According to Assumption \ref{bern}, $\eps_t$ is an $m$ decomposable Bernoulli shift. Hence Theorem B.1  of Aue et al.\ (2014) implies both \eqref{oneapp} and \eqref{twoapp}.
\end{proof}

\noindent
{\it Proof of Theorem \ref{th-null}.} First we note that by the proof of Lemma \ref{empeps} we have
\beq\label{sicos}
\hat{\sigma}_M^2\stackrel{P}{\to}\sigma^2.
\eeq
Hence according to Lemma \ref{empeps} we need to show only
\begin{align}\label{f-0}
\sup_{1\leq s<\infty} \frac{\displaystyle\left| \sum_{u=M+1}^{M+s}\eps_u-\frac{s}{M}\sum_{u=1}^M\eps_u\right|}{\displaystyle \sigma M^{1/2}\left(1+\frac{s}{M}\right)\left(\frac{s}{s+M}\right)^\gamma}\;\;\stackrel{{\mathcal D}}{\to}\;\;\sup_{0<t\leq 1}\frac{|W(t)|}{t^\gamma},
\end{align}
where $W$ stands for a Wiener process. Using \eqref{oneapp} we get
\begin{align}\label{f-1}
\sup_{1\leq s<\infty}\left|\frac{s}{M}\left(\sum_{u=1}^M\eps_u- \sigma W_{M,1}(M)\right)  \right|\biggl/g(M,s)
&=O_P(M^{1/2-\delta})\sup_{1\leq s <\infty}\frac{s/M}{g(M,s)}\\
&=O_P(M^{-\delta})\sup_{0<x<\infty}\frac{x}{(1+x)(x/(1+x))^\gamma}\notag\\
&=o_P(1).\notag
\end{align}
Similarly, \eqref{twoapp} implies
\begin{align}\label{f-2}
\sup_{1\leq s <\infty}&\left|\sum_{u=M+1}^{M+s}\eps_u-\sigma W_{M,2}(s)  \right|\biggl/g(M,s)\\
&=O_P(1)\sup_{1\leq s <\infty}\frac{s^{1/2-\delta}}{g(M,s)}\notag\\
&=O_P(M^{-\delta})\sup_{0<x<\infty}\frac{x^{1/2-\delta}}{(1+x)(x/(1+x))^\gamma}\notag\\
&=o_P(1),\notag
\end{align}
since we can assume without loss of generality that $0<\delta<1/2-\gamma$. By the scale transformation of the Wiener process we have that
\begin{align}\label{f-3}
\sup_{1\leq s <\infty}\frac{\displaystyle \left|W_{M,2}(s)-\frac{s}{M}W_{M,1}(M)\right|}{M^{1/2}(1+s/M)(s/(M+s))^\gamma}
&\stackrel{{\mathcal D}}{=}\sup_{1\leq s <\infty}\frac{\displaystyle \left|W_{2}(s/M)-\frac{s}{M}W_{1}(1)\right|}{(1+s/M)(s/(M+s))^\gamma}\\
&=\sup_{1/M\leq  x <\infty}\frac{\displaystyle \left|W_{2}(x)-xW_{1}(1)\right|}{(1+x)(x/(x+1))^\gamma},\notag\\
&\stackrel{{\mbox{a.s.}}}{\to}
\sup_{0< x <\infty}\frac{\displaystyle \left|W_{2}(x)-xW_{1}(1)\right|}{(1+x)(x/(x+1))^\gamma},\notag
\end{align}
where $W_1$ and $W_2$ are independent Wiener processes. It is shown in Chu et al.\ (1996) (cf.\ also Horv\'ath et al., 2004) that
\begin{align*}
\left\{W_1(x)-xW_2(1), 0\leq x<\infty \right\}\stackrel{{\mathcal D}}{=}\left\{(1+x)W(x/(1+x)), 0\leq x<\infty \right\},
\end{align*}
where $W$ stands for a Wiener process. Hence
$$
\sup_{0< x <\infty}\frac{\displaystyle \left|W_{1}(x)-xW_{2}(1)\right|}{(1+x)(x/(x+1))^\gamma}\stackrel{{\mathcal D}}{=}\sup_{0<x\leq 1}\frac{|W(x)|}{x^\gamma},
$$
and therefore \eqref{f-0} follows from \eqref{f-1}--\eqref{f-3}.
\qed

\medskip

\section{Proof of Theorems \ref{th-alt-0}--\ref{expo}}\label{sec-pr-2}

\noindent
{\it Proof of Theorem \ref{th-alt-0}.} It follows from the proof of Theorem \ref{th-null}
$$
\max_{1\leq s \leq s^*}\frac{\Gamma(M,s)}{g(M,s)}=O_P(1).
$$
Using Lemma \ref{xlaw} we get
\beq\label{xnull}
\max_{1\leq s \leq s^*}\frac{1}{s}\left\|\sum_{u=M+1}^{M+s}\bx_u\right\|=O_P(1).
\eeq
Following the proof of Lemma \ref{xnull}, Assumption \ref{destat} implies that for any $\zeta>1/2$
$$
\max_{s^*<s<\infty}\frac{1}{(s-s^*)^\zeta}\left\|\sum_{u=M+s^*}^{M+s}(\bx_u-E\bx_u)\right\|=O_P(1)
$$
and
$$
\lim_{s-s^*\to \infty}\left\|\sum_{u=M+s^*}^{M+s}(E\bx_u-(1, Ex_{0,2},\ldots ,Ex_{0,d-1}, Ey_A)^\T)\right\|=0.
$$
Thus we conclude
\beq\label{xalta}
\max_{s^*<s<\infty}\frac{1}{s-s^*}\left\|\sum_{u=M+s^*}^{M+s}\bx_u\right\|=O_P(1).
\eeq
Hence Lemma \ref{hatga} yields
$$
\max_{ s^*<s<\infty}\frac{\displaystyle \sum_{u=M+1}^{M+s}\left|\bx_u^\T(\bbe_0-\hat{\bbe}_M)\right|}{g(M,s)}=O_P(1).
$$
Let
$$
s_M=C\left(\frac{M^{1/2-\gamma}}{|\Delta_M|}  \right)^{1/(1-\gamma)}.
$$
We showed that
\begin{align*}
\frac{|\Gamma(M,s_M)|}{g(M, s_M)}&=O_P(1)+ \frac{\displaystyle\left|\sum_{u=M+1}^{M+s_M}\bx^\T_u(\bbe_0-\bde_M)\right|}{\hat{\sigma}_M g(M,s_M)}\\
&=O_P(1)+ \frac{s_M|\Delta|}{c\sigma M^{1/2}(s_M/M)^\gamma}(1+o_P(1)).
\end{align*}
Assumptions of Theorem \ref{th-alt-0} yield
$$
\lim_{C\to \infty}\liminf_{M\to\infty}\frac{s_M|\Delta_M|}{c\sigma M^{1/2}(s_M/M)^\gamma}=\infty,
$$
completing the proof.
\qed\\

The proof of Theorem \ref{alt-stat} is based on a series of lemmas. The first lemma considers the detector before the time of change and it will be used in the proofs of Theorems \ref{th-alt-2}--\ref{expo} as well.

\begin{lemma}\label{cn-00} If Assumptions \ref{bern}--\ref{a-adef} hold and
\beq\label{s*m}
s^*/M\to 0,
\eeq

then we have that
\beq\label{s*m-1}
\max_{1\leq s \leq s^*}\frac{\displaystyle\left|\sum_{u=M+1}^{M+s}\hat{\eps}_u\right|}{\sigma M^{1/2}(1+s/M)(s/(M+s))^\gamma}
=O_P\left(\left(\frac{s^*}{M}\right)^{\gamma-1/2}\right).
\eeq
\end{lemma}
\begin{proof} It follows from the proof of Lemma \ref{empeps} that
\begin{align}\label{cn-2-1}
\max_{1\leq s\leq s^*}\frac{\displaystyle \left|\sum_{u=M+1}^{M+s}\hat{\eps}_u-\left( \sum_{u=M+1}^{M+s}\eps_u-\frac{s}{M} \sum_{u=1}^M\eps_u  \right)   \right|}{M^{1/2}(1+s/M)^{1-\gamma}(s/M)^\gamma}
=O_P\left(\left( \frac{s^*}{M} \right)^{1/2-\gamma}\right).
\end{align}
By \eqref{twoapp} we have
\begin{align*}
\max_{1\leq s\leq s^*}&\frac{\displaystyle \left| \sum_{u=M+1}^{M+s}\eps_u-\sigma W_{M,2}(s)   \right|}{M^{1/2}(1+s/M)^{1-\gamma}(s/M)^\gamma}\\
&=
O_P\left(\max_{1\leq s\leq s^*}\frac{ s^{1/2-\delta}}{M^{1/2}(1+s/M)(s/(M+s))^\gamma}\right)\\
&=O_P\left(\left( \frac{s^*}{M} \right)^{1/2-\gamma}\right).
\end{align*}
Using the scale transformation of the Wiener process we obtain
$$
\max_{1\leq s\leq s^*}\frac{\displaystyle \left|  W_{M,2}(s)   \right|}{ M^{1/2}(1+s/M)^{1-\gamma}(s/M)^\gamma}\stackrel{{\mathcal D}}{=}
\max_{1/M\leq u\leq s^*/M}\frac{\displaystyle \left|  W(u)   \right|}{ (1+u)^{1-\gamma}u^\gamma},
$$
where $W$ is a Wiener process and
$$
\left(\frac{s^*}{M}\right)^{\gamma-1/2}\max_{1/M\leq u\leq s^*/M}\frac{\displaystyle \left|  W(u)   \right|}{ (1+u)^{1-\gamma}u^\gamma}
\stackrel{{\mathcal D}}{\to}\sup_{0<u\leq 1}\frac{|W(u)|}{u^\gamma}.
$$
Thus we get
$$
\max_{1\leq s\leq s^*}\frac{\displaystyle \left| \sum_{u=M+1}^{M+s}\eps_u  \right|}{M^{1/2}(1+s/M)^{1-\gamma}(s/M)^\gamma}=O_P\left(\left( \frac{s^*}{M} \right)^{1/2-\gamma}\right)
$$
and similar arguments yield
$$
\max_{1\leq s\leq s^*}\frac{\displaystyle \left| \frac{s}{M} \sum_{u=1}^{M}\eps_u  \right|}{M^{1/2}(1+s/M)^{1-\gamma}(s/M)^\gamma}
=O_P\left(\left( \frac{s^*}{M} \right)^{1/2-\gamma}\right).
$$
Using \eqref{cn-2-1} we obtain immediately \eqref{s*m-1}.
\end{proof}

 We can assume without loss of generality the $\Delta=\Delta_M>0$. Let
$$
\cN=\cN(M,x)=\left[ \frac{c\sigma M^{1/2-\gamma}} {\Delta_M}-
x\left( \frac{c^{1/2-\gamma}M^{(1/2-\gamma)^2}}{(\Delta_M/\sigma )^{3/2-2\gamma}}  \right)^{1/(1-\gamma)}  \right]^{1/(1-\gamma)}.
$$

\begin{lemma}\label{cn-1} If Assumptions \ref{destat}, \ref{ta-2} and \eqref{ta-3} hold, then we have for all $x$ that
\beq\label{cn-1-1}
\Delta\cN^{1/2}\to \infty,
\eeq
\beq\label{cn-1-1-1}
\frac{s^*}{\cN}\to 0,
\eeq
\beq\label{cn-1-2}
\left(\frac{\cN}{M}\right)^{1/2-\gamma}\left( c- \frac{\Delta_M\cN}{\sigma M^{1/2}(\cN/M)^\gamma} \right)\to x.
\eeq
\end{lemma}
\begin{proof} We note that
$$
\frac{\cN^{1-\gamma}\Delta}{M^{1/2-\gamma}}\to{c\sigma}
$$
and therefore
\begin{align*}
(\Delta\cN^{1/2})^{-(1-\gamma)}=O\left((\Delta M^{1/2} )^{-(1/2-\gamma)}  \right)=o(1)
\end{align*}
on account of Assumption \ref{ta-2}. \\
The result in   \eqref{cn-1-1-1} is proven in Lemma 3.1 of Aue and Horv\'ath (2004).\\
We claim that
\beq\label{cn-1-1/3}
\left(\frac{M^{(1/2-\gamma)^2}}{\Delta^{3/2-2\gamma}}\right)^{1/(1-\gamma)}\Biggr/\left(\frac{M^{1/2-\gamma}}{\Delta}\right)\to 0.
\eeq
Indeed,
\begin{align}\label{cn-1-1/4}
\frac{M^{(1/2-\gamma)^2}}{\Delta^{3/2-2\gamma}}\frac{\Delta^{1-\gamma}}{M^{(1/2-\gamma)(1-\gamma)}}=\left( \frac{1}{M^{1/2}\Delta } \right)^{1/2-\gamma}\to 0
\end{align}
by Assumption \ref{ta-2}. Since
\begin{align*}
\frac{\Delta_M\cN}{\sigma M^{1/2}(\cN/M)^\gamma}&=
\frac{\Delta}{\sigma}M^{\gamma-1/2}\left(\frac{c\sigma M^{1/2-\gamma}}{\Delta}
-x\left( \frac{c^{1/2-\gamma}M^{(1/2-\gamma)^2}}{(\Delta_M/\sigma )^{3/2-2\gamma}}  \right)^{1/(1-\gamma)}\right)\\
&=c-x\left(\frac{c \sigma }{M^{1/2}\Delta}\right)^{(1/2-\gamma)/(1-\gamma)}
\end{align*}
and by \eqref{cn-1-1/4} we have
\begin{align*}
\lim_{M\to\infty}&\left(\frac{\cN}{M}\right)^{1/2-\gamma}\left(\frac{c \sigma }{M^{1/2}\Delta}\right)^{(1/2-\gamma)/(1-\gamma)}\\
&=\lim_{M\to\infty}\left[\left(\frac{c\sigma M^{1/2-\gamma}} {\Delta_M}\right)^{1/(1-\gamma)}\frac{1}{M}\right]^{1/2-\gamma}\left(\frac{c \sigma }{M^{1/2}\Delta}\right)^{(1/2-\gamma)/(1-\gamma)}\\
&=1,
\end{align*}
so  the proof of \eqref{cn-1-2} is complete.
\end{proof}

\begin{lemma}\label{cn-2} If Assumptions \ref{bern}--\ref{a-adef}, \ref{destat}, \ref{ta-2} and \eqref{ta-3} hold, the we have that
$$
\left(\frac{\cN}{M}\right)^{\gamma-1/2}\left(\max_{1\leq s \leq s^*}\frac{\displaystyle\left|\sum_{u=M+1}^{M+s}\hat{\eps}_u\right|}{\sigma M^{1/2}(1+s/M)(s/(M+s))^\gamma}-
\frac{\Delta_M\cN}{\sigma M^{1/2}(\cN/M)^\gamma}\right)\stackrel{P}{\to}-\infty .
$$
\end{lemma}
\begin{proof}
Putting together Lemma \ref{cn-00} and \eqref{cn-1-1-1}
we get that
\begin{align*}
\left(\frac{\cN}{M}\right)^{\gamma-1/2}\max_{1\leq s \leq s^*}\frac{\displaystyle\left|\sum_{u=M+1}^{M+s}\hat{\eps}_u\right|}{\sigma M^{1/2}(1+s/M)(s/(M+s))^\gamma}=
O_P\left(\left( \frac{s^*}{\cN}  \right)^{1/2-\gamma}\right)=o_P(1).
\end{align*}
Since Lemma \ref{cn-1} imply
$$
\frac{\Delta_M\cN}{\sigma M^{1/2}(\cN/M)^\gamma}\to c
$$
and $\cN/M\to 0$, the result in Lemma \ref{cn-2} is established.
\end{proof}

\begin{lemma}\label{cn-3} If Assumptions \ref{bern}--\ref{a-adef}, \ref{destat}, \ref{ta-2} and \eqref{ta-3} hold, then we have that
$$
\left(\frac{\cN}{M}\right)^{\gamma-1/2}\max_{s^*+1\leq s \leq \cN}\frac{\displaystyle\left|\sum_{u=M+1}^{M+s}\hat{\eps}_u-\left(\sigma W_{M,2}(s)+\Delta_Ms\right)\right|}{ M^{1/2}(1+s/M)(s/(M+s))^\gamma}=o_P(1)
$$
and
$$
\left(\frac{\cN}{M}\right)^{\gamma-1/2}\left|
\max_{s^*+1\leq s \leq \cN}\frac{\displaystyle \left|\sigma W_{M,2}(s)+\Delta_Ms\right|}{ M^{1/2}(1+s/M)(s/M)^\gamma}
-\max_{s^*+1\leq s \leq \cN}\frac{\displaystyle \left| \sigma W_{M,2}(s)+\Delta_Ms\right|}{ M^{1/2}(s/M)^\gamma}\right|=o_P(1).
$$
\end{lemma}
\begin{proof} We note that for $s>s^*$ we have
\begin{align}\label{hepsum}
\sum_{u=M+1}^{M+s}\hat{\eps}_u=\sum_{u=M+1}^{M+s}\eps_u-\left(\sum_{u=M+1}^{M+s}\bx_u  \right)^\T(\hat{\bbe}_M-\bbe_0)
+\left(\sum_{u=M+s^*+1}^{M+s}\bx_u  \right)^\T({\bde}_M-\bbe_0).
\end{align}
It follows from Lemmas \ref{epsapp} and \ref{cn-1} that
\begin{align*}
&\hspace{-1cm}\left(\frac{\cN}{M}\right)^{\gamma-1/2}\max_{s^*+1\leq s \leq \cN}\frac{\displaystyle\left|\sum_{u=M+1}^{M+s}{\eps}_u-\sigma W_{M,2}(u)\right|}
{ M^{1/2}(1+s/M)(s/M)^\gamma}\\
&=O_P\left(\left(\frac{\cN}{M}\right)^{\gamma-1/2}\max_{s^*+1\leq s \leq \cN}\frac{s^{1/2-\delta}}{ M^{1/2}(1+s/M)(s/M)^\gamma}
\right)\\
&=O_P(\cN^{-\delta})=o_P(1).
\end{align*}
 Lemmas \ref{hatga}, \ref{xlaw} and \eqref{xalta} yield
$$
\left|\left(\sum_{u=M+1}^{M+s}\bx_u  \right)^\T(\hat{\bbe}_M-\bbe_0)\right|=O_P(s\Delta)
$$
and therefore
\begin{align*}
\left(\frac{\cN}{M}\right)^{\gamma-1/2}\max_{s^*+1\leq s \leq\cN}&\frac{\displaystyle \left|\left(\sum_{u=M+1}^{M+s}\bx_u  \right)^\T(\hat{\bbe}_M-\bbe_0)\right|}{M^{1/2}(1+s/M)(s/(s+M))^\gamma}  \\
&=O_P\left(\max_{s^*+1\leq s \leq\cN}\frac{\Delta s M^{-1/2}}{M^{1/2}(1+s/M)(s/(s+M))^\gamma}\right)\\
&=o_P(1).
\end{align*}
Similarly to Lemma \ref{xlaw} with some $\zeta>1/2$ we have
\beq\label{xbound-alt}
\sup_{s^*+1\leq s <\infty}\frac{1}{(s-s^*)^\zeta}\left\|\sum_{u=M+s^*+1}^{s+M}{\bx}_u-\bc_A (u-s^*)\right\|=O_P(1)
\eeq
and therefore by Lemma \ref{xlaw} and Assumption \ref{ta-3} we conclude
\begin{align*}
\left(\frac{\cN}{M}\right)^{\gamma-1/2}\max_{s^*+1\leq s \leq\cN}\frac{\left| \left(\displaystyle  \sum_{u=M+s^*+1}^{M+s}\bx_u  \right)^\T({\bde}_M-\bbe_0)-\Delta s\right|  }{M^{1/2}(1+s/M)(s/(s+M))^\gamma}=o_P(1).
\end{align*}
Thus the proof of the first part of lemma \ref{cn-3} is complete.\\
To prove the second part we note that
\begin{align*}
\left(\frac{\cN}{M}\right)^{\gamma-1/2}\max_{s^*+1\leq s \leq\cN}\left| \frac{\displaystyle \sigma W_{M,2}(s)+\Delta_Ms}{ M^{1/2}(s/M)^\gamma}\right|\left|\frac{(s/M)^\gamma}{(1+s/M)(s/(s+M))^\gamma}  -1 \right|=o_P(1)
\end{align*}
(cf.\ the proof of Lemma 3.3 of Aue and Horv\'ath, 2014).
\end{proof}
\begin{lemma}\label{cn-4} If Assumptions \ref{bern}--\ref{a-adef}, \ref{destat}, \ref{ta-2} and \eqref{ta-3} hold, the we have that
\begin{align*}
\lim_{M\to\infty}\Biggm\{&
\left( \frac{\cN}{M} \right)^{\gamma-1/2}
\left(\max_{s^*+1\leq s\leq \cN}  \frac{|\sigma W(s)+\Delta s|}{\sigma  M^{1/2}(s/M)^\gamma}
-\frac{\Delta \cN}{\sigma  M^{1/2}(\cN/M)^\gamma}\right)\\
&\hspace{1cm}\leq \left( \frac{\cN}{M} \right)^{\gamma-1/2}\left(c-  \frac{\Delta \cN}{\sigma  M^{1/2}(\cN/M)^\gamma}  \right)
\Biggm\}=\Phi(x),
\end{align*}
where $\Phi(x)$ denotes the standard normal distribution function.
\end{lemma}
\begin{proof}
By the scale transformation of the Wiener process we have
$$
\max_{1\leq s\leq \cN}\frac{| W(s)|}{ M^{1/2}(s/M)^\gamma}=O_P\left(\left( \frac{\cN}{M} \right)^{-(\gamma-1/2)}\right)=
o_P\left( \frac{\Delta\cN}{M^{1/2}(\cN/M)^\gamma } \right)
 $$
 on account of \eqref{cn-1-1}. Hence
 $$
 P\left\{\max_{s^*+1\leq s\leq \cN}  \frac{|\sigma W(s)+\Delta s|}{\sigma  M^{1/2}(s/M)^\gamma} =
 \max_{s^*\leq s\leq \cN}  \frac{\sigma W(s)+\Delta s}{\sigma  M^{1/2}(s/M)^\gamma}\right\}\to 1.
 $$
 For every $0<\delta<1$
 $$
 \left( \frac{\cN}{M} \right)^{\gamma-1/2}\max_{0<s\leq(1-\delta) \cN}  \frac{ W(s)}{ M^{1/2}(s/M)^\gamma}\stackrel{{\mathcal D}}{=}
 \max_{0<u\leq 1-\delta}\frac{W(u)}{u^\gamma},
 $$
 and thus we have
 $$
 P\left\{\max_{s^*+1\leq s\leq \cN}  \frac{\sigma W(s)+\Delta s}{\sigma  M^{1/2}(s/M)^\gamma}=\max_{(1-\delta)\cN\leq s\leq \cN}  \frac{\sigma W(s)+\Delta s}{\sigma  M^{1/2}(s/M)^\gamma}
 \right\}\to 1.
 $$
 Similarly,
 $$
 \left( \frac{\cN}{M} \right)^{\gamma-1/2}\max_{(1-\delta)\cN\leq s\leq \cN}  \frac{|W(s)-W(\cN)|}{ M^{1/2}(s/M)^\gamma}\stackrel{{\mathcal D}}{=}
 \max_{1-\delta\leq u \leq 1}\frac{|W(u)-W(1)|}{u^\gamma}
 $$
 and by the continuity of $W$ we get
 $$
 \max_{1-\delta\leq u \leq 1}\frac{|W(u)-W(1)|}{u^\gamma}\to 0\;\;\mbox{a.s.}\;\;(\delta\to 0).
 $$
 We showed
 \begin{align*}
&\hspace{-.5cm} \lim_{M\to\infty}\Biggm\{
\left( \frac{\cN}{M} \right)^{\gamma-1/2}
\left(\max_{s^*+1\leq s\leq \cN}  \frac{|\sigma W(s)+\Delta s|}{\sigma  M^{1/2}(s/M)^\gamma}
-\frac{\Delta \cN}{\sigma  M^{1/2}(\cN/M)^\gamma}\right)\\
&\hspace{1cm}\leq \left( \frac{\cN}{M} \right)^{\gamma-1/2}\left(c-  \frac{\Delta \cN}{\sigma  M^{1/2}(\cN/M)^\gamma}  \right)
\Biggm\}\\
=\lim_{M\to\infty}\Biggm\{&
\left( \frac{\cN}{M} \right)^{\gamma-1/2}
 \frac{ W(\cN)}{  M^{1/2}(\cN/M)^\gamma}
\leq \left( \frac{\cN}{M} \right)^{\gamma-1/2}\left(c-  \frac{\Delta \cN}{\sigma  M^{1/2}(\cN/M)^\gamma}  \right)
\Biggm\}.
 \end{align*}
 Since $(\cN/{M} )^{\gamma-1/2} W(\cN)/(  M^{1/2}(\cN/M)^\gamma)$ is a standard normal random variable, the result follows from \eqref{cn-1-2}.
\end{proof}
\begin{lemma}\label{cn-5} If Assumptions \ref{bern}--\ref{a-adef}, \ref{destat}, \ref{ta-2} and \eqref{ta-3} hold, the we have for all $x$ that
$$
\lim_{M\to \infty}P\{\tau_M>\cN(M,x)  \}=\Phi(x),
$$
where $\Phi(x)$ denotes the standard normal distribution function.
\end{lemma}
\begin{proof} By Assumption \ref{bern}, Lemmas \ref{ybern} and \ref{hatga} we have
$$
\hat{\sigma}_M-\sigma=O_P(M^{-1/2})
$$
and
therefore by Lemmas \ref{cn-2}--\ref{cn-4}
$$
\left( \frac{\cN}{M} \right)^{\gamma-1/2}\max_{1\leq s \leq \cN}\frac{\displaystyle \left|\sum_{u=M+1}^{M+s}\hat{\eps}_u \right|}{M^{1/2}(1+s/M)(s/(M+s))^\gamma}
\left|\frac{1}{\hat{\sigma}_M} -\frac{1}{\sigma} \right|=o_P(1).
$$
Hence Lemmas \ref{cn-2}--\ref{cn-4} imply
\begin{align*}
\lim_{M\to \infty}&P\{\tau_M>\cN(M,x)  \}=\lim_{M\to\infty} P\left\{\max_{1\leq s\leq \cN}\frac{\Gamma(M,s)}{g(M,s)}\leq 1
\right\}\\
&=\lim_{M\to\infty}P\Biggm\{ \Biggm( \frac{\cN}{M}\Biggm)^{\gamma-1/2}\Biggm(\max_{s^*<s\leq \cN}\frac{\displaystyle \left|\sum_{u=M+1}^{M+s}\hat{\eps}_u\right|}{\sigma M^{1/2}(s/M)^\gamma}-\frac{\Delta \cN}{\sigma M^{1/2}(\cN/M)^\gamma} \Biggm)\\
&\hspace{4.5cm}\leq \left(\frac{\cN}{M} \right)^{\gamma-1/2}
\left(c-\frac{\Delta \cN}{\sigma M^{1/2}(\cN/M)^\gamma}\right)  \Biggm\}\\
&=\Phi(x).
\end{align*}
\end{proof}

\noindent
{\it Proof of Theorem \ref{th-alt-0}.} Elementary calculus gives
\begin{align*}
\cN(M,x)&=a_M\left[ 1-x\left( \frac{c^{1/2-\gamma}\sigma^{3/2-2\gamma} M^{(1/2-\gamma)^2}}{\Delta^{3/2-2\gamma}}\right)^{1/(1-\gamma)}\frac{1}{a_M^{1-\gamma}}   \right]^{1/(1-\gamma)}\\
&=a_M+a_M\left[-x\frac{1}{1-\gamma}\left( \frac{c^{1/2-\gamma}\sigma^{3/2-2\gamma} M^{(1/2-\gamma)^2}}{\Delta^{3/2-2\gamma}}\right)^{1/(1-\gamma)}\frac{1}{a_M^{1-\gamma}}(1+o(1)   \right]
\end{align*}
and
$$
b_M=\frac{a_M^\gamma}{1-\gamma}\left( \frac{c^{1/2-\gamma}\sigma^{3/2-2\gamma} M^{(1/2-\gamma)^2}}{\Delta^{3/2-2\gamma}}\right)^{1/(1-\gamma)}.
$$
So by Lemma \ref{cn-5} we have for all $x$
\begin{align*}
\lim_{M\to\infty}P\left\{(\tau_M-a_M)/b_M<x\right\}&=1-\lim_{M\to\infty}P\left\{(\tau_M-a_M)/b_M>x\right\}\\
&=1-\lim_{M\to\infty}P\left\{(\tau_M>\cN(M,-x)\right\}\\
&=1-\Phi(-x),
\end{align*}
completing the proof of Theorem \ref{alt-stat}.
\qed

\medskip
Using Assumption \ref{unir} we get that for all $u\geq 1$ that
\beq\label{yunir}
y_{M+s^*+u}=y_{M+s^*}+\sum_{z=1}^u(\bw_{M+s^*+z}^\T\bar{\bde}+\epsilon_{M+s^*+z}).
\eeq
According to \eqref{hepsum} we have for all $s\geq s^*+1$ that
\begin{align*}
\sum_{u=M+1}^{M+s}\hat{\eps}_u&=\sum_{u=M+1}^{M+s}\eps_u-\left(\sum_{u=M+1}^{M+s^*}\bx_u\right)^\T(\hat{\bbe}_M-\bbe_0)
-\left(\sum_{u=M+s^*1}^{M+s}\bx_u\right)^\T(\bde_M-\bbe_0)\\
&=\sum_{u=M+1}^{M+s}\eps_u-\left(\sum_{u=M+1}^{M+s^*}\bx_u\right)^\T(\hat{\bbe}_M-\bbe_0)
-\left(\sum_{u=M+s^*+1}^{M+s}\bw_u\right)^\T(\bar{\bde}-\bar{\bbe}_0)\\
&\hspace*{2cm}-(1-\beta_{0,d})\sum_{u=M+s^*+1}^{M+s}y_{u-1}.
\end{align*}
\begin{lemma}\label{le-alt-9} If Assumptions \ref{bern}--\ref{a-adef}, \ref{unir},  \ref{destat-1}, \eqref{th-alt-20} and  \eqref{th-alt-21} hold, then we have
\begin{align}\label{le-9-2}
\left|\left(\sum_{u=M+1}^{M+s^*}\bx_u\right)^\T(\hat{\bbe}_M-\bbe_0)\right|=O_P(s^*M^{-1/2})
\end{align}
\begin{align}\label{le-9-3}
\max_{1\leq s<\infty}\left|\frac{1}{s}\left(\sum_{u=M+s^*+1}^{M+s}\bw_u\right)^\T(\bar{\bde}-\bar{\bbe}_0)\right|=O_P(1)
\end{align}
and for each $M$ there is a Wiener process $\{W_{M,3}(u), u\geq 0\}$ such that
\begin{align}\label{le-9-4}
\max_{1\leq s <\infty}\frac{1}{s^{3/2-\delta}}\left|\sum_{u=M+s^*+1}^{M+s}(y_{u-1}-s\fa_1)-\fb_1\int_0^s W_{M,3}(u)du\right|=O_P(1)
\end{align}
with some $\delta>0$.
Furthermore, if $\cN^*=\cN^*_M\to \infty$, then we have
\begin{align}\label{le-9-1}
\left(\frac{\cN^*}{M}\right)^{\gamma-1/2}\max_{1\leq s \leq \cN^*}\frac{\displaystyle \left|\sum_{u=M+1}^{M+s}\eps_u\right|}{M^{1/2}(s/M)^\gamma}=O_P(1)
\end{align}
\end{lemma}
\begin{proof}
The upper bound in \eqref{le-9-2} is an immediate consequence of Lemmas \ref{hatga} and \ref{xlaw}.\\
It follows from Aue et al.\ (2014)
$$
\max_{1\leq u<\infty}\frac{1}{u^{1/2-\delta}}\left| (y_{M+s^*+u}-s\fa_1)-\fb_1W_{M,3}(u) \right|=O_P(1)
$$
with some Wiener processes $\{W_{M,3}(u), u\geq 0\}$
which implies  \eqref{le-9-4}. \\
Using Lemma \ref{epsapp} we get
$$
\left(\frac{\cN^*}{M}\right)^{\gamma-1/2}\max_{1\leq s\leq \cN^*}\frac{\displaystyle \left|\sum_{u=M+1}^{M+s}\eps_u -\sigma W_{M,2}(s)\right|}{M^{1/2}(s/M)^\gamma}=o_P(1)
$$
and by the scale transformation of the Wiener process we have
$$
\left(\frac{\cN^*}{M}\right)^{\gamma-1/2}\max_{1\leq s\leq \cN^*}\frac{\displaystyle \left| W_{M,2}(s)\right|}{M^{1/2}(s/M)^\gamma}\stackrel{{\mathcal D}}{\to}
\sup_{0<u\leq 1}\frac{|W(u)|}{u^\gamma},
$$
where $\{W(u), u\geq 0\}$ stands for a Wiener process. Hence \eqref{le-9-1} is proven.
\end{proof}

{\it Proof of Theorem  \ref{th-alt-2}.} Let
$$
\cM=M^{(1-2\gamma)/(3-2\gamma)}.
$$
We note that
$$
P\{ \tau_M>x\cM  \}=P\left\{
\max_{1\leq s \leq x\cM}\Gamma(M, s)/g(M,s)\leq 1  \right\}.
$$
Following  the proof of Theorem \ref{th-null} we get
$$
\max_{1\leq s \leq s^*}\left|\sum_{u=M+1}^{M+s}\hat{\eps}_u\right|=O_P((s^*)^{1/2-\gamma}),
$$
since $s^*/\cM\to 0$. Thus we conclude
\beq\label{a-21}
\max_{1\leq s \leq s^*}\Gamma(M,s)/g(M,s)=o_P(1).
\eeq
Using \eqref{le-9-1} we obtain
$$
\max_{s^*< s\leq x\cM}\frac{\displaystyle \left| \sum_{u=M+1}^{M+s}\eps_u\right|}{g(M,s)}=O_P\left( \left(\frac{\cM}{M}\right)^{1/2-\gamma}  \right)=o_P(1).
$$
Now \eqref{le-9-1} results in
$$
\max_{s^*< s\leq x\cM}\frac{\displaystyle \left|\left(\sum_{u=M+1}^{M+s^*}\bx_u\right)^\T(\hat{\bbe}_M-\bbe_0)\right|    }{g(M,s)}
=O_P\left( \frac{s^*M^{-1/2}}{M^{1/2-\gamma}(s^*)^\gamma}\right)=o_P(1).
$$
Since
$$
\sum_{u=M+s^*+1}^{M+s}y_{u-1}=\sum_{u=M+s^*+1}^{M+s}(y_{u-1}-\fa_1s)+\fa_1s(s+1)/2
$$
we get by \eqref{le-9-4}
\begin{align*}
\max_{s^*< s\leq x \cM}&\frac{\displaystyle \left|\sum_{u=M+s^*+1}^{M+s}y_{u-1}-\left( \fb_1\int_0^s W_{M,3}(u)du+\fa_1s(s+1)/2  \right)   \right| }{M^{1/2}(s/M)^\gamma}\\
&=O_P\left(\max_{s^*< s\leq x \cM}\frac{\displaystyle s^{3/2-\gamma}}{M^{1/2}(s/M)^\gamma}  \right) \\
&=O_P\left( \frac{\cM^{3/2-2\gamma}}{M^{1/2-\gamma}} \right)\\
&=O_P\left(\left( \frac{\cM}{M}\right)^{1/2-\gamma}\right)=o_P(1).
\end{align*}
Since the distribution of $W_{M,3}$ does not depend on $M$ we note that
\begin{align*}
\max_{s^*< s\leq x\cM}&\frac{\displaystyle \left|\fb_1\int_0^s W_{M,3}(u)du+\fa_1s(s+1)/2  \right|}{M^{1/2}(s/M)^\gamma}\\
&=\max_{s^*/\cM< z\leq x}\frac{\displaystyle \left|\fb_1\int_0^{z\cM} W_{M,3}(u)du+\fa_1z\cM(z\cM+1)/2  \right|}{M^{1/2-\gamma}(u\cM)^\gamma}\\
&\stackrel{{\mathcal D}}{=}\frac{\cM^{3/2-\gamma}}{M^{1/2-\gamma}}\max_{s^*/\cM< z\leq x}\frac{\displaystyle \left|\fb_1\int_0^{z} W(u)du+\fa_1z\cM^{1/2}(z+1/\cM)/2  \right|}{z^\gamma}\\
&\stackrel{{\mathcal D}}{\to}\max_{0< z\leq x}\frac{\displaystyle \left|\fb_1\int_0^{z} W(u)du+\bar{\fa}_1z^2/2  \right|}{z^\gamma},
\end{align*}
concluding the proof of Theorem \ref{th-alt-2}.
\qed

\medskip
\begin{lemma}\label{w-int} If $\{W(u), u\geq 0\}$ is  a Wiener process, then $\int_0^x W(u)du$  is normally distributed with zero mean and
$$
\mbox{{\rm var}}\left(\int_0^x W(u)du   \right)=\frac{x^3}{3}.
$$
\end{lemma}
\begin{proof} Since $W$ is Gaussian, the normality of the integral is clear. Direct calculations give the value of the variance.
\end{proof}
Let
\begin{align*}
\cP=\cP(M,x)=\left[ \frac{c\sigma}{\fc_1{\fa}_1}M^{1/2-\gamma}-x\cP_1M^{1/2-\gamma}  \right]^{1/(2-\gamma)}
\end{align*}
with
$$
\cP_1=\cP_1(M)=c^{(3-2\gamma)/(4-2\gamma)}\sigma^{(7-4\gamma)/(4-2\gamma)}\fc_1^{-(7-4\gamma)/(4-2\gamma)}\fa_1^{-(7-4\gamma)/(4-2\gamma)}
M^{-(1/2-\gamma)/(4-2\gamma)}
$$
with
$$
\fc_1=(1-\beta_{0,d})/2.
$$
We can assume without loss of generality that
$$
\fa_1>0.
$$
\begin{lemma}\label{mprpop} If Assumptions  \ref{unir}, \ref{destat-1}, \eqref{th-alt-31} and \eqref{th-alt-23} hold, the we have that
\beq\label{prop-1}
\cP \left[\frac{\fa_1}{\displaystyle M^{1/2-\gamma}}\right]^{1/(2-\gamma)}\;\;\to\;\;\left( \frac{c\sigma}{\fc_1}\right)^{1/(2-\gamma)}.
\eeq
\beq\label{prop-2}
\frac{s^*}{\cP}\to 0
\eeq
\beq\label{prop-3}
\frac{M^{1/2-\gamma}}{\cP^{3/2-\gamma}}\left(c-\frac{\fa_1\fc_1\cP^2  }{\sigma M^{1/2-\gamma}\cP^\gamma}\right)\;\;\to\;\;x
\eeq
for all $x$.
\end{lemma}
\begin{proof} Since by condition \eqref{th-alt-23}
$$
{\cP_1}{\fa_1}\to 0,
$$
we have immediately \eqref{prop-1}.\\
The assumption in \eqref{th-alt-31} and \eqref{prop-1} yield \eqref{prop-2}.\\
By the definition of $\cP$ we have
\begin{align*}
&\frac{M^{1/2-\gamma}}{\cP^{3/2-\gamma}}\left(c-\frac{\fa_1\fc_1\cP^2  }{\sigma M^{1/2-\gamma}\cP^\gamma}\right)\\
&=\frac{M^{1/2-\gamma}}{\cP^{3/2-\gamma}}\left(c-\frac{\fa_1\fc_1 }{\sigma M^{1/2-\gamma}}\left[ \frac{c\sigma}{\fc_1{\fa}_1}M^{1/2-\gamma}-x\cP_1M^{1/2-\gamma}  \right]\right)\\
&=x\frac{M^{1/2-\gamma}}{\cP^{3/2-\gamma}}\frac{\fa_1\fc_1 }{\sigma M^{1/2-\gamma}}\cP_1M^{1/2-\gamma},
\end{align*}
so the result in \eqref{prop-3} follows from the definition of $\cP_1$ and \eqref{prop-1}.
\end{proof}
\begin{lemma}\label{minf} If Assumptions \ref{bern}--\ref{a-adef}, \ref{unir}, \ref{destat-1}, \eqref{th-alt-31} and \eqref{th-alt-23} hold, the we have that
\begin{align*}
\frac{ M^{1/2-\gamma}}{\cP^{3/2-\gamma}}\left(\max_{1\leq s \leq s^*}
\frac{|\Gamma(M,s)|}{g(M,s)}-\frac{\fc_1\fa_1\cP^2}{\sigma M^{1/2}(\cP/M)^\gamma}
\right)\;\stackrel{P}{\to}\;-\infty.
\end{align*}
\end{lemma}
\begin{proof} Using  \eqref{sicos}, Lemma \ref{cn-00} and \eqref{prop-1} we conclude
\begin{align*}
\frac{ M^{1/2-\gamma}}{\cP^{3/2-\gamma}}\max_{1\leq s \leq s*}\frac{|\Gamma(M,s)|}{g(M,s)}
=O_P\left( \frac{M^{1/2-\gamma}}{\cP^{3/2-\gamma}}\frac{(s^*)^{1/2-\gamma}}{M^{1/2-\gamma}}\right)
=O_P\left(\left( \frac{s^* }{M^{(3-2\gamma)/(4-2\gamma)}}  \right)^{1/2-\gamma}\right)=o_P(1)
\end{align*}
on account of \eqref{th-alt-31}. There is  $C>0$ such that
$$
\frac{ M^{1/2-\gamma}}{\cP^{3/2-\gamma}}\geq C M^{1/2-\gamma}\fa_1^{(3-2\gamma)/(4-2\gamma)}M^{-(3-2\gamma)(1/2-\gamma)/(4-2\gamma)}
\geq C \left(\fa_1^{3-2\gamma} M^{1/2-\gamma}\right)^{1/(4-2\gamma)}\to \infty
$$
by \eqref{th-alt-23}. Also, using again \eqref{prop-1} we obtain
$$
\frac{\fc_1\fa_1\cP^2}{\sigma M^{1/2}(\cP/M)^\gamma}\to c,
$$
completing the proof of the lemma.
\end{proof}

\begin{lemma}\label{root-1} If Assumptions \ref{bern}--\ref{a-adef}, \ref{unir}, \ref{destat-1}, \eqref{th-alt-31} and \eqref{th-alt-23} hold, the we have that
$$
\frac{ M^{1/2-\gamma}}{\cP^{3/2-\gamma}}\max_{s^*< s\leq \cP}\frac{\displaystyle \left|\sum_{u=M+1}^{M+s}\eps_u\right|}{M^{1/2}(s/M)^\gamma}=o_P(1).
$$
\end{lemma}
\begin{proof} It follows from Lemma \ref{epsapp}
$$
\max_{s^*< s\leq \cP}\frac{1}{s^\gamma}\left|\sum_{u=M+1}^{M+s}\eps_u\right|=O_P\left(\cP^{1/2-\gamma}\right)
$$
and therefore
\begin{align*}
\frac{ M^{1/2-\gamma}}{\cP^{3/2-\gamma}}\max_{s^*< s\leq \cP}\frac{\displaystyle \left|\sum_{u=M+1}^{M+s}\eps_u\right|}{M^{1/2}(s/M)^\gamma}
=O_P\left( \frac{1}{ \cP} \right)=o_P(1),
\end{align*}
since $\cP\to \infty$.
\end{proof}
\begin{lemma}\label{root-2} If Assumptions \ref{bern}--\ref{a-adef}, \ref{unir}, \ref{destat-1}, \eqref{th-alt-31} and \eqref{th-alt-23} hold, the we have that
\beq\label{rr-1}
\frac{ M^{1/2-\gamma}}{\cP^{3/2-\gamma}}\max_{s^*< s\leq \cP}\frac{\displaystyle \left|\left(\sum_{u=M+1}^{M+s^*}\bx_u\right)^\T(\hat{\bbe}_M-\bbe_0)\right|}
{M^{1/2}(s/M)^\gamma}=o_P(1)
\eeq
and
\beq\label{rr-2}
\frac{ M^{1/2-\gamma}}{\cP^{3/2-\gamma}}\max_{s^*< s\leq \cP}\frac{\displaystyle \left|\left(\sum_{u=M+s^*+1}^{M+s}\bw_u\right)^\T({\bde}-\bar{\bbe}_0)\right|}
{M^{1/2}(s/M)^\gamma}=o_P(1)
\eeq
\end{lemma}
\begin{proof} Lemmas \ref{hatga} and \ref{xlaw} imply
\begin{align*}
\frac{ M^{1/2-\gamma}}{\cP^{3/2-\gamma}}&\max_{s^*< s\leq \cP}\frac{\displaystyle \left|\left(\sum_{u=M+1}^{M+s^*}\bx_u\right)^\T(\hat{\bbe}_M-\bbe_0)\right|}
{M^{1/2}(s/M)^\gamma}\\
&=O_P\left(\cP^{-(3/2-\gamma)}M^{-1/2}(s^*)^{1-\gamma}\right)\\
&=O_P\left( \fa_1^{(1-2\gamma)/(4-2\gamma)}M^{-(1-2\gamma)^2/(4-2\gamma)}\right)\\
&=o_P(1)
\end{align*}
on account of condition \eqref{th-alt-23}. \\
Similarly to the proof of \eqref{rr-1} we have
\begin{align*}
\frac{ M^{1/2-\gamma}}{\cP^{3/2-\gamma}}\max_{s^*< s\leq \cP}\frac{\displaystyle \left|\left(\sum_{u=M+s^*+1}^{M+s}\bw_u\right)^\T({\bde}-\bar{\bbe}_0)\right|}
{M^{1/2}(s/M)^\gamma}=O_P\left(\frac{\fa_1}{\cP}\right)=o_P(1).
\end{align*}
\end{proof}

\begin{lemma}\label{root-3} If Assumptions \ref{bern}--\ref{a-adef}, \ref{unir}, \ref{destat-1}, \eqref{th-alt-31} and \eqref{th-alt-23} hold, the we have that
\begin{align}\label{ro-3-1}
\frac{M^{1/2-\gamma}}{\cP^{3/2-\gamma}}\max_{s^*<s\leq \cP}
\frac{\displaystyle  \left|(1-\beta_{0,d})\sum_{u=M+s^*+1}^{M+s}y_{u-1}-\left(\fd_1\int_0^s W_{M,3}(u)du+\fa_1\fc_1s^2\right)\right|}
{M^{1/2}(1+s/M)(s/(M+s))^\gamma}=o_P\left(  1 \right)
\end{align}
and
\begin{align}\label{ro-3-2}
&\frac{M^{1/2-\gamma}}{\cP^{3/2-\gamma}}\max_{s^*<s\leq \cP}
{\displaystyle  \left|\fd_1\int_0^s W_{M,3}(u)du+\fa_1\fc_1s^2 \right|}\\
&\hspace{1cm}\times\left|
\frac{1}{\hat{\sigma}_MM^{1/2}(1+s/M)(s/(M+s))^\gamma}
-\frac{1}{\sigma M^{1/2}(s/M)^\gamma}\right|
=o_P\left( 1  \right),\notag
\end{align}
where $\{W_{M,3}(u), u\geq 0\}$ is defined in Lemma \ref{le-alt-9}.
\end{lemma}
\begin{proof} According to Lemmas \ref{minf}--\ref{root-2} we need to show only
\begin{align*}
\max_{s^*<s\leq \cP}
\frac{\displaystyle (1-\beta_{0,d}) \left[\sum_{u=M+s^*+1}^{M+s}y_{u-1}-\left(\fb_1\int_0^s W_{M,3}(u)du+\fa_1s^2/2 \right)\right]}
{M^{1/2}(1+s/M)(s/(M+s))^\gamma}=o_P\left( \frac{\cP^{3/2-\gamma}}{M^{1/2-\gamma}}  \right).
\end{align*}
It follows from Lemma \ref{le-alt-9} that
\begin{align*}
\max_{s^*<s\leq \cP}
\frac{\displaystyle  \left|\sum_{u=M+s^*+1}^{M+s}(y_{u-1}-Ey_{u-1})-\fb_1\int_0^s W_{M,3}\right|}
{M^{1/2}(1+s/M)(s/(M+s))^\gamma}=O_P\left( \frac{\cP}{M^{1/2-\gamma}}  \right)
\end{align*}
and
$$
\max_{s^*<s\leq \cP}\frac{\displaystyle\left|\sum_{u=M+s^*+1}^{M+s}Ey_{u-1}-\fa_1 s^2/2\right|}
{M^{1/2}(1+s/M)(s/(M+s))^\gamma}O_P\left( \frac{\cP}{M^{1/2-\gamma}}  \right).
$$
Thus \eqref{ro-3-1} is proven.  Following  the proof of Lemma \ref{cn-3}, one can verify \eqref{ro-3-2}.

\end{proof}

\begin{lemma}\label{root-5} If Assumptions \ref{bern}--\ref{a-adef}, \ref{unir}, \ref{destat-1}, \eqref{th-alt-31} and \eqref{th-alt-23} hold, the we have that
$$
\lim_{M\to\infty}P\{\tau_M>\cP(M,x)\}=\Phi\left(\frac{x\sigma\sqrt{3}}{\fb_1(1-\beta_{0,d})}\right)
$$
for all $x$.
\end{lemma}
\begin{proof} By definition,
\begin{align}\label{r-5-1}
\lim_{M\to\infty}P\left\{\tau_M>\cP(M,x)\right\}=\lim_{M\to\infty}P\left\{\max_{1\leq s\leq \cP}\Gamma(M,s)/g(M,s)\leq 1\right\}.
\end{align}
Let $\{W(u), u\geq 0\}$ be a Wiener process. Let $\fd_1=(1-\beta_{0,d})\fb_1$. Putting together Lemmas \ref{mprpop}--\ref{root-3}, and \eqref{r-5-1} we need to show only that
\begin{align}\label{r-5-2}
P\left\{ \max_{s^*+1\leq s\leq \cP}\frac{\displaystyle \left| \fd_1\int_0^s W_{M,3}(u)du+\fa_1\fc_1s^2  \right|}{\sigma M^{1/2}(s/M)^\gamma}\leq c  \right\}
=\Phi\left( \frac{x\sigma\sqrt{3}}{\fb_1(1-\beta_{0,d})}\right).
\end{align}
 Let $\{W(u), u\geq 0\}$ be a Wiener process. We note that $W$ and $W_{M,3}$ are equal in distribution. By the scale transformation of the Wiener process we have
$$
{\cP^{-(3/2-\gamma)}}\max_{s^*+1\leq s\leq \cP}\frac{\displaystyle \left| \fd_1\int_0^s W_{}(u)du\right|}{s^\gamma}\stackrel{{\mathcal D}}{=}
\max_{(s^*+1)/\cP\leq s\leq 1}\frac{\displaystyle \left| \fd_1\int_0^s W_{}(u)du\right|}{s^\gamma}.
$$
Hence
\begin{align*}
\lim_{M\to \infty}P\left\{\max_{s^*+1\leq s\leq \cP}\frac{\displaystyle \left| \fd_1\int_0^s W_{}(u)du+\fa_1\fc_1s^2  \right|}{s^\gamma}
=\frac{\displaystyle  \fd_1\int_0^\cP W_{}(u)du+\fa_1\fc_1\cP^2  }{\cP^\gamma}
\right\}=1.
\end{align*}
We note that
\begin{align*}
P&\left\{\frac{\displaystyle  \fd_1\int_0^\cP W_{}(u)du+\fa_1\fc_1\cP^2  }{\sigma M^{1/2-\gamma}\cP^\gamma}\leq c
\right\}\\
&=
P\left\{\frac{M^{1/2-\gamma}}{\cP^{3/2-\gamma}}\frac{\displaystyle  \fd_1\int_0^\cP W_{}(u)du }{\sigma M^{1/2-\gamma}\cP^\gamma}\leq
\frac{M^{1/2-\gamma}}{\cP^{3/2-\gamma}}\left(c-\frac{\fa_1\fc_1\cP^2  }{\sigma M^{1/2-\gamma}\cP^\gamma}\right)
\right\}.
\end{align*}
It follows from Lemma \ref{w-int} that
$$
\frac{M^{1/2-\gamma}}{\cP^{3/2-\gamma}}\frac{\displaystyle  \fd_1\int_0^\cP W_{}(u)du }{\sigma M^{1/2-\gamma}\cP^\gamma}\;\;\stackrel{{\mathcal D}}{=}\;\;
\frac{\fd_1}{\sigma \sqrt{3}}N,
$$
where $N$ stands for a standard normal random variable. Using now \eqref{prop-3}, the lemma is proven.
\end{proof}

\medskip
\noindent
{\it Proof of Theorem \ref{th-alt-3}}. Observing that
\begin{align*}
\cP(M, x)=\left(\frac{c\sigma}{\fc_a\fa_1}M^{1/2-\gamma}\right)^{1/(2-\gamma)}-
\frac{x}{2-\gamma}\left(\frac{c\sigma}{\fc_a\fa_1}M^{1/2-\gamma}\right)^{-(1-\gamma)/(2-\gamma)}\cP_1M^{1/2-\gamma}(1+o(1)),
\end{align*}
the result follows from Lemma \ref{root-5}.
\qed

\begin{lemma}\label{exle-1} If Assumptions \ref{bern}--\ref{a-adef}, \ref{destat} and \ref{expo-1} hold, then
for every $M$ and $s^*$
$$
\bar{\delta}_d^{-u}y_{M+s^*+u}\to Z_{M+s^*}= y_{M+s^*}+\sum_{z=1}^\infty \bar{\delta}_d^{-z } (\bw_{M+s^*+z}^\T(\bar{\bde}-\bar{\bbe}_0)+\eps_{M+s^*+z})\;\;\mbox{a.s.}\;\;
$$
as $u\to \infty$.
\end{lemma}
\begin{proof} Since
$$
y_{M+s^*+u}=\bar{\delta}_d^{u}y_{M+s^*}+ +\sum_{z=1}^{u} \bar{\delta}_d^{z-u } (\bw_{M+s^*+z}^\T(\bar{\bde}-\bar{\bbe}_0)+\eps_{M+s^*+z})
$$
the result follows immediately from Assumption \ref{expo-1} and the mean stationarity of $\bw_z$ and $\eps_z$.
\end{proof}

\medskip
\noindent
{\it Proof of Theorem \ref{expo}.} It follows from Lemma \ref{cn-00} that
\begin{equation*}
\max_{1\leq s \leq s^*}\frac{\displaystyle \left|\sum_{u=1}^s\hat{\eps}_u\right|}{\hat{\sigma}_M(1+s/M)M^{1/2}(s/(M+s))^\gamma}=o_P(1).
\end{equation*}
Using \eqref{hepsum} we get for all $s>s^*$ that
\begin{align*}
\sum_{u=M+1}^{M+s}\hat{\eps}_u&=\sum_{u=M+1}^{M+s}\eps_u-\left( \sum_{u=M+1}^{M+s^*}\bx_u  \right)^\T\left(\hat{\bbe}_M-\bbe_0\right)-\left(\sum_{u=M+s^*+1}^{M+s}\bw_u\right)^\T(\bar{\bde}-\bar{\bbe}_0)
\\
&\hspace{2cm}-(\bar{\delta}_d-\beta_{0,d})\sum_{u=M+s^*+1}^{M+s}y_{u-1}.
\end{align*}
Let
$$
\cQ=\cQ(M)=s^*+x+[(1/2-\gamma)\log\bar{\delta}_d]\log M+[\gamma\log\bar{\delta}_d]\log \log M.
$$
Next we note
\begin{align*}
\max_{s^*+1\leq s \leq \cQ}\frac{\displaystyle \left|\sum_{u=M+1}^{M+s}\eps_u\right|}{(1+s/M)M^{1/2}(s/(M+s))^\gamma}=O_P\left(\left( \frac{\cQ}{M}\right)^{1/2-\gamma}\right)=o_P(1).
\end{align*}
Using Lemmas  \ref{hatga}  and \ref{xlaw}    we conclude
\begin{align*}
\max_{s^*+1\leq s \leq \cQ}\frac{\displaystyle\left|\left( \sum_{u=M+1}^{M+s^*}\bx_u  \right)^\T\left(\hat{\bbe}_M-\bbe_0\right)\right|}
{(1+s/M)M^{1/2}(s/(M+s))^\gamma}&=O_P\left( \max_{s^*+1\leq s \leq \cQ}\frac{sM^{-1/2}}{(1+s/M)M^{1/2}(s/(M+s))^\gamma}   \right)\\
&=O_P\left(\left(\frac{\cQ}{M}\right)^{1-\gamma}\right)=o_P(1)
\end{align*}
and
\begin{align*}
\max_{s^*+1\leq s \leq \cQ}\frac{\displaystyle\left|\left( \sum_{u=M+s^*+1}^{M+s^*}\bw_u  \right)^\T\left(\bar{\bde}_M-\bar{\bbe}_0\right)\right|}
{(1+s/M)M^{1/2}(s/(M+s))^\gamma}&=O_P\left( \max_{s^*+1\leq s \leq x\cQ}\frac{s}{(1+s/M)M^{1/2}(s/(M+s))^\gamma}   \right)\\
&=O_P\left(\frac{\cQ^{1-\gamma}}{M^{1/2-\gamma}}\right)=o_P(1).
\end{align*}
Using Lemma \ref{exle-1} we get
$$
\bar{\delta}_d^{-(s-s^*)}\sum_{u=M+s^*+1}^{M+s}y_{u-1}\to \frac{1}{\bar{\delta}_d-1}Z_{M+s^*}
$$
as $s-s^*\to \infty$. Hence
\begin{align*}
\max_{s^*+1\leq s \leq \cQ}\frac{\displaystyle \left|\sum_{u=M+s^*+1}^{M+s}y_{u-1}-\bar{\delta}_d^{s-s^*}\frac{1}{\bar{\delta}_d-1}Z_{M+s^*}\right|}{(1+s/M)M^{1/2}(s/(M+s))^\gamma}=o_P(1)
\end{align*}
and
\begin{align*}
P\left\{\max_{s^*+1\leq s \leq \cQ}\frac{\displaystyle \left|\bar{\delta}_d^{s-s^*}\frac{1}{\bar{\delta}_d-1}Z_{M+s^*}\right|}{(1+s/M)M^{1/2}(s/(M+s))^\gamma}
=\frac{\displaystyle \left|\bar{\delta}_d^{\cQ-s^*}\frac{1}{\bar{\delta}_d-1}Z_{M+s^*}\right|}{(1+\cQ/M)M^{1/2}(\cQ/(M+\cQ))^\gamma}\right\}\to 1,
\end{align*}
as $M\to\infty$. We observe that by Assumption \ref{expo-1}, \eqref{expos} and \eqref{sicos}
$$
\frac{\displaystyle \left|\bar{\delta}_d^{\cQ-s^*}\right|}{\hat{\sigma}_M(1+\cQ/M)M^{1/2}(\cQ/(M+\cQ))^\gamma}\stackrel{P}{\to}\frac{|\bar{\delta}_d|^x}{
\sigma((1/2-\gamma)\log \bar{\delta}_d)^\gamma}.
$$
Since
$$
P\{\tau_M>\cQ\}=P\{\max_{1\leq s\leq \cQ}\Gamma(M,s)/g(M,s)<1\},
$$
the proof of Theorem \ref{expo} is complete.
\qed
\medskip


\begin{thebibliography}{99}
\bibitem{sinica}  Aue, A.,\   H\"ormann, S., Horv\'ath, L.\  and   Hu\v{s}kov\'{a}, M.: Dependent functional linear models with applications to
monitoring structural change. {\it   Statistica Sinica} {\bf 24}(2014), 1043--1073.
\bibitem{ah} Aue, A.\ and Horv\'ath, L.: Delay time in sequential detection of change. {\it Statistics \& Probability Letters} {\bf 67}(2004), 221--231.
\bibitem{ahhk} Aue,A.,\  Horv\'ath, L.,\   Hu\v{s}kov\'a and   Kokoszka, P.: Change‐point monitoring in linear models. {\it The Econometrics Journal}  {\bf 9}(2006), 373–-403.
\bibitem{ahks}Aue, A.,\  Horv\'ath, L.,\  Kokoszka, P.\  and  Steinebach, J.:  Monitoring shifts in mean: asymptotic normality of stopping times. {\it  TEST } {\bf 17}(2008) 515--530.
\bibitem{billingsley} Billingsley, P.:{\it  Convergence of Probability Measures}, Wiley, New York, 1968.
\bibitem{burn} Burnside, C.,\ Eichenbaum, M.\ and  Rebelo, S.:  Understanding booms and busts in
housing markets. {\it Journal of Political Economy} {\bf 124}(2016), 1088--1147.
\bibitem{caron} Caron, E.: Asymptotic distribution of least square estimators for linear models with dependent errors. {\it Statistics} {\bf 53}(2019), 885--902.
\bibitem{cs} Case, K.E.\ and  Shiller, R.J.:  The efficiency of the market for single--family homes. {\it American Economic Review} {\bf 79}(1989), 125--137.
\bibitem{cs-2} Case, K.E.\ and  Shiller, R.J.:  Is there a bubble in the housing market? {\it Brookings
Papers on Economic Activity},  2003, No.\ 2, 299--362.
\bibitem{chu} Chu, C.-S.J.,\ Stinchcombe, M.\ and  White, H.:  Monitoring structural change. {\it Econometrica } {\bf 64}(1996), 1045–-1065.
\bibitem{dav} Davis, M.A.\ and  Heathcote, J.:  The price and quantity of residential land in the
united states. {\it Journal of Monetary Economics} {\bf 54}(2007), 2595--2620.
\bibitem{du}  D\"umbgen, L.:  The asymptotic behavior of some nonparametric change--point estimators. {\it
Annals of Statistics} {\bf 19}(1991),  1471–-1495.
\bibitem{ga} Gallin, J.:  The long-run relationship between house prices and income: evidence from
local housing markets. {\it Real Estate Economics} {\bf 34}(2006), 417--438.
\bibitem{gla} Glaeser, E. L.\ and  Nathanson, C. G.: An extrapolative model of house price dynamics.
{\it Journal of Financial Economics} {\bf 126}(2017), 147--170.
\bibitem{gra} Granziera, E.\ and   Kozicki, S.:  House price dynamics: fundamentals and expectations.
{\it Journal of Economic Dynamics and Control} {\bf 60}(2015), 152--165.
\bibitem{gyu} Gyourko, J.,\  Mayer, C.\ and  Sinai, T.:  Superstar cities. {\it American Economic Journal:
Economic Policy} {\bf 5}(2013), 167--99.
\bibitem{him} Himmelberg, C.,\  Mayer, C.\ and   Sinai, T.:  Assessing high house prices: Bubbles,
fundamentals and misperceptions. {\it Journal of Economic Perspectives}  {\bf 19}(2005), 67--92.
\bibitem{hhkm} Hl\'avka, Z.,\   Hušková, M.,\ Kirch, C.\ and   Meintanis, S.:  Monitoring changes in the error distribution of autoregressive models based on Fourier methods. {\it TEST}  {\bf 21}(2012), 605--634.
\bibitem{hoga} Hoga, Y.: Monitoring multivariate time series. {\it Journal of Multivariate Analysis} {\bf 155}(2017), 105--121.
\bibitem{hb} Homm, U.\ and Breitung, J.: Testing speculative bubbles in stock markets: a comparison of alternative methods. {\it Journal of Financial Econometrics} {\bf 10}(2012), 198--231.
\bibitem{hk} H\"ormann, S.\ and Kokoszka, P.: Weakly dependent functional data. {\it Annals of Statistics} {\bf 3}(2010), 1845--1884.
\bibitem{hhks}  Horv\'ath, L.,\ Hu\v{s}kov\'a, M.,\  Kokoszka, P.\ and Steinebach, J.: Monitoring changes in linear models. {\it Journal of  Statistical Planning and  Inference} {\bf 116}(2004), 225--251.
\bibitem{Hks} Horv\'ath, L.,\ Kokoszka, P.\ and Steinebach, J.: On sequential detection of parameter changes in linear regression. {\it  Statistics \& Probability Letters} {\bf 77}(2007) 885--895.
\bibitem{hlrw} Horv\'ath, L.,\ Liu, Z.,\ Rice, G.\ and Wang, S.: Sequential monitoring for changes from stationarity to mild non--stationarity. {\it Journal of Econometrics} To appear (2019+).
\bibitem{huki}  Hu\v{s}kov\'a, M.\  and  Kirch, C.: Bootstrapping sequential change--point tests for linear regression. {\it Metrika} {\bf 75}(2012), 673--708.
\bibitem{ib-1} Ibragimov, I.A.: Some limit theorems for strict--sense stationary  stochastic processes (in Russian). {\it Doklady Akademii Nauk SSSR }  {\bf 125}(1959), 711--714.
\bibitem{ib-2} Ibragimov,   I.A.:  Some limit theorems for stationary  processes. {\it Theory of Probability  and Its Applications} {\bf  7}(1962),  349--382.
\bibitem{kir-1} Kirch, C.: Block permutation principles for the change analysis of dependent data. {\it Journal of  Statistical Planning and  Inference} {\bf 137}(2007), 2453--2474.
\bibitem{kir-2} Kirch, C.: Bootstrapping sequential change--point tests. {\it Sequential  Analysis} {\bf 27}(2008), 330--349.
\bibitem{kkps} Kwiatkowski, D.,\ Phillips, P.C.B.,\ Schmidt, P.\ and Shin, Y.:  Testing the null hypothesis
of stationarity against the alternative of a unit root: how sure are we that economic time series
have a unit root? {\it Journal of Econometrics} {\bf 54}(1992), 159--178.
\bibitem{lp} Lee, J.H.\ and Phillips, P.C.B.: Asset pricing with financial bubble risk. {\it Journal of Empirical Finance} {\bf 38}(2016), 590-622.
\bibitem{lin} Linton, O.: {\it Financial Econometrics: Models and Methods.} Cambridge University Press, 2019.
\bibitem{ma} Mayer, C.: Housing bubbles: A survey.  {\it  Annual Review of Economics} {\bf 3}(2011), 559--577.
\bibitem{psy0} Phillips, P.C.B.,\  Shi, S.\ and Yu. J.: Specification sensitivity in right--tailed unit
root testing for explosive behaviour.  {\it  Oxford Bulletin of Economics and Statistics}
{\bf 76}(2014), 15--333.
\bibitem{psy1} Phillips, P.C.B.,\ Shi, S.\ and Yu, J.:Testing for multiple bubbles: historical episodes of exuberance and collepse. in the S\&P 500. {\it International Economic Review} {\bf 56}(2015a), 1043--1177.
\bibitem{psy2} Phillips, P.C.B.,\ Shi, S.\ and Yu, J.: Testing for multiple bubbles: limit theory of real--time detectors. {\it International Economic Review} {\bf 56}(2015b), 1079--1133.
\bibitem{py} Phillips,  P.C.B.\ and Yu, J: Dating the timeline of financial bubbles during the
subprime crisis. {\it  Quantitative Economics} {\bf 2}(2011), 455--491.
\bibitem{pia} Piazzesi, M.\ and Schneider, M.:  Momentum traders in the housing market: survey
evidence and a search model. {\it American Economic Review} {\bf 99}(2009), 406--411.
\bibitem{pic} Picard, D.: Testing and estimating change-points in time series. {\it Advances in Applied Probability} {\bf 17}(1985), 841--867.
\bibitem{saiz} Saiz, A.:  The geographic determinants of housing supply. {\it The Quarterly Journal of
Economics} 125(2010), 1253--1296.
\bibitem{shi} Shiller, R. J.:  Historic turning points in real estate. {\it Eastern Economic Journal}
{\bf 34}(2008), 1--13.
\bibitem{shi2} Shiller, R. J.:  {\it Irrational Exuberance.}  (Revised and expanded third edition),  Princeton
University Press, 2015.
\bibitem{wu} Wu, W.B.: $M$--estimates of linear models with dependent errors. {\it Annals of Statistics} {\bf 35}(2007), 495--521.
\bibitem{zuth} Zeckerhauser, R.\ and Thompson, M.: Linear regression with non--normal error terms. {\it Review of Economics and Statistics} {\bf 52}(1970), 280--286.
\bibitem{zlkh}  Zeileis,A.,\    Leisch, F.,\   Kleiber, C.\ and    Hornik, K.: Monitoring structural change in dynamic econometric models. {\it Journal of Applied Econometrics} {\bf 20}(2005), 99--121.
\bibitem{zsh} Zheng, S.,\ Sun, W.\  and  Kahn, M. E.:  Investor confidence as a determinant of China's
urban housing market dynamics. {\it Real Estate Economics} {\bf 44}(2016), 814--845.
\bibitem{zhu} Zhu, Z.: Inference for linear models with dependent errors. {\it Journal of the Royal Statistical Society Ser.\ B} {\bf 75}(2013), 1--21.


\end{thebibliography}
 \end{document}